\documentclass{scrartcl}

\usepackage{authblk}

\usepackage[utf8]{inputenc}

\usepackage{negotiations}

\usepackage{amssymb}
\usepackage{eucal}
\usepackage{graphicx}
\usepackage{amsmath}
\usepackage{xcolor}
\usepackage{stmaryrd}
\usepackage{rotating}
\usepackage{xspace}
\usepackage{eepic}
\usepackage{epsfig}
\usepackage{latexsym} 
\usepackage{exscale}
\usepackage{verbatim}

\usepackage{caption}
\usepackage{subcaption}

\usepackage{amsthm}

\newtheorem{theorem}{Theorem}
\newtheorem{lemma}[theorem]{Lemma}
\newtheorem{corollary}[theorem]{Corollary}
\newtheorem{proposition}[theorem]{Proposition}
\newtheorem{definition}[theorem]{Definition}

\newtheorem*{remark*}{Remark}
\newtheorem{fact}{Fact}

\makeatletter
\newenvironment{proof*}[1][\proofname]{\par
  \normalfont 
  \partopsep=\z@skip \topsep=\z@skip
  \penalty100
  \trivlist
  \item[\hskip\labelsep
        \itshape
    #1\@addpunct{.}]\ignorespaces
}{%
  \endtrivlist\@endpefalse
  \penalty100
}
\makeatother

\newcommand{\reductionrule}[4]{
\begin{definition}{#1}\\
\label{#2}
\vspace{3pt}
\setlength{\extrarowheight}{5pt}
\begin{tabularx}{\textwidth}{lX}
{\bfseries Guard}: & {
\begin{minipage}[t]{\linewidth}
\itemsep-0.5ex 
#3
\end{minipage}}
\\
{\bfseries Action}: &{
\begin{minipage}[t]{\linewidth}
\itemsep-0.5ex 
#4
\end{minipage}}
\end{tabularx}

\end{definition}
}
\newcommand{\reductionrulenodef}[2]{
\vspace{3pt}
\setlength{\extrarowheight}{5pt}
\begin{tabularx}{\textwidth}{lX}
{\bfseries Guard}: & {
\begin{minipage}[t]{\linewidth}
\itemsep-0.5ex 
#1
\end{minipage}}
\\
{\bfseries Action}: &{
\begin{minipage}[t]{\linewidth}
\itemsep-0.5ex 
#2
\end{minipage}}
\end{tabularx}
}

\usepackage{algpseudocode,algorithmicx,algorithm}
\makeatletter
\newcommand{\algmargin}{\the\ALG@thistlm}
\makeatother

\usepackage{tabularx}
\usepackage{enumerate}

\newcommand{\je}[1]{#1}
\newcommand{\ph}[1]{#1}
\newcommand{\jd}[1]{#1}

\newcommand{\theoremlike}[2]{\par\medskip\penalty-250\refstepcounter{theorem}{{\bfseries\noindent#2
\ref{#1}.}}}

\newcommand{\thmhelperpre}[2]{\theoremlike{#1}{#2}}
\newcommand{\thmhelperpost}{\par\medskip}

\newenvironment{reftheorem}[1]{\thmhelperpre{#1}{Theorem}}{\thmhelperpost}
\newenvironment{reflemma}[1]{\thmhelperpre{#1}{Lemma}}{\thmhelperpost}

\newenvironment{refproposition}[1]{\thmhelperpre{#1}{Proposition}}{
\thmhelperpost }

\title{\jd{Negotiation as Concurrency Primitive}\thanks{This work was partially supported by the Graduiertenkolleg 1480 PUMA and the project ``Negotiations: A Model of Tractable Concurrency,'' both funded by the German Research Council, and by the Institute of Advanced Studies of the Technical University of Munich.}}
\author[1]{J\"org Desel}
\affil[1]{FernUniversit\"at in Hagen}
\author[2]{Javier Esparza}
\author[2]{Philipp Hoffmann}
\affil[2]{Technische Universit\"at M\"unchen}

\begin{document}
\maketitle

\begin{abstract}
\jd{This paper introduces} negotiations, a model of concurrency close to Petri nets, 
with multi-party negotiations as concurrency primitive. We study two fundamental analysis problems.
The soundness problem consists of deciding if it is always possible for a negotiation to 
terminate successfully, whatever the current state \jd{is}. Given a sound negotiation, 
the summarization problem 
\jd{aims at} computing an equivalent  
one-step negotiation with the same input/output behavior. 
The soundness and summarization problems can be solved 
by means of simple algorithms acting on the state space of the negotiation, which however 
face the well-known state explosion problem. We study alternative algorithms
that avoid the construction of the state space. In particular, we define reduction
rules that simplify a negotiation while preserving the sound/non-sound character of the negotiation
and its summary. In a first result we show that our rules are complete for the class of weakly 
deterministic acyclic negotiations, meaning that they reduce all sound negotiations in this class, and only them,
to equivalent one-step negotiations. This provides algorithms for both the soundness and summarization problems
that avoid the construction of the state space. We then study the class of deterministic negotiations.
Our \ph{second} main result shows that the rules are also complete for this class, even if the negotiation contains cycles.
Moreover, we present an algorithm that completely reduces all sound deterministic negotiations, and only them, in polynomial time. 
\end{abstract}

\section{Introduction}

Negotiation has long been identified as a paradigm for process interaction 
\cite{davis1983negotiation}. It has been applied to different problems
(see e.g. \cite{winsborough2000automated,atdelzater2000qos}), and
studied on its own \cite{jennings2001automated}. The purpose of this paper 
is to initiate the study of negotiations from a concurrency-theoretic,
observational point of view. Observationally, a negotiation is an interaction in 
which several partners come together to agree on one out of a number of 
possible outcomes (a synchronized nondeterministic choice). While such an interaction
can be modelled in any standard process algebra
as a combination of parallel composition and nondeterministic 
choice, or as \jd{a} small Petri net, we argue that much can be gained by studying 
formal models with {\em atomic negotiation} as concurrency primitive. In 
particular, we show that the negotiation point of view reveals new 
classes of systems with polynomial analysis algorithms.

Atomic negotiations can be combined into {\em distributed negotiations}. 
For example, a distributed negotiation between a buyer, a seller, and a broker
consists of one or more rounds of atomic negotiations involving the buyer and 
the broker, or the seller and the broker, followed by a 
\jd{concluding} atomic negotiation 
between the buyer and the seller. 

We introduce a formal model for distributed negotiations, \jd{inspired by}
van der Aalst's {\em workflow Petri nets} \cite{aalst}. A negotiation atom, or just {\em atom}, involves a set of
{\em parties} (for instance, buyer and broker), and has a set of possible {\em outcomes}
(for example, buy and sell). Each party has a set of possible {\em internal states},
and each outcome has associated a {\em state-transformer}; if the parties agree
on a given outcome, then its associated transformer 
determines the possible ``exit states'' of the parties after \ph{executing} the atom as a function
of their ``entry states'' before \ph{executing} the atom. Atoms are combined into 
{\em distributed negotiations}, by means of a {\em next-atoms} function 
that determines for each atom, each party, and each outcome, the set 
of atoms the party is ready to engage in next if the atomic negotiation ends with that 
outcome. 
\jd{We} assume that a distributed negotiation
always starts with an {\em initial atom} and end\jd{s} with a {\em final atom},
both of which involve all parties of the complete distributed negotiation.

Like \jd{workflow} nets, distributed negotiations can 
be {\em unsound} because of deadlocks or livelocks. The {\em soundness} 
problem consists of deciding if a given negotiation is sound. Further, 
a sound negotiation is equivalent to one single atom whose state-transformer
determines the possible final internal states of all parties
as a function of their initial internal states.  
The {\em summarization problem} consists of computing such an atomic negotiation, 
called a {\em summary}. \jd{The set of reachable global states can be finite or infinite }
\ph{because each party can have infinitely many internal states}.
\jd{If it is finite, then}
the soundness and summarization problems can be solved 
by means of well-known algorithms based on the exhaustive 
exploration of the state space. However, this approach badly suffers 
from the state-explosion problem: the state-space of distributed negotiations
grows exponentially in the size of the negotiation itself, \jd{even if all atoms 
had only one single local state}.

In this paper we first show that the state-explosion problem cannot be avoided in full generality,
because the soundness problem is PSPACE-complete for arbitrary negotiations,
and a decision problem related to the summarization problem is PSPACE-complete even for 
negotiations in which each party has only two internal states. We then provide {\em reduction} 
algorithms for the soundness and summarization 
problems. Reduction algorithms exhaustively 
apply syntactic reduction rules that simplify the system while preserving some aspects of the behavior,
like absence of deadlocks. The rules are exhaustively applied, until the system is either a trivial
one, or a non-trivial system, but with a smaller state space. This approach has been extensively applied to 
Petri nets \jd{and} workflow nets, but most of this work has been devoted to the 
liveness or soundness problems \cite{DBLP:conf/ac/Berthelot86,DBLP:conf/apn/Haddad88,DBLP:journals/ppl/HaddadP06}. 
For these problems many reduction rules are known, and they have been proved
{\em complete} for certain classes of systems \cite{DBLP:journals/tcs/GenrichT84,DBLP:books/daglib/0073545,Desel:1995:FCP:207572}, meaning that they reduce all live or sound systems in the class,
and only those, to a trivial system (in our case,
 to a single atom). However, many of these rules, like the linear dependency rule 
of \cite{Desel:1995:FCP:207572}, cannot be applied to the summarization problem, because 
they do not preserve the summary, only the soundness property. 

{\jd We provide a solution  to the summarization problem for
{\em deterministic} negotiations, based on a complete set of reduction rules}. In deterministic negotiations 
all involved agents are deterministic,
 meaning that they are never ready to engage in more than one atomic negotiation. 
\jd{We} provide a reduction strategy guaranteeing that a sound deterministic negotiation will be summarized 
by means of a polynomial number of applications of the rules. 
 
 Intuitively, 
nondeterministic agents may be ready to engage in 
\jd{several} atomic negotiations, and which one takes place 
is decided by the deterministic parties, which play thus the role of negotiation leaders.  
We \jd{also} introduce {\em weakly deterministic}
negotiations (\jd{roughly speaking, in weakly deterministic negotiations, each atomic negotiation has a deterministic party}),
and provide a complete set of reduction rules for acyclic weakly deterministic negotiations. 

This paper is the definitive version of the results presented in \cite{negI,negII}. 
It contains detailed proofs of all theorems. In particular, it corrects an error of the algorithm 
of \cite{negII}, which in a certain special case might not terminate. We also conduct a
better complexity analysis, leading to an algorithm of cubic complexity, instead of the
$O(n^4)$ analysis of \cite{negII}. 

\paragraph{Related work.} Specific distributed negotiation protocols have been modeled with
 the help of Petri nets or process 
calculi (see e.g. \cite{SalaunFC04,JiTL05,BacarinMMA11,Cap12}). 
However, these papers do not address the issue of negotiation as concurrency primitive. 

The reduction rules introduced in Petri net 
theory by Berthelot and Haddad \cite{DBLP:conf/ac/Berthelot86,DBLP:conf/apn/Haddad88}
have more recently been applied to the analysis of workflow nets in \cite{DBLP:conf/caise/DongenAV05,DBLP:journals/jcss/VerbeekWAH10}. 
These works do not address completeness issues.
Complete rule sets have been proposed for several Petri net classes \cite{DBLP:journals/tcs/GenrichT84,DBLP:books/daglib/0073545,Desel:1995:FCP:207572}. However, the rules only preserve  a property close to soundness,
and so cannot be used for summarization.

\paragraph{Applications.} In work published after \cite{negI,negII} we have applied the reduction 
algorithms of this paper to the analysis of workflow Petri nets. Deterministic negotiations are 
very tightly related to free-choice workflow nets \cite{DBLP:conf/apn/DeselE15}, and the reduction algorithms can be adapted. 
In \cite{DBLP:conf/fase/EsparzaH16} we report on experimental results on the analysis of a collection of almost 2000 workflows 
from industrial sources. 

\paragraph{Structure.} The paper is structured as follows. Section \ref{sec:definitions} introduces 
the syntax and semantics of negotiations, and the classes of deterministic and weakly deterministic 
negotiations. Section \ref{sec:analysis} presents the soundness and summarization problems,
and analyzes their computational complexity. Section \ref{sec:rules} introduces our set 
of reduction rules. The rest of the paper presents completeness and complexity results of our
set of rules. Section~\ref{sec:acyclic} studies acyclic negotiations. It shows that our rules 
summarize all sound and acyclic weakly-deterministic negotiations, and presents a simple reduction 
strategy that reduces all acyclic sound deterministic negotiations in polynomial time. 
Sections \ref{sec:cyclic} and \ref{subsec:many-agents-cyclic} present the main result of the paper: 
a polynomial reduction algorithm for arbitrary deterministic negotiations. 
Several lengthy proofs are containd in a number of appendices.

\section{Negotiations: Syntax and Semantics}
\label{sec:definitions}


We fix a finite set $\agents$ of \emph{agents} representing potential parties 
of negotiations.
Each agent $p \in \agents$ has a (possibly infinite) nonempty set $Q_p$ of {\em internal states}. 
We denote by $Q_\agents$ the cartesian product $\prod_{p \in \agents} Q_p$, for a subset $P \subseteq \agents$ we similarly define $Q_P = \prod_{p \in P} Q_p$.

A \emph{transformer} is a left-total relation $\tau \subseteq Q_\agents \times Q_\agents$, representing a nondeterministic state transforming function. Given $P \subseteq \agents$,
we say that a transformer $\tau$ is a {\em $P$-transformer} if, for each $p_i \notin P$, 
$\left((q_{p_1}, \ldots, q_{p_i}, \ldots, q_{p_{|A|}}), (q'_{p_1}, \ldots,q'_{p_i}, \ldots, q'_{p_{|A|}})\right)\in\tau$ implies $q_{p_i} = q'_{p_i}$.
This means that a $P$-transformer only transforms the internal states of agents in $P$. 

\begin{definition}
A {\em negotiation atom}, or just an \emph{atom}, is a triple $n=(P_n, R_n,\delta_n)$,
 where $P_n \subseteq \agents$ is a nonempty set of \emph{parties}, 
$R_n$ is a finite, nonempty set of \emph{results}, and 
$\delta_n$ is a mapping assigning to each result $\r$ in $R_n$ 
a $P_n$-transformer $\delta_n (\r)$. \jd{For each result $r\in R_n$, we call the pair $(n,r)$ an \emph{outcome}}.
\end{definition}

\noindent Intuitively, if the respective
states of the agents before an atom $n$ are given by a tuple $q$ 
and the result of the negotiation is $\r$, then the agents change
their states to another tuple $q'$ for some $(q,q') \in \delta_n (\r)$. 

For a simple example, consider a negotiation atom $n$ with 
parties \texttt{F} (Father) and \texttt{D} (teenage Daughter). The goal of the
negotiation is to determine whether \texttt{D} can go to a party, and the time 
at which she must return home. The possible results are
$\{\texttt{yes}, \texttt{no}, \texttt{ask\_mother}\}$.
Both sets $Q_{\texttt{F}}$ and $Q_{\texttt{D}}$ contain a state 
{\it angry} plus a state $t$ for every time $T_1 \leq t \leq T_2$ in a 
given interval $[T_1,T_2]$
\jd{(we assume that the time values are appropriately ordered, choose for example $\{9,10,11,12\}$)}.
Initially, \texttt{F} is in state $t_f$ and \texttt{D} in state $t_d$.

\jd{If the state of } \ph{$\texttt{F}$ or of $\texttt{D}$} 
\jd{ is {\it angry}, then both will be {\it angry} after negotiating, no matter 
what the result of the negotiation atom is. 
The remaining transformations of the mapping $\delta_{n}$ are given by}

$$
\begin{array}{lcl}
\delta_{n}(\texttt{yes}) & \supset & \left\{ \left( (t_f, t_d), (t,t)\right) \; \mid \; t_f \leq t \leq t_d \vee t_d \leq t \leq t_f \right\} \\
\delta_{n}(\texttt{no}) & \supset & \left\{\left((t_f, t_d), ({\it angry}, {\it angry})\right) \; \right\} \\
\delta_{n}(\texttt{ask\_mother})& \supset & \left\{ \left( (t_f, t_d), (t_f, t_d) \right) \right\} 
\end{array} 
$$
\noindent That is, if the result is \texttt{yes}, then \texttt{F} and \texttt{D} agree on a time $t$ 
which is not earlier and not later than both suggested times.
If it is \texttt{no}, then there is a quarrel and both parties get angry. If the
result is \texttt{ask\_mother}, then the parties keep their previous times.

\subsection{Combining Atomic Negotiations}

A negotiation is a composition of atoms. We add a {\em transition function} $\trans$ that assigns to every 
triple $(n,p,\r)$ consisting of an atom $n$, a participant $p$ of $n$, and a result $\r$ of $n$ a set 
$\trans(n,p,\r)$ of atoms. Intuitively, this is the set of atomic negotiations 
agent $p$ is ready to engage in after the atom $n$, if the result of $n$ is $\r$. 

Negotiations can be graphically represented as shown in Figure \ref{fig:FDM_example} 
(ignore the black dots on the arrows for now). 
For each atom $n \in N$ we draw a black bar; for each party $p \in P_n$ we 
draw a white circle on the bar, called a \emph{port}. For each triple $(n,p,\r) $, 
we draw a hyperarc leading from the port of $p$ in $n$ to all the ports of $p$ in 
the atoms of $\trans(n,p,\r)$, and label it by $\r$. (If $\trans(n,p,\r)$ contains only one atom, then the hyperarc is  actually an arc.)
\jd{If $\trans(n,p,\r) = \trans(n,p,\r')$, i.e., if the next possible atoms of an agent $p$
after the negotiation atom $n$ are the same for two results $r$ and $r'$, then 
we draw a single hyperarc and label it by $r$ and $r'$, and similarly for more than two results.}
In later examples, we   omit the labels whenever we are only interested in the structure of the 
negotiation. 

\begin{figure}[h]
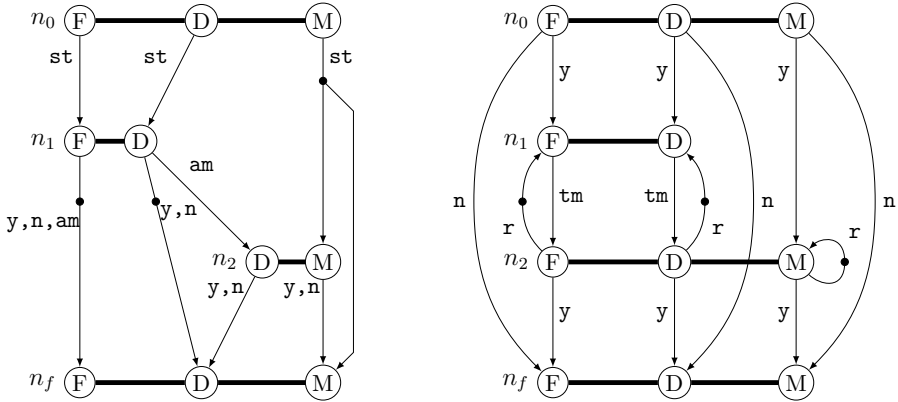

\centering
\scalebox{0.8}{
\input{tikz/FDM_acyclic}
\hspace{1cm}
\input{tikz/FDM_cyclic}
}
\caption{Two negotiations between three agents.}
\label{fig:FDM_example}
\end{figure}

Figure \ref{fig:FDM_example} shows two Father-Daughter-Mother negotiations,
omitting, for the sake of brevity, the state transformers associated to the results.
On the left, after an initial atomic negotiation $n_0$ involving all three agents,
Daughter and Father negotiate \ph{(atom $n_1$)} with possible results \texttt{yes} ($\texttt{y}$), 
\texttt{no} ($\texttt{n}$), and \texttt{ask\_mother} ($\texttt{am}$). If the
result is \texttt{ask\_mother}, then Daughter and Mother negotiate \ph{(atom $n_2$)} with possible results
\texttt{yes} and \texttt{no}. The negotiation ends with a final atom $n_f$

In the negotiation on the right,
Father, Daughter and Mother \ph{initially} negotiate with possible results \texttt{yes} and \texttt{no}.
If the result is \texttt{yes}, then Father and Daughter negotiate a time for the Daughter to return home 
(atom $n_1$) and propose it to Mother (atom $n_2$). If Mother approves (result 
\texttt{yes}), then the negotiation terminates with $n_f$, otherwise (result \texttt{\r})
Daughter and Father renegotiate the time.

Before defining negotiations, we need a more general structure that we call pre-negotiation:

\begin{definition}
Given a finite set of atoms $N$, let $T(N)$ denote the set of triples $(n, p, \r)$ such that 
$n \in N$, $p\in P_n$, and $\r \in R_n$. 
A \emph{pre-negotiation} is a tuple ${\cal N}=(N, n_0, n_f, \trans)$, where 
$n_0, n_f \in N$ are the \emph{initial} atom and the \emph{final} atom, and $\trans \colon T(N) \rightarrow 2^N$ is the {\em transition function}. 
\end{definition}

Abstracting from ports leads to the notion of graph of a pre-negotiation:

\begin{definition}
Let ${\cal N}=(N, n_0, n_f, \trans)$ be a pre-negotiation. The {\em graph} of $\N$  is the directed graph with
$N$ as set of vertices and an edge from $n$ to $n'$ labeled by $(p, \r)$ if $n' \in \trans(n, p, \r)$.

A {\em path} of $\N$ is a sequence $(n_1, p_1, \r_1) \, (n_2, p_2, \r_2) \cdots (n_k, p_k, \r_k)$ 
such that, for  $1 \leq i \leq k-1$, $n_{i+1} \in \trans(n_i, p_i, \r_i)$ . 
If $n_1 \in \trans(n_k, p_k, \r_k)$ then the path is also a {\em cycle}.

${\cal N}$ is \emph{acyclic} if its graph has no cycles, otherwise it is \emph{cyclic}.
\end{definition} 

\begin{figure}[h]
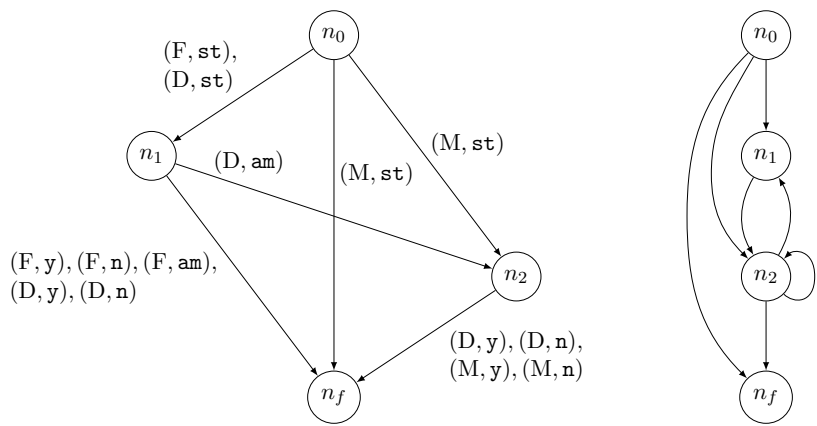

\centering
\scalebox{0.8}{
\input{tikz/Graph_example_acyclic}
\hspace{1cm}
\input{tikz/Graph_example_cyclic}
}
\caption{Graphs of the negotiations in Figure \ref{fig:FDM_example}.}
\label{fig:FDM_graphs}
\end{figure}

The graphs of the negotiations in Figure \ref{fig:FDM_example} are shown in Figure \ref{fig:FDM_graphs}.
In the graph of the negotiation on the right we have omitted the labels. The negotiation 
on the left of Figure \ref{fig:FDM_example} is acyclic, while the one the right is cyclic. \jd{Observe that in acyclic pre-negotiations and negotiations every agent can reach an atom at most once.}

Now we are ready to formally define negotiations. 

\begin{definition}
\label{def:neg}
A pre-negotiation ${\cal N}=(N, n_0, n_f, \trans)$ is a \emph{negotiation} if it satisfies the following properties: 

(1) every agent of $\agents$ participates in both $n_0$ and $n_f$ \jd{(i.e., $P_{n_0} = P_{n_f} = \agents$)};

(2) for every $(n, p, \r) \in T(N)$: $\trans(n, p, \r)= \emptyset$ if{}f $n=n_f$.

(3) for every atom $n\in N$ there is a path $(n_1, p_1, \r_1) \, \cdots \,(n_k, p_k, \r_k)$ such that $n_0 = n_1$, $n_f = n_k$, and $n = n_i$ for some $1 \leq i \leq k$. \medskip

\end{definition}

\noindent \jd{We call a result $\r$ of an atom $n$ \emph{final result for party $p$} if $\trans(n, p, \r) = \emptyset$.
Condition (2) of the above definition states that only the final atom has final results and that 
all results of the final atom are final results for all its parties, which is -- by (1) -- the set of all agents.} 

If $|N| = 1$ then the negotiation consists of a single atom, which is both the initial
atom $n_0$ and the final atom $n_f$; further, \jd{in this case} every result is final.

\jd{If $|N| > 1$ then $n_0$ and $n_f$ are distinct; otherwise, in the graph of $\N$ the vertex $n_0$ ($=n_f$) has no successor and so no other vertex can
be on a path starting with $n_0$.}


\jd{Figure \ref{fig:FDM_example} shows two negotiations. 
The initial atoms have no ingoing arc (which is not necessarily the case in general), whereas 
the final atoms have no outgoing arcs by definition. Therefore, 
there is no graphical representation of the final results. 
Since $\trans(n_f,p,\r) = \emptyset$ holds for every final result $\r$ 
and every agent $p$, any visual representation holds no further information other than the existence of that result.}

\jd{Condition (3) of the definition of negotiations is inspired by a corresponding property of workflow Petri nets \cite{aalst}.
Paths of the graph represent  possible subsequent atomic negotiations of an agent. Therefore, if for some vertex there was no path from the initial vertex (representing the initial atom) to that vertex, then the  corresponding atom can never take place (remember that all participants of all atoms participate in the initial atom).
As will become clear in the following subsection, we are mainly interested in distributed negotiations which can always eventually terminate (represented by the final atom); hence the requirement of a path from each vertex to the  vertex representing the final atom (which also involves all agents).}

\subsection{Semantics}

A \emph{marking} of a negotiation ${\cal N}=(N, n_0, n_f, \trans)$ is a mapping 
$\vx \colon \agents \rightarrow 2^N$. Intuitively, $\vx(p)$ is the set of atoms that agent $p$ is currently ready to engage in next. 
The \emph{initial} and \emph{final} markings, denoted by $\vx_0$ and $\vx_f$ respectively, are given by $\vx_0(p)=\{n_0\}$ and 
$\vx_f(p)=\emptyset$ for every $p \in \agents$.

A marking $\vx$ \emph{enables} an atom $n$ if $n \in \vx(p)$ for every $p \in P_n$,
i.e., if every party of $n$ is currently ready to engage in it.
If $\vx$ enables $n$, then $n$ can take place and its parties
agree on a result $\r$; we say that \jd{the outcome} $(n,\r)$ \emph{occurs}.
The occurrence of $(n,\r)$ produces a next marking $\vx'$ given by $\vx'(p) = \trans(n,p,\r)$ for every $p \in P_n$, 
and $\vx'(p)=\vx(p)$ for every $p \in \agents \setminus P_n$. 
We write $\vx \by{(n,\r)} \vx'$ to denote this,
and call it a {\em small step}. 

We write $\vx_1 \by{\sigma}$ to denote that there is a sequence 
\[
\vx_1 \by{(n_1,\r_1)} \vx_2 \by{(n_2,\r_2)}\cdots \by{(n_{k-1},\r_{k-1})} \vx_{k} \by{(n_k,\r_k)} \vx_{k+1} \cdots
\] 
of small steps such that $\sigma = (n_1, \r_1) \ldots (n_{k}, \r_{k}) \ldots$. We
call $\sigma$ an {\em occurrence sequence} from marking $\vx_1$.
If $\sigma$ is finite \ph{and ends with $(n_k, \r_k)$}, then we write $\vx_1 \by{\sigma} \vx_{k+1}$ and say that $\vx_{k+1}$ is 
\emph{reachable} from $\vx_1$. 
If $\vx_1$ is the initial marking, then we call $\sigma$ an {\em initial occurrence sequence}. If moreover $\vx_{k+1}$ is the final marking, then $\sigma$ is a {\em large step}.

Given an agent $p$, \ph{we always have that} either $\vx(p)=\{n_0\}$ or $\vx(p)=\trans(n,p,\r)$ for some atom $n$ and result $\r$. 
The marking $\vx_f$ can only be reached by the occurrence of an outcome $(n_f, \r)$, where $n_f$ is the final atom and thus $\r$ is a final result.  It does not enable any atom. Any other marking that does not enable any atom is a \emph{deadlock}.
A marking which is reachable from itself but from which the final marking is not reachable is a \emph{livelock}.

Reachable markings are graphically represented by placing tokens (black dots) on the forking points of the hyperarcs (or, if the hyperarc consists of just one arc, in the middle of the arc). Figure \ref{fig:FDM_example} shows \jd{on the left a marking in which all agents are ready to engage in $n_f$ and \texttt{M} moreover is ready to engage in $n_2$. So the only enabled outcomes are $(n_f, \r_f)$,
where $\r_f$ is a final result.}
The figure on the right shoes a marking in which \texttt{F} and \texttt{D} are ready to engage in $n_1$ and \texttt{M} is ready to engage in $n_2$. \jd{Since  $n_1$ involves only \texttt{F} and \texttt{D} and has only the  result \texttt{tm}, the only enabled outcome is $(n_2, \texttt{tm})$.}


\subsection{Determinism and Weak Determinism}
\label{sec:determ}

We introduce deterministic and weakly deterministic negotiations. Intuitively, 
an agent of a negotiation is deterministic if it is never ready to engage in more than
one atom; graphically, an agent is deterministic if all hyperarcs of the agent are normal arcs. Formally:

\begin{definition}
An agent $p \in \agents$ is \emph{deterministic} if for every
$(n,p,\r) \in T(N)$ such that $n \neq n_f$ there exists one atom
$n'$ such that $\trans(n,p,\r) = \{n'\}$. 
\end{definition}

Consider again Figure \ref{fig:FDM_example}. In the negotiation 
on the left, Father and Daughter are deterministic agents, but Mother is not,
because she has a proper hyperarc from $n_0$ to $n_2$ and $n_f$. After $n_0$, Mother is ready to
engage in the atoms $n_2$ and $n_f$. 

The fundamental property of deterministic agents is that, loosely speaking, they ``force'' 
atoms to occur. Assume that 
a deterministic agent is currently ready to engage in atom $n$. Then the agent is
not ready to engage in any other atom, and the only way to change this is by executing $n$.
Therefore, any extension of the current execution to a large step must necessarily execute
$n$.


\begin{definition}
A negotiation is \emph{deterministic} if all its agents are deterministic.

A negotiation is {\em weakly deterministic} if for every $(n,p,\r) \in T(N)$ there is a deterministic
agent $b$ that is a party of every atom in $\trans(n,p,\r)$, i.e., $b \in P_{n'}$ for every 
$n' \in \trans(n,p,\r)$.
\end{definition}

Observe that every deterministic negotiation is also weakly deterministic.
Consider again Figure \ref{fig:FDM_example}. The right negotiation is 
deterministic, while the left negotiation is not, as Mother has 
a proper hyperarc from $n_0$ to $n_2$ and $n_f$. However, the left negotiation
is weakly deterministic, as Daughter is deterministic and
participates in every atom.




\section{Analysis Problems}
\label{sec:analysis}

We introduce a notion of well-behavedness of a negotiation \jd{which we call soundness. It is similar to the soundness property of workflow Petri nets \cite{aalst}}. 

\begin{definition}
A negotiation is {\em sound} if

\noindent
 {\em (a)} every atom is enabled at some reachable marking, and 
 
 \noindent
 {\em (b)} every occurrence sequence from the initial marking is either a large step or can be extended to a large step. 
\end{definition}

Intuitively, (a) captures that there are no useless atoms, and (b) that the negotiation can never reach
a state from which it cannot terminate. In particular, sound negotiations can reach neither a deadlock
nor a livelock.

The negotiations of Figure \ref{fig:FDM_example} are sound. However, if we set in the left negotiation
$\trans(n_0,\texttt{M}, \texttt{st})= \ph{\{n_2\}}$ instead of $\trans(n_0,\texttt{M}, \texttt{st})= \ph{\{n_2, n_f\}}$, then the occurrence sequence $(n_0,\texttt{st}) \ph{(n_1, \texttt{yes})}$
leads to a deadlock.


We now introduce the notion of summary transformer, and summary of a negotiation. 
Intuitively, the summary transformer gives for each final outcome $(n_f,\r)$ and for each
tuple of initial states $q_0$ the possible \ph{resulting} tuples of states\ph{,} after the
negotiation finishes with the result $\r$.  

\begin{definition}
\label{def:summary}
Given a negotiation $\N=(N,n_0,n_f,\trans)$, we attach to each \jd{final} result $\r$  a 
{\em summary transformer} $\trf{\N,\r} \subseteq Q_\agents \times Q_\agents$ as follows.

Given two transformers $\tau_1, \tau_2 \subseteq Q_\agents \times Q_\agents$, we define their
{\em concatenation} as the transformer
\begin{equation*}
\tau_1 \, \tau_2 =  \{ (q, q') \in Q_\agents \times Q_\agents \mid(q, q'') \in \tau_1 \mbox{ and } (q'', q') \in \tau_2  \mbox{ for some $q''\in Q_\agents$} \}
\end{equation*}
For every finite occurrence sequence $\sigma = (n_1, \r_1) (n_2, \r_2)\ldots (n_k, \r_k)$ of $\N$, define\\
$\trf{\sigma}= \delta_{n_1} (\r_1)\, \delta_{n_2} (\r_2) \cdots \, \delta_{n_k} (\r_k)$.
Let $L_\r$ be the set of large steps of $\N$ that end with $(n_f,\r)$. 
We define $\trf{\N,\r} = \bigcup_{\sigma \in L_\r} \trf{\sigma}$. 
\end{definition}

Finally, we introduce the notion of equivalent negotiations and summary.

\begin{definition}
\label{def:equivalence}
Two negotiations $\N_1$ and $\N_2$ over the same set of 
agents are {\em equivalent} (\mbox{$\N_1 \equiv \N_2$}) if 
\begin{itemize}
\item[(1)] they are either both sound, or both unsound;
\item[(2)] they have the same final results; and
\item[(3)] $\trf{\N_1, \r} = \trf{\N_2, \r}$ for every final \jd{result}  $\r$.
\end{itemize}
If $\N_1$ and $\N_2$ are equivalent and $\N_2$ consists of a single atom, 
then $\N_2$ is the {\em summary} of $\N_1$. 
\end{definition}

We collect some easy consequences of the definition.
\begin{itemize}
\item A negotiation has a summary if{}f it is sound. \\
Indeed, if $\N$ is sound, then \jd{the negotiation with only one atom,
all final results of $\N$ as results of this atom, 
 and with transformers $\trf{\N,\r}$ for each of these final results $\r$} is a summary. 
Conversely, if $\N$ has a summary $\N'$, then, since negotiations 
consisting of one single atom are sound, $\N'$ is sound. Since $\N \equiv \N'$, by condition (1) $\N$ is sound.\footnote{Alternatively, we could also define summaries of 
 unsound negotiations by introducing {\em unreliable atoms} that, intuitively, may ``get stuck''. The summary of an unsound negotiation 
would then be an unreliable atom with the same summary transformers as the 
original negotiation. In this paper we do not further investigate this possibility.}
\item Summaries are unique up to \jd{the identity of the single atom}.\\
A negotiation with one single atom $n$ is completely determined by the \jd{transformers} of $n$ (observe that if two atoms have the same \jd{transformers} then they necessarily have the same sets of agents and results).
So two one-atom negotiations whose atoms have the same \jd{transformers}  are equivalent, \jd{although the two single atoms need not be identical}.
\item Equivalence of negotiations is a congruence with respect to 
substitution: if in a negotiation $\N$ we replace a subnegotiation ${\cal M}$ by an equivalent negotiation ${\cal M'}$,
then the resulting negotiation $\N'$ is equivalent to $\N$. \\
The formal definitions of subnegotiation and ``replacing a negotiation 
by an equivalent one'' are the expected ones, and the proof is easy but laborious, and we omit it. \jd{Without} condition (2) in Definition \ref{def:equivalence} 
the congruence property does not hold. Indeed, without condition (2) the notion of substituting a negotiation by an equivalent one is not 
even well defined. 

\item While Definition \ref{def:equivalence} preserves soundness in a 
way that makes substitutions possible, \jd{there are behavioral aspects}
that are not preserved under equivalence. 
In particular, \je{a negotiation with infinite behaviors may be equivalent to a negotiation having none}.
\end{itemize}

\subsection{Deciding soundness}
The soundness problem consists of deciding if a given negotiation is
sound. It can be solved with the help of the reachability graph. 

The {\em reachability graph} of a negotiation $\N$ has all markings reachable from $\vx_0$ as vertices, and an edge from $\vx$ to $\vx'$ whenever $\vx \by{(n,\r)} \vx'$. To decide soundness we can (1) compute the reachability graph of $\N$
and (2a) check that every atom appears at some arc, and (2b) that, for every reachable marking $\vx$, there is an occurrence sequence $\sigma$
such that $\vx \by{\sigma} \vx_f$. 

Step (1) needs exponential time in the number of atoms, and steps (2a) and (2b) are polynomial 
in the size of the reachability graph. So the algorithm is single exponential in 
the number of atoms. Appendix \ref{app:complex} shows that this cannot be easily avoided, 
because the problem is PSPACE-complete, and co-NP-hard and in DP for acyclic negotiations
\jd{(a} language $L$ is in the class DP if there exist languages $L_1$ in NP and $L_2$ \jd{in}
co-NP such that $L=L_1 \cap L_2$ \cite{PapadimitriouY82}).

\begin{theorem}
\label{thm:complexity}
The soundness problem is PSPACE-complete. For acyclic negotiations, the problem is 
co-NP-hard and in DP (and so at level $\Delta^P_2$ of the polynomial hierarchy).
\end{theorem}
\begin{proof}
See the Appendix.
\end{proof}

\subsection{A summarization algorithm}
\label{subsec:sumalg}
The {\em summarization problem} consists of computing
a summary of a given negotiation, if it is sound. We show that 
\ph{the summary}  can be  computed from the {\em labeled reachability graph}
of the negotiation. 

Let $\trf{n,\r}$ denote the transformer of the outcome $(n,\r)$. The labeled reachability
graph of a negotiation $\N$ is defined as its reachability graph, but labeling the 
edge corresponding to a step $\vx \by{(n,\r)} \vx'$ with the transformer $\trf{n,\r}$.
In other words, the labeled reachability graph has an edge from $\vx$ to $\vx'$ labeled with $\trf{n,\r}$
for every step  $\vx \by{(n,\r)} \vx'$.

Observe that the labeled reachability graph is in fact a multi-graph, 
since there may be more than one edge between two nodes.
Figure \ref{fig:reachgraph} shows a negotiation and its labeled 
reachability graph. If all the results of a negotiation have different 
names, we can identify an outcome $(n,\r)$ with the result $\r$ and further shorten $\trf{n,\r}$ to $\trf{\r}$.

\begin{figure}[h]
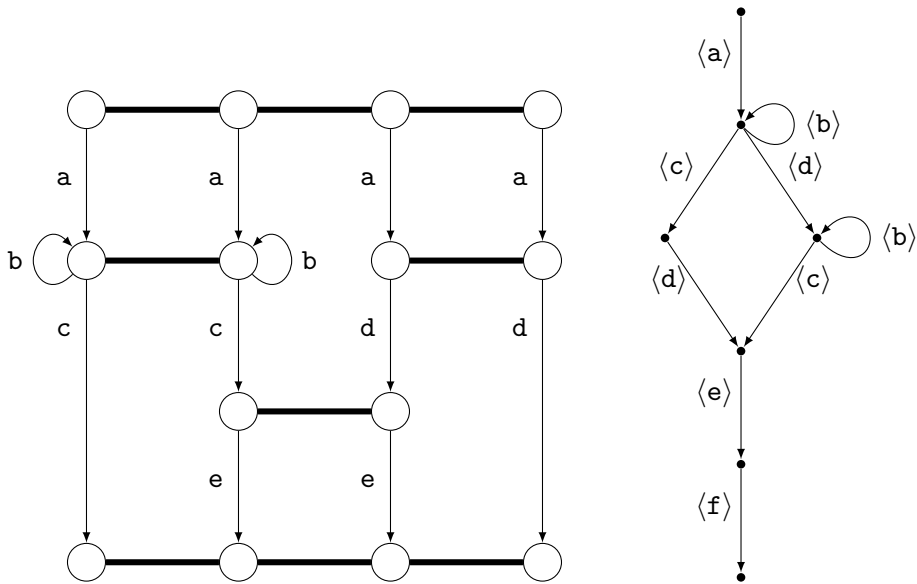

\begin{center}
\centerline{
\input{tikz/lts.tex} \qquad
\input{tikz/lts2.tex}
}
\end{center}
\caption{A negotiation and its labeled reachability graph. We assume that the only outcome of the
final atom is $(n_f, \texttt{f})$.}
\label{fig:reachgraph}
\end{figure}

The concatenation of two transformers was already introduced in Definition \ref{def:summary}.
Similarly we define the union $\tau_1 \cup \tau_2$ of two transformers, 
and the Kleene star $\tau^*$ of a transformer:
$$\begin{array}{rcl}
\tau_1 \cup \tau_2 & = & \{ (q, q') \in Q_\agents \times Q_\agents \mid (q, q') \in \tau_1 \mbox{ or } (q, q') \in \tau_2 \}\\
\tau^0 & = & \{ (q, q) \in Q_\agents \times Q_\agents \mid q \in  Q_\agents \} \\
\tau^{i+1} & = & \tau \, \tau^i \mbox{ for every $i \geq 0$ }\\
\tau^*     & = & \bigcup_{i \geq 0} \tau^i
\end{array}$$

By definition, the summary transformer $\trf{\N, \texttt{f}}$ is the union 
over all large steps $\sigma$ ending with $(n_f, \texttt{f})$
of the transformers $\trf{\sigma}$. We recall a well-known algorithm for the 
computation of this union based on state elimination (see e.g. \cite{HUM}). 
The algorithm iteratively reduces the graph to smaller graphs with the same 
summary transformers, until the graph contains only the nodes $\vx_0$ and $\vx_f$,
and one edge $\vx_0 \by{\tau_r} \vx_f$ for each final outcome $(n_f, \texttt{f})$ of
$\N$. Then we have $\tau_r = \trf{\N, \texttt{f}}$. At each step, the algorithm applies one \ph{of} the
following three {\em reduction rules}:

\begin{itemize}
\item[(1)] Replace a pair $\vx_1 \by{\tau} \vx_2$, $\vx_1 \by{\tau'} \vx_2$ of distinct 
edges, where $\vx_2 \neq \vx_f$,  by an edge $\vx_1 \by{\tau \cup \tau'} \vx_2$. 
\item[(2)] Given a self-loop $\vx \by{\tau} \vx$, replace every edge $\vx \by{\tau'} \vx'$ such that $\vx' \neq \vx$ by an edge $\vx \by{\tau^* \tau'} \vx'$, and then remove the self-loop.
\item[(3)] Given an edge $\vx' \by{\tau'} \vx$ such that $\vx' \neq \vx$ and $\vx$ has at least one successor, add for every
edge $\vx \by{\tau''} \vx''$ a {\em shortcut edge} $\vx' \by{\tau' \,\tau'' } \vx''$, 
and then remove the edge $\vx' \by{\tau} \vx$. Further, if after removing this edge the node
$\vx$ has no other incoming edges, then remove $\vx$, together with all its outgoing edges.
\end{itemize}

The algorithm applies rules (1)-(3) in phases. At each phase, it follows this
strategy: 
\begin{itemize}
\item Apply rule (1) as long as possible.
\item Apply rule (2) as long as possible. 
\item When neither rule (1) nor (2) are applicable, select a node $\vx$ different from $\vx_0$ and $\vx_f$,
and apply rule (3) to {\em all} its ingoing edges, that is, to all  edges of the form $\vx' \by{\tau'} \vx$.
(Observe that this step necessarily ends with the removal of $\vx$ and its outgoing edges.)
\end{itemize}

It is easy to see that the application of any of these rules does 
not change the summary of the graph. In particular, the idea behind rule (3) is that, in every large step, a step $\vx' \by{\tau'} \vx$ must necessarily be followed by one of the steps $\vx \by{\tau''} \vx''$. \jd{After} all shortcut edges have been added, the edge $\vx' \by{\tau'} \vx$ becomes redundant: the two
graphs with and without the edge have the same summary.

\jd{After} each phase the number of nodes decreases by one. If the original negotiation is sound 
(which implies that every
node of the reachability graph lies on a path starting at the \jd{initial marking} $\vx_0$ and ending at \jd{the final marking} $\vx_f$), then the algorithm
can only terminate with a graph containing exactly the nodes $\vx_0$ and $\vx_f$, and an edge from 
$\vx_0$ to $\vx_f$ for each \jd{final} result. If the original negotiation is not sound because \jd{from} some 
reachable marking the final marking cannot be reached, then the algorithm terminates with a graph containing 
additional nodes. If the original negotiation is not sound only because some atom can never be executed, 
then the algorithm terminates with a graph as in the sound case. 

When applied to the labeled reachability graph shown in Figure \ref{fig:reachgraph}, the algorithm reduces the graph completely because the negotiation is sound. Figure \ref{fig:reductioninter} shows some intermediate steps of the algorithm.

\begin{figure}[h]
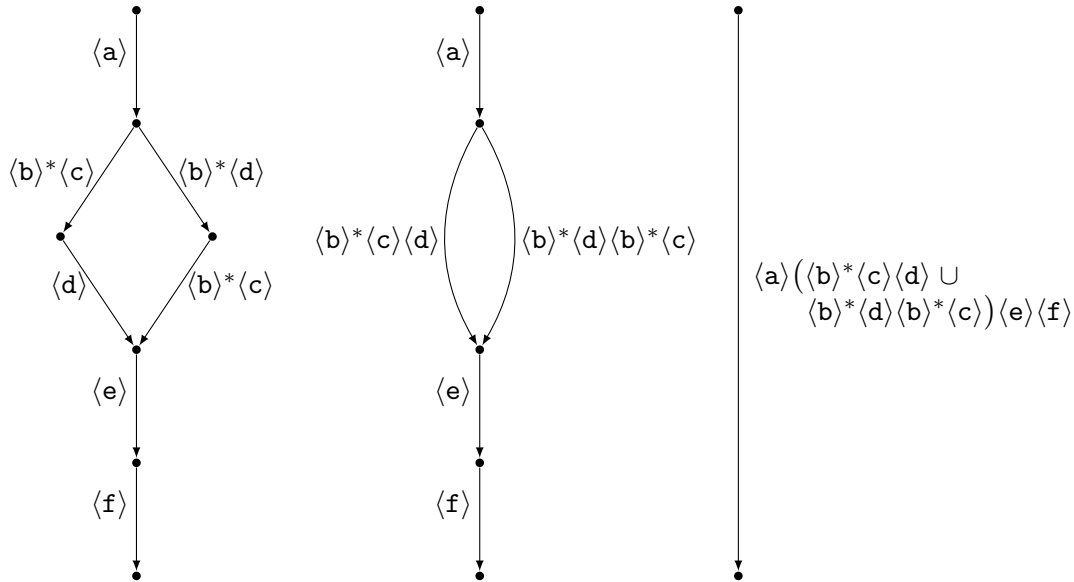

\input{tikz/lts3.tex}
\input{tikz/lts4.tex}
\input{tikz/lts5.tex}
\caption{Intermediate steps in the reduction of the graph of Figure \ref{fig:reachgraph}.}
\label{fig:reductioninter}
\end{figure}

This algorithm is simple and elegant, but it suffers from the state-explosion problem: the size of the 
reachability graph can be exponential in the size of the negotiation, and so the algorithm has
exponential worst-case time complexity. While the complexity results of 
Section \ref{sec:analysis} show that
this is unavoidable unless PSPACE=P, the main problem is that the algorithm takes exponential time for
{\em any} family of negotiations whose reachability \jd{graphs exhibit} exponential growth, even if \jd{the negotiations} have a very
simple structure. 

In the next section we present reduction rules that act directly on the negotiation diagram, {\em not} on the
reachability graph, and thus can be applied without constructing it. The rules can be applied to any negotiation,
until no reduction is possible anymore. If the resulting negotiation consists of one single atom, the summary
can be read out directly from it. Otherwise the algorithm above can be applied, with the advantage that,
since the negotiation is now smaller, the reachability graph is smaller, too. 


\section{Reduction Rules for Negotiations}
\label{sec:rules}

We introduce four \jd{reduction} rules acting on negotiation diagrams.
\jd{The rules are correct in the following sense:
The} negotiation obtained after applying a rule is equivalent to
the negotiation before applying it. In particular, this implies that the rules
preserve soundness, \jd{unsoundness} and the summary.

The first two rules are straightforward generalizations of the rules
applied in steps (1) and (2) of the summarization algorithm of Section \ref{subsec:sumalg}.

The third rule deals with a characteristic of negotiation diagrams that is not present when 
we consider their reachability graphs. Finally, the last rule is a generalization of the one 
applied in step (3), but it is far from straightforward, and we \ph{will} discuss it in detail.

\subsection{Reduction Rules}

A {\em reduction rule}, or just a rule,
is a binary relation on the set of negotiations. Given a rule $R$,
we write ${\cal N}_1 \by{R} {\cal N}_2$ for $({\cal N}_1, {\cal N}_2) \in R$.
{\jd Notice that a rule $R$ is not necessarily applicable to a given negotiation ${\cal N}_1$.  
If it is applicable, then the resulting negotiation ${\cal N}_2$ is not necessarily unique.}

We describe rules as pairs of a {\em guard} and an {\em action};
${\cal N}_1 \by{R} {\cal N}_2$ holds if ${\cal N}_1$ satisfies the guard and
${\cal N}_2$ is a possible result of applying the action to ${\cal N}_1$.

A rule $R$ is {\em correct} if its application to a negotiation {\jd always} yields an equivalent negotiation, i.e., if
${\cal N}_1 \by{R} {\cal N}_2$ implies ${\cal N}_1 \equiv{\cal N}_2$. In particular, this
implies that ${\cal N}_1$ is sound if{}f
${\cal N}_2$ is sound.


A finite sequence
$R_1 \ldots R_k$ \jd{of rules} is a {\em reduction sequence for a negotiation $\N$} if there are
$\N_1, \ldots, \N_k$ such that $\N \by{R_1} \N_1 \by{R_2} \cdots \by{R_k} \N_k$. We say that the sequence
reduces $\N$ to $\N_k$. Infinite reduction sequences are defined similarly\footnote{Infinite reduction sequences are of course undesirable. We define them, \jd{but} show how to avoid them.}. 

 \jd{A set of} \ph{ correct } \jd{rules} ${\cal R}$ is {\em complete with respect to a class of negotiations} if every sound
negotiation in the class can be reduced to a negotiation consisting of a single atom \jd{by a finite sequence of reductions in ${\cal R}$}.

\jd{In the following we introduce the reduction rules for negotiations. For convenience, we assume in this section that a rule is applied to a negotiation $\N = (N, n_0, n_f, \trans )$.
The actions are formulated as assignments to the components of $\N$, so that the reduced negotiation is the one obtained after performing these assignments.}



\subsection{Merge and iteration rules}

\paragraph{Merge rule.}
The merge rule merges results with identical transition functions into one.

\reductionrule{Merge rule}{def:merge}{
	$\N$ contains an atom $n$, $n \neq n_f$ with two distinct results $\r_1, \r_2 \in R_n$ such that $\trans(n,p,\r_1) = \trans(n,p,\r_2)$ for every $p \in P_n$.
}{
	\begin{enumerate}[(1)]
	\itemsep-0.5ex
	\item $R_n \leftarrow (R_n \setminus \{\r_1, \r_2\}) \cup \{\r\}$,
	where $\r$ is a fresh name.
	\item For all $p \in P_n$: $\trans(n,p,\r) \leftarrow \trans(n,p,\r_1)$.
	\item $\delta_n (\r) \leftarrow \delta_n (\r_1) \cup \delta_n (\r_2)$.
	\end{enumerate}
}

\begin{figure}[h]
\centering
\input{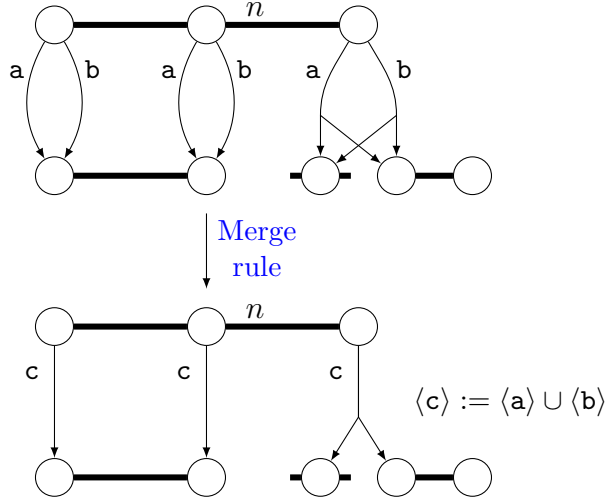}
\caption{Illustration of the merge rule}
\label{fig:mergeExample}
\end{figure}

An example \jd{is} sketched in Figure \ref{fig:mergeExample}. 
The results \texttt{a} and \texttt{b} are merged into one result \texttt{c} with an appropriate transformer.

Observe that the guard forbids application \jd{of the merge rule} to the final atom; by the definition of equivalence,
merging results of the final atom yields a non-equivalent negotiation. Notice further that the 
result of applying the rule is always a negotiation, according to Definition \ref{def:neg}.

\paragraph{Iteration rule.}
Loosely speaking, the iteration rule removes self-loops, i.e., \jd{it removes} results of an atom
after which {\em all} parties are  \jd{ready} to take part only in the same atom again.
The rule changes the transformers of all other results of that atom.

\reductionrule{Iteration rule}{def:iteration}{
	$\N$ contains an outcome $(n, \r)$ such that $\trans(n,p,\r)= \{n\}$ \jd{for} every party $p$ of $n$.
}{
	\begin{enumerate}[(1)]
	\itemsep-0.5ex
	\item $R_n \leftarrow R_n \setminus \{\r\}$.
	\item For every $\r' \in R_n $: $\delta_n (\r') \leftarrow \delta_n (\r)^*\: \delta_n (\r')$.
	\end{enumerate}
}

\begin{figure}[h]
\centering
\input{tikz/ruleExamples/iteration}
\caption{Illustration of the iteration rule}
\label{fig:iterationExample}
\end{figure}

For an example, consider Figure \ref{fig:iterationExample}. 
The result \texttt{a} for which all agents stay in $n$ is removed and the transformer of result \texttt{b} is changed. 

\je{Observe that the application of the rule always yields a negotiation. }
\jd{Indeed, before applying the rule every atom lies on some path of the negotiation graph 
leading from the initial to the final atom; since the rule only removes a self-loop of this graph, 
the same property holds after the rule is applied. 
}

The correctness of the merge and iteration rules is an immediate consequence of the definitions,
and we state it without proof.

\begin{theorem}
The merge rule and \jd{the} iteration rule  are correct.
\end{theorem}


\subsection{Useless arc rule}

\jd{If the graphical representation of a negotiation
contains a hyperarc for some agent leading from an atom $n$ to 
atoms $\{n_1, \ldots, n_k\}$,  we say that this hyperarc consists of $k$ different arcs.}

\begin{definition}
An {\em arc} of \ph{the} negotiation $\N$ is a tuple $(n,p,\r,n')$ such that
$(n,p,\r) \in T(N)$ and $n' \in \trans(n, p,\r)$. 
\end{definition}

\begin{figure}
\centering
\input{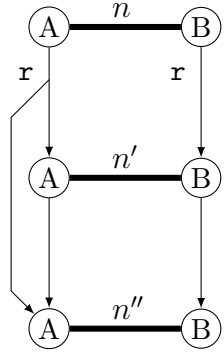}
\caption{Example of a useless arc}
\label{fig:example_useless}
\end{figure}

Some of these arcs may be ``useless''. Intuitively, an arc is useless if no occurrence sequence makes a token flow along it.
\jd{More precisely, an arc $(n,p,\r,n'')$ is useless if in every initial occurrence sequence,  atom $n''$ is never enabled after the occurrence of the outcome $(n,\r)$, and $n''$ remains disabled unless another atom involving agent $p$ has occurred}.

Consider \jd{as an} example the negotiation in Figure \ref{fig:example_useless}. The arc for agent A from $n$ to $n''$ can never
 be ``used'', that is, after outcome $(n,\r)$ agent A is ready to engage \jd{in $n'$} and \jd{also} in $n''$, but $n''$ can not happen directly after $n$ 
because after $(n, \r)$ agent B is only ready to engage in $n'$. Thus $n'$ must happen before $n''$, and agent A is also required \jd{for} $n'$.

\je{The useless arc rule deletes useless arcs, which obviously preserves the occurrence sequences of the negotiation}. 
However, in general useless arcs may be very difficult to identify. \jd{So, instead of requiring in the guard of the rule 
that an arc is useless, we require a stronger, but easier to check condition}. 

\je{
\reductionrule{Useless arc rule}{def:useless}{
\begin{itemize}	
	\item
	There are three distinct arcs $(n,p,\r, n')$, $(n,p,\r, n'')$,  $(n,q,\r, n')$, such that $\trans(n,q,\r) = \{n'\}$.
      \item 
    The pre-negotiation obtained by removing the arc $(n,p,\r, n'')$ (see the action of this rule) is a negotiation, i.e., 
    satisfies the conditions of Definition \ref{def:neg}.
    \end{itemize}
}{
	$\trans(n,p,\r) \leftarrow \trans(n,p,\r) \setminus \{n''\}$.
}
}

\begin{figure}[h]
\centering
\input{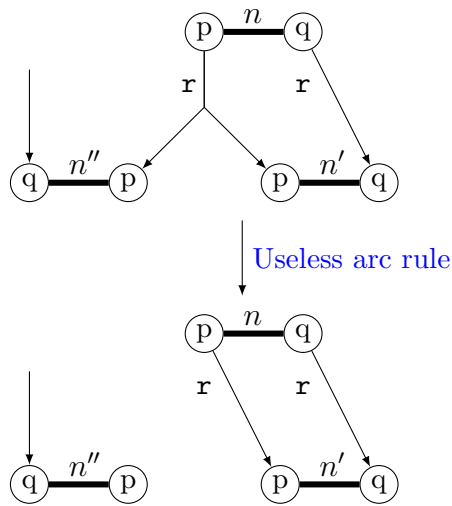}
\caption{Illustration of the useless arc rule}
\label{fig:uselessExample}
\end{figure}

Intuitively (see Figure \ref{fig:uselessExample}), after choosing
$\r$ at atom $n$, agent $p$ is ready to engage in both $n'$ and $n''$. However, another party $q$ of both $n'$ and
$n''$ is only willing to engage in $n'$. \jd{Then} agent $p$ \jd{can} never engage directly in $n''$ after
$n$, and so this part of the hyperarc can be removed without changing the behavior.

\jd{To see why the second part of the guard is necessary, recall that in a negotiation every atom lies on some path 
from the initial to the final atom. Without the second part, the result of applying the rule may no longer \ph{satisfy} this condition. 
For example, it is easy to \ph{define} a negotiation in which all paths from the initial to the final atom containing $n''$ 
traverse the arc $(n,p,\r, n'')$. We will show later that for the \ph{negotiations} of the classes of negotiations 
considered in this paper, the second part of the guard is always true and can be omitted.}

For the correctness proof we first specify what it means for an arc to
occur in an occurrence sequence.

\begin{definition}
An arc $(n,p,\r,n')$ of a negotiation {\em occurs} in an occurrence sequence
 $\sigma$ if $\sigma = \sigma_1 \, (n,\r) \, \sigma_2 \, (n', \r') \, \sigma_3$ for
some outcome $\r'$ of $n'$, and agent $p$ does not participate in any outcome of $\sigma_2$.
\end{definition}

\begin{lemma}
\label{lem:uselessRemoval}
\je{Let $(n,p,\r,n')$ be an arc of a negotiation $\N_1$ that does not occur in any initial occurrence sequence and belongs to a hyperarc with more than one arc. If
the result of removing this arc is a negotiation $\N_2$ then $\N_2$ is equivalent to $\N_1$.}
\end{lemma}
\begin{proof}
\jd{Assume that $\N_2$ is a negotiation. We proceed by showing that $\N_2$ is equivalent to $\N_1$.}

We claim that $\N_1$ and $\N_2$ have the same occurrence sequences. Since $\trans_2(n'',p'',\r'') \subseteq \trans_1(n'',p'',\r'')$ for every $(n'',p'',\r'') \in T(N_1)$,
every occurrence sequence of $\N_2$ is also an occurrence sequence of $\N_1$.
Since the transition functions only differ in $(n,p,\r)$, and the arc $(n,p,\r,n')$ never occurs in an initial
occurrence sequence of $\N_1$, every initial occurrence sequence of $\N_1$ is an initial occurrence sequence of $\N_2$, and the claim is proved.

A first immediate consequence of the claim is that an atom can occur in $\N_1$ if{}f it can occur in $\N_2$. 
It remains to show that every occurrence sequence of $\N_1$ can be extended to a large step if{}f the same holds for $\N_2$. Since 
both negotiations have the same occurrence sequences, we only have to show that an occurrence sequence is a large step of $\N_1$ if{}f it is a large step of $\N_2$. For this we observe that the markings $\vx_1$ and $\vx_2$ of $\N_1$ and $\N_2$ reached by executing an occurrence sequence $\sigma$ in both $\N_1$ and $\N_2$ can only differ with respect to agent $p$ and atom $n''$: If $\vx_1 \neq  \vx_2$, then 
$n',n'' \in  \vx_2(p)$, whereas $n' \in \vx_1(p)$ but $n'' \notin \vx_2(p)$. If $\sigma$ is a large step of $\N_1$ or $\N_2$ then $\vx_1(p) = \emptyset$ or $\vx_2(p) = \emptyset$, and so $\vx_1 = \vx_2$.
\end{proof}

\begin{theorem}
The useless arc rule is correct.
\end{theorem}
\begin{proof}
\je{We adopt the notations of Definition \ref{def:useless} and show that an arc removed by the rule satisfies the conditions of Lemma \ref{lem:uselessRemoval}. } 
After every occurrence of $(n, \r)$, agent $q$ is only ready to engage in $n'$. 
Therefore, after an occurrence of $(n, \r')$ the atom $n''$ must occur before $n'$ can occur. Thus the arc 
$(n, p, \r, n')$ can never occur in any initial occurrence sequence. 
\end{proof}





\subsection{Shortcut rule}

\jd{The  shortcut rule is inspired by the rule used in step (3) of
the summarization algorithm presented in Section \ref{subsec:sumalg}.} Intuitively, the goal of the rule is to
remove an outcome $(n, \r)$ such that, for
every large step $\sigma$ containing $(n, \r)$ \jd{of the original negotiation},
there is a large step $\sigma'$  \jd{of the reduced negotiation that does not contain $(n, \r)$,
such that}  $\trf{\sigma} = \trf{\sigma'}$.
Moreover, sometimes the rule must allow us to remove some successor atom of $(n, \r)$;
otherwise we  never reduce the number of atoms. Finally, this
must be achieved while preserving the sound \jd{or the} unsound character of a negotiation.

As a warm-up we first define the rule for the simple case of negotiations with
only one agent. Then we illustrate the problems involved in extending the rule to
arbitrary negotiations. Finally, we introduce the rule for the general case and provide its
correctness proof.

\jd{In the one-agent case} we can consider atoms and outcomes as nodes
and edges of a graph (see e.g. Figure \ref{fig:shortcut_easy1}).
We write $n \by{\r} n'$ instead of $\trans(n,p,\r) = \{n'\}$, where $p$ is the single agent of the negotiation.
Informally, the rule states:

\begin{itemize}
\item[(1)] If $n \by{\r} n'$ and $n' \neq n_f$, then
\begin{itemize}
\item add for every $n''$ such that $n' \by{\r'} n''$ a new {\em shortcut} $n \by{\r''} n''$
with transformer $\trf{n,\r''} = \trf{n,\r}\trf{n', \r'}$, and then remove $n \by{\r} n'$; 
\item if $n'$ has no other incoming edges, then remove $n'$, together with its outgoing edges.
\end{itemize}
\item[(2)] If $n \by{\r} n_f$, $\r$ is the only result of $n$ and $n_f$ has no other incoming edges, then
\begin{itemize}
\item add for every result $\r'$ of $n_f$  a new result $\r''$ of $n$  with $\trans (n,p,\r'') = \emptyset$
and transformer $\trf{n,\r''} = \trf{n,\r}\trf{n', \r'}$, and then remove $n \by{\r} n'$; 
\item remove $n_f$;
\item consider $n$ as the new final atom.
\end{itemize}
\end{itemize}

For example, in the negotiation of Figure \ref{fig:shortcut_easy1} we can
apply the rule to $n_2 \by{d} n_3$, obtaining the negotiation of
Figure \ref{fig:shortcut_easy3}, and then to $n_1 \by{c} n_3$, yielding
the negotiation of Figure \ref{fig:shortcut_easy5}. 

To see why we require $n'\neq n_f$ \jd{in the guard of part (1) of the rule}, 
\jd{consider} Figure \ref{fig:shortcut_easy5}. Without this condition, the application of the 
rule to $n_1 \by{\mathtt{c}'} n_f$ replaces the result $\mathtt{c}'$ by a new result, say $\mathtt{c}''$, 
such that $\trans(n_1,p,\mathtt{c}'') = \emptyset$ (graphically, the arc from $n_1$ to $n_f$ disappears).
The result is not even a negotiation, because in a negotiation all atoms but the final one must have at least 
one outgoing edge.

Part (2) of the rule allows us to apply \ph{the rule} to $n_3 \by{\mathtt{e}} n_f$ in Figure \ref{fig:shortcut_easy1}. 
Since $n_f$ has no other incoming edge, it is removed. Atom $n_3$ loses its only outgoing edge, and becomes 
the new final atom. To preserve equivalence, the results of $n_f$ are transferred to the new final atom $n_3$.


\begin{figure}[ht]
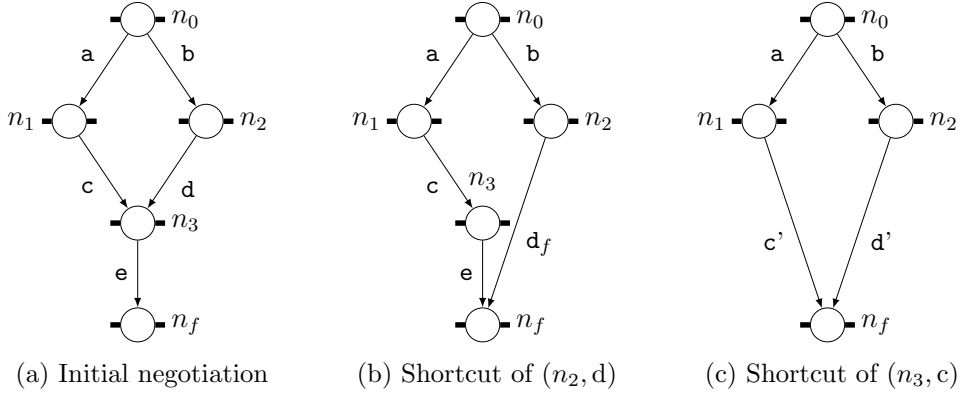

\centering
\begin{subfigure}[c]{0.3\textwidth}
	\centering
	\scalebox{0.9}{
		\input{tikz/shortcut_easy/shortcut_easy1}
	}
	\subcaption{Initial negotiation}
	\label{fig:shortcut_easy1}
\end{subfigure}
\begin{subfigure}[c]{0.3\textwidth}
	\centering
	\scalebox{0.9}{
		\input{tikz/shortcut_easy/shortcut_easy3}
	}
	\subcaption{Shortcut of $(n_2, {\rm d})$}
	\label{fig:shortcut_easy3}
\end{subfigure}
\begin{subfigure}[c]{0.3\textwidth}
	\centering
	\scalebox{0.9}{
		\input{tikz/shortcut_easy/shortcut_easy5}
	}
	\subcaption{Shortcut of $(n_3, {\rm c})$}
	\label{fig:shortcut_easy5}
\end{subfigure}
\caption{Shortcut for one agent}
\label{fig:shortcut_easy}
\end{figure}

\begin{figure}[ht]
\centering
	\centering
	\scalebox{0.8}{
		\input{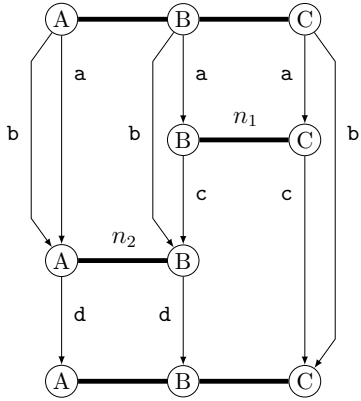}
        }
        \label{fig:short2_1}
\caption{Searching for a shortcut rule}
\label{fig:short2}
\end{figure}

In the case of multiple agents, \jd{defining} a correct shortcut rule is more complicated.
Consider the negotiation of Figure \ref{fig:short2}. 
\jd{The large steps of this negotiation are} 
$$(n_0,\mathtt{a}), (n_1, \mathtt{c}), (n_2,\mathtt{d}), (n_f,\mathtt{f})$$ where \ph{$ \mathtt{f}$} is a final result, and 
$$(n_0,\mathtt{b}), (n_2,\mathtt{d}), (n_f,\mathtt{f}).$$
\jd{So, if $(n_1, \mathtt{c})$ occurs, it is succeeded by $(n_2,\mathtt{d})$. However,}
there is no \jd{obvious}
way to remove the outcome $(n_1, \mathtt{c})$ by shortcutting it with its
``successor'' $(n_2, \mathtt{d})$, \jd{in particular because the latter has an additional participant}.
For this reason, we 
\jd{apply the shortcut rule only} to outcomes $(n, \r)$ such
that, after executing them, an atom $n'$ is enabled, and remains enabled until it occurs.
\jd{This condition apparently depends on the markings enabling $(n, \r)$.
Since there is no efficient way to identify the markings reachable from the initial marking 
and since the conditions for applying the rule must be easy to verify, 
we require that any marking, reachable or not, that enables $(n, \r)$, subsequently enables $n'$ until it occurs.}
For example, in Figure \ref{fig:short2} we can
shortcut $(n_0, \texttt{a})$, because after its occurrence the atom $n_1$ is  enabled,
and can 
become disabled only by the occurrence of $n_1$. 
However, we cannot shortcut $(n_1,\texttt{c})$, because the marking that places a token on each of the \jd{ingoing arcs} of $n_1$ does enable $n_1$, but after $(n_1,\texttt{c})$ occurs no other atom is enabled.

\jd{The following simple structural condition formulates an obvious sufficient condition of the property mentioned before. }



\begin{definition}[Unconditionally enables]
\label{def:uncond}
An outcome $(n,\r)$  {\em unconditionally enables} an atom $n'$
if $P_n \supseteq P_{n'}$ and $\trans(n, p, \r) = \{n'\}$ for every $p \in P_{n'}$.
\end{definition}

We now consider a second problem. If a negotiation is sound \jd{before application of the shortcut 
rule to an outcome $(n, \r)$ and an atom $n'$}, and after 
the shortcut the atom $n'$ \jd{is not enabled at any reachable marking}, then we must remove $n'$, otherwise
the negotiation becomes unsound. 
For example, if in the sound negotiation of Figure \ref{fig:shortcut_easy} we do not 
remove $n_3$ after shortcutting $n_1 \by{c} n_3$, we are left with an unsound negotiation.
In the one-agent case, the atom $n'$ can be enabled after the shortcut
if{}f it still has incoming edges. However, in the case of multiple agents, 
deciding whether $n'$ must  be removed or not can be much harder. Consider the sound negotiation of
Figure \ref{fig:shortcut_intro1}. The only result of $n_2$
enables $n_3$ unconditionally. After shortcutting this result 
the atom $n_3$ can never become enabled, and must be removed. However, after the shortcut
all ports of $n_3$ still have incoming edges,
\jd{ and so this criterion does not suffice to decide whether the atom must be removed}.
In general, deciding if $n'$ must be removed 
is an intractable problem \jd{(observe that it is equivalent to deciding whether these arcs are useless)}. 
\jd{Fortunately, it suffices to consider 
two cases for which the decision is simple, and to apply the shortcut rule only in these cases}. 

First, as in the one-agent case, if after the shortcut no edge leads to \jd{a port of} $n'$, then 
$n'$ cannot be enabled, and must be removed. \jd{This is the case if  all ports have, before applying the rule,
only one ingoing arc, namely the one from $n$ labeled by $\r$. In this case we say that $(n,\r)$ has exclusive access to $n'$.}

\begin{definition}[Exclusive access]
\label{def:depends}
An outcome  $(n, \r)$ has {\em exclusive access} to an atom $n'$ if,  for each $p \in P_{n'}$, 
$n'\in \trans(n,p,\r)$  and $n'\notin \trans(n'',p,\r'')$ for $(n,\r)\neq (n'',\r'')$.
\end{definition}
\jd{So, if after the shortcut \ph{all} ports of $n'$ have no ingoing arc, we can remove $n'$.}
Assume now that after the shortcut some {\em deterministic} arc leads to a port of $n'$ \jd{(that is,
an arc which is not a proper hyperarc)}. Say this arc corresponds to an outcome $(n'', \r'')$.
As we shall see, 
if the negotiation was sound before the shortcut, then $n''$ can still be enabled after it,
\jd{the outcome $(n'', \r'')$ can occur, and $n'$ is the only atom removing the token 
on the above deterministic arc. So keeping $n'$ preserves the soundness in this case.
If the negotiation was unsound before, then keeping $n'$ will obviously not make the resulting negotiation sound.}
So in both cases, if $n'$ is not removed
the sound/unsound character of the negotiation is preserved.

\jd{If $(n,\r)$ unconditionally enables $n'$ and has exclusive access to $n'$ then a shortcut is possible and $n'$ has to be removed.
}

\begin{definition}[Commits to]
\label{def:commits}
An outcome $(n'', \r'')$ {\em commits to} an atom $n'$ if $\{n'\} = \trans (n'', p, \r'')$ 
for some $p \in P_{n''}$.
\end{definition}

\jd{As in the case of a single agent, we can only have $n' = n_f$ if $n'$ is removed and $n$ qualifies as a new final atom.
Therefore, we require $n' \neq n_f$  unless $(n,\r)$ has exclusive access to $n'$ and moreover $\r$ is the only result of $n$.}

Equipped with Definitions \ref{def:uncond}, \ref{def:depends} and \ref{def:commits} we can now formally
define the general shortcut rule, illustrated in Figure \ref{fig:shortcutExample}.

\begin{figure}[ht]
\centering
\input{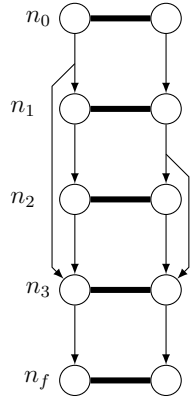}
\caption{Problematic case for a shortcut}
\label{fig:shortcut_intro1}
\end{figure}



\reductionrule{Shortcut rule}{def:shortcut}{
        $N$ contains an outcome $(n, \r)$ that unconditionally enables an atom $n'$ distinct from $n$. Moreover: 
        \vspace{-0.9ex}
        \begin{itemize}
        \itemsep-0.5ex
        \item $n' \neq n_f$ and $(n, \r)$ has exclusive access to $n'$,  or 
        \item $n' \neq n_f$ and some outcome $(n'', \r'') \neq (n, \r)$ commits to $n'$, or
        \item $n' = n_f$, the outcome $(n, \r)$ has exclusive access to $n'$, and $\r$ is the only result of $n$.
       
        \end{itemize}
}{
	\begin{enumerate}[(1)]
	\itemsep-0.5ex
	\item $R_n \leftarrow (R_n \setminus \{\r\}) \cup \{\r'_s \mid \r' \in R_{n'}\}$, where $\r'_s$ are fresh names.
	\item For all $p \in P_{n'}$ and all $\r' \in R_{n'}$: $\trans(n,p,\r_s') \leftarrow \trans(n',p ,\r')$. \\
 For all $p \in P_n \setminus P_{n'}$ and all $\r' \in R_{n'}$: $\trans(n,p,\r'_s) \leftarrow \trans(n,p,\r)$.
	\item For all $\r' \in R_{n'}$: $\delta_n (\r'_s)\leftarrow \delta_n (\r)\delta_{n'} (\r')$.
	\item If $(n,\r)$ has exclusive access to $n'$, then remove $n'$. 
	\item If $n'$ was the final atom before applying the rule and is removed, then $n$ is the new final atom.
	\end{enumerate}
}

\begin{figure}[h]
\centering
\input{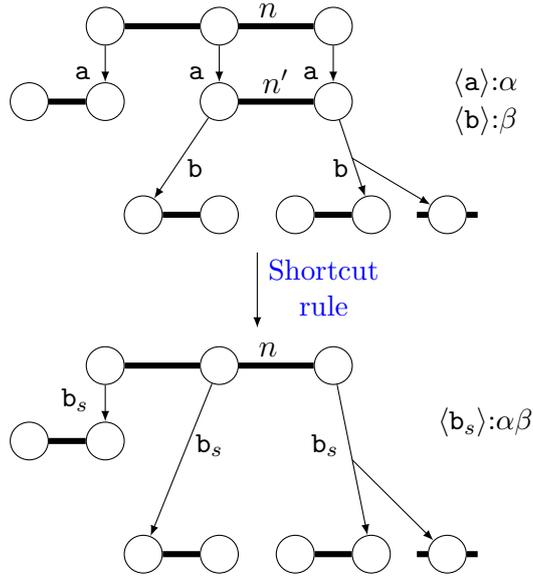}
\caption{Illustration of the shortcut rule}
\label{fig:shortcutExample}
\end{figure}

Consider the negotiation of Figure \ref{fig:short2}. We can shortcut result $\texttt{a}$ of $n_0$
because it unconditionally enables $n_1$ and has exclusive access to it. We can also shortcut 
result $\texttt{b}$, because it unconditionally enables $n_2$ and result $\texttt{a}$ of
$n_0$ (or result $\texttt{c}$ of $n_1$) commits to $n_2$. 

Consider now Figure \ref{fig:shortcut_intro1}. We cannot shortcut the results of $n_0$ and $n_1$ because they
do not unconditionally enable any node. The result of $n_2$ unconditionally enables $n_3$, 
but we cannot shortcut it either: it does not have exclusive access to $n_3$, and no other 
result commits to $n_3$. The result of $n_3$ unconditionally enables $n_f$ with exclusive access, 
and so we can shortcut it. The application of the rule removes $n_f$, and $n_3$ becomes the new final 
atom of the negotiation.


\begin{theorem}
\label{thm:shortcutcorrect}
The shortcut rule is correct.
\end{theorem}
\begin{proof}
Let $\N_2$ be the result of applying the shortcut rule to $\N_1$. Assume the atoms $n$ and $n'$ and the result $\r \in R_n$ are as in the shortcut rule (Definition \ref{def:shortcut}). 
\jd{We first show that $\N_2$ is in fact a negotiation.}

\begin{itemize}
\item
$\trans (n'', p , \r'') = \emptyset$ if and only if $n'' = n_f$ in $\N_2$. This is obviously the case if 
$n'$ was not the final atom before the reduction. If it was the final atom, then we only apply the rule if we also remove $n'$, and if moreover $n$
has only the single result $\r$. So it remains to show that in this case $n$ satisfies the criteria of a final atom. 

We have $P_{n'} \subseteq P_{n}$ because $(n,\r)$ unconditionally enables $n'$. We have $P_{n'} = \agents$ because $n' = n_f$, and therefore $P_{n} = \agents$.
By definition of the rule, $\trans (n, p , \r) = \emptyset$ after the rule was applied, for all agents $p$.
Since $\r$ is the only result, we obtain
 $\trans (n, p , \r'') = \emptyset$ for all results $\r''$.
So $n$ qualifies as new final atom. The transition function $\trans$ is not changed by the rule application for any other atom.

\item
Every atom $n''$ of $\N_2$ is on a path from the initial atom to the final atom in the graph of the negotiation.
Consider first the case $n'' \neq n'$. A path of $\N_2$ that visits $n''$ is either a path of $\N_1$, or
it is obtained from a path of $\N_1$ by replacing $(n, p, \r) (n', p', \r')$ by $(n, p', \r'_s)$.
So, no matter if $n'$ is the final atom or not, $n_2$ lies on some path of $\N_2$ from $n_0$ to $n_f$. 

Assume now $n''=n'$. Then $n'$ has not been removed by the rule, and so $(n,r)$ does not have exclusive access to $n'$.
This implies that some party $p$ of $n'$ has an additional input arc, i.e., that $p$ is in $\trans (n''', p , \r''')$ for some $(n,\r) \neq (n''', \r''')$.
There is a path from $n_0$ via  $n'''$ to $n'$, before and after application of the rule.
Moreover, the path from $n'$ to $n_f$ remains unchanged by the rule application (assuming that $n'$ does not occur again in the path).
So, also after application of the rule, $n'$ is on a path from $n_0$ to $n_f$.
\end{itemize}

We \jd{continue}  by studying the relation between occurrence sequences in $\N_1$ and $\N_2$. We introduce the concept of corresponding sequences which we then use to show that soundness of $\N_1$ implies soundness of $\N_2$ and \jd{vice versa}. We say that an atom $\overline{n}$ occurs in an occurrence sequence $\sigma$ if $(\overline{n}, \overline{\r})$ occurs in $\sigma$ for some $\overline{\r} \in R_{\overline{n}}$.

For each initial occurrence sequence $\sigma_2$ of $\N_2$, replacing all occurrences of $(n, \r'_s)$ by $(n,\r), (n',\r')$ yields an initial occurrence sequence $\sigma_1$ of $\N_1$. We call $\sigma_1$ the occurrence sequence in $\N_1$ corresponding to $\sigma_2$. The markings reached by these two sequences are the same.

For each initial occurrence sequence $\sigma_1$ of $\N_1$, between each two occurrences of $(n,\r)$ there is an occurrence of $(n',\r')$ (with possibly varying $\r'$) because $n$ unconditionally enables $n'$. Thus only the last occurrence of $(n,\r)$ might not be followed by an occurrence of $n'$, which may happen only if $n'$ is still enabled.
We transform $\sigma_1 = \sigma_a(n,\r)\sigma_a'(n',\r_a)\sigma_b(n,\r)\sigma_b'(n',\r_b)\ldots$ to a sequence $\sigma_2=\sigma_a(n,\r_{af})\sigma_a'\sigma_b(n,\r_{bf})\sigma_b'\ldots$ \jd{of} $\N_2$. Should, \jd{in $\sigma_1$}, the last occurrence of $(n,\r)$ not be followed by an occurrence of $n'$, we instead apply the transformation to \jd{the extended sequence} $\sigma_1(n',\r')$ for some $\r'$. \jd{In both cases, this} yields an occurrence sequence $\sigma_2$ \jd{of} $\N_2$, which we call the corresponding sequence to $\sigma_1$, and which leads to the same marking as $\sigma_1$ (or $\sigma_1(n',\r')$).

(1) Assume that $\N_1$ is sound. We prove that $\N_2$ is sound. 

First we show that every atom $\tilde{n}$ of $\N_2$ can be enabled by some initial occurrence sequences, and then that 
every initial occurrence sequence $\sigma_2$ in $\N_2$ can be extended to a large step.
For the first part, we consider two cases.

Case $\tilde{n}\neq n'$. Since $\N_1$ is sound, $\tilde{n}$ can be enabled by some 
initial occurrence sequence $\sigma_1$ of $\N_1$. Since the corresponding sequence $\sigma_2$ in $\N_2$ 
leads to the same marking as either $\sigma_1$ or $\sigma_1(n',\r')$ for some $\r'$, 
it also enables $\tilde{n}$. (Note that the addition of $(n',\r')$ to the occurrence sequence cannot influence that $\tilde{n}$ is enabled: We only add $(n',\r')$ if it was unconditionally enabled by some prior occurrence of $(n,\r)$, but then the agents of $n'$ were only ready to engage in $n'$ and thus do not participate in $\tilde{n}$.)

Case $\tilde{n}= n'$.
Since $n'$ is still an atom of $\N_2$ and by the definition of the guard of the rule, 
some outcome $(n'',\r'')\neq (n,\r)$ of $\N_1$ commits to $n'$, i.e., there
is $(n'',p'',\r'')\in T(\N_1)$ such that $(n'',\r'')\neq (n,\r)$ and $\{n'\}=\trans(n'',p'',\r'')$. 
Since $(n'',\r'')\neq (n,\r)$, the outcome $(n'',\r'')$ is unchanged in $\N_2$. By soundness of $N_1$, 
some initial occurrence sequence $\sigma_1$ of $\N_1$ enables $n''$. We extend it by an occurrence of $(n'',\r'')$ and then to a large step $\sigma_1(n'',\r'')\rho_1$ in $\N_1$, which is possible since $\N_1$ is sound. Since after $(n'',\r'')$ agent $p''$ is only ready to engage in $n'$, the sequence $\rho_1$ contains an occurrence of $n'$. During the construction we replace occurrences of $(n,\r)$ and those occurrences of $n'$ that are direct consequences of the occurrence of some $(n,\r)$. We will however not replace the occurrence of $n'$ that was caused by $(n'',\r''$), and so the corresponding sequence in $\N_2$ will also contain an occurrence of $n'$. Thus $n'$ can be enabled in $\N_2$.

We now prove that every initial occurrence sequence $\sigma_2$ in $\N_2$ can be extended to a large step. 
Take the corresponding occurrence sequence $\sigma_1$ in $\N_1$ and extend it to a large step 
$\tau_1 = \sigma_1\rho_1$ in $\N_1$ (possible by soundness of $\N_1$). The corresponding sequence in $\N_2$ 
is $\tau_2=\sigma_2\rho_2$, which is an extension of $\sigma_2$ (by construction of corresponding sequences) to a 
large step.

(2)  Assume that $\N_2$ is sound. We prove that $\N_1$ is sound. 

Since every atom in $\N_2$ can be enabled by some initial occurrence sequence, using the corresponding sequence in 
$\N_1$ we see that the same is true for all atoms but $n'$. Further, it is easy to show that also $n'$ can be enabled 
in $\N_1$: Take the initial occurrence sequence that enables $n$ and extend it by $(n,\r)$, which unconditionally enables $n'$.

For an initial occurrence sequence $\sigma_1$ in $\N_1$, the corresponding occurrence sequence $\sigma_2$ 
in $\N_2$ can be extended to a large step $\tau_2 = \sigma_2\rho_2$ in $\N_2$, \jd{because $\N_2$ is sound}.
\jd{Let $\tau_1$ be the occurrence sequence of $\N_1$ that  corresponds to $\tau_2$. 
Then $\tau_1 = \sigma'_1\rho'_1$ where $\sigma'_1$ corresponds to $\sigma_2$, by the definition of ``corresponds''.
Now either $\sigma_1$ or $\sigma_1 (n',\r') $ leads to the same marking as $\sigma_2$, and $\sigma_2$ leads to the
same marking as $\sigma'_1$. Therefore, either $\sigma_1 \rho'_1$ or $\sigma_1 (n',\r') \rho'_1 $ is a large step of $\N_1$.}
\end{proof}
\subsection{Rules Preserve (Weak-)Determinism}

We conclude the section with a simple but important observation.

\begin{proposition}
If a negotiation $\N_1$ is deterministic (weakly deterministic) and the application of one of the rules 
above yields negotiation $\N_2$, then $\N_2$ is also deterministic (weakly deterministic, \jd{respectively}). Further, if $\N_1$ is
acyclic, then so is $\N_2$.
\end{proposition}

\begin{proof} 
For the merge rule and the iteration rule the proposition is obvious. 
The useless-arc rule removes an arc and thus may even lead to a deterministic negotiation starting from a 
weakly deterministic one (but never the other way round). For the shortcut rule, observe 
that, by the definition of the rule, for every $(n_2,a_2,\r_2) \in T(\N_2)$ there is $(n_1,a_1,\r_1) \in T(\N_1)$ such that 
$\trans_2(n_2,a_2,\r_2) = \trans(n_1,a_1,\r_1)$ (in fact, we can always take $a_1= a_2$). 
Since whether a negotiation $\N$ is deterministic or weakly deterministic only depends on the 
set $\{ \trans(n,a,\r) \mid (n,a,\r) \in T(\N)\}$, we are done.
Finally, \jd{the definitions of the rules imply immediately that acyclicity is preserved (but not cyclicity).}
\end{proof}

\subsection{Reducible outcomes}
In the rest of the paper we present reduction algorithms for different classes of negotiations.
All of them repeatedly choose a {\em reducible} outcome, and apply the corresponding rule. Formally:

\begin{definition}
\label{def:reducibleoutcome}
Let $\N$ be a negotiation. An outcome $(n, \r)$ of $\N$ is {\em reducible} if 
\begin{itemize}
\item it satisfies the guard of the iteration or the shortcut rule; or
\item it satisfies, together with another outcome, the guard of the merge rule; or
\item $(n, p, \r, n'')$ is a useless arc for some agent $p$ and atom $n''$.
\end{itemize}
The set of reducible outcomes of $\N$ is denoted $R(\N)$.
\end{definition}

As we shall see, a nondeterministic reduction procedure that can choose any reducible outcome may
not terminate, or terminate in exponentially many steps. In order to obtain polynomial algorithms we
will constrain the choices of the procedure.

\section{Summarizing Acyclic Weakly Deterministic Negotiations}
\label{sec:acyclic}

Recall that if a set of rules is complete for a class of negotiations, then it completely reduces 
all (and only) sound negotiations in the class to an atomic negotiation. Therefore, a reduction
procedure that transforms every negotiation in the class into an equivalent {\em irreducible} negotiation 
allows us to check soundness and compute a summary without constructing the 
reachability graph of the negotiation. Indeed, the negotiation is sound if{}f its corresponding
irreducible negotiation is atomic. 

\begin{definition}
\label{def:maximal}
Let ${\cal R}$ be a set of rules and let $\N$ be a negotiation.

A negotiation is {\em irreducible} with respect to ${\cal R}$ if no rule of ${\cal R}$ can be applied to it. 

A \jd{finite or infinite  reduction sequence $R_1 \, R_2 \,  R_3 \ldots $ for $\N$ 
over ${\cal R}$ is {\em maximal} if it is, applied to $\N$, either infinite or if it is 
finite and cannot be extended, i.e.}  leads to an irreducible negotiation. 
\end{definition}

In this section we show that the rules presented in the previous section are complete for 
acyclic, weakly deterministic negotiations. This class contains for example the 
\texttt{Father} / \texttt{Daughter} / \texttt{Mother} negotiation from the beginning of this paper.
Further, we show in Section \ref{sec:acyclictime} that if the negotiation is deterministic, then there is a
reduction algorithm that reaches an irreducible negotiation after a polynomial number of rule 
applications, and \ph{thus}  runs in polynomial time. Whether such an algorithm also exists in the weakly deterministic case 
is an open problem.

Acyclicity is a severe restriction, which in particular implies that every atom occurs at most once. 
However, the results on acyclic negotiations will be reused as lemmas in the following sections, 
where we consider cyclic deterministic negotiations. 

We start with a first lemma, showing that for acyclic negotiations the guard of the useless arc rule
can be simplified.

\begin{lemma}
\label{lemma:uselessacyclic}
Let $\N$ be an acyclic negotiation. $\N$ can be transformed into $\N'$ by the useless arc rule if{}f it can be transformed
into $\N'$ by the following rule:\\
\reductionrulenodef{
\begin{itemize}	
\item There are three distinct arcs $(n,p,\r, n')$, $(n,p,\r, n'')$,  $(n,q,\r, n')$, such that $\trans(n,q,\r) = \{n'\}$.
\item The arc $(n,p,\r, n'')$ is not the only arc leading to \ph{a port of} $n''$.
\end{itemize}
}{
$\trans(n,p,\r) \leftarrow \trans(n,p,\r) \setminus \{n'\}$.
}
If moreover $\N$ is sound, then $\N$ can be transformed into $\N'$ by the useless arc rule if{}f it can be transformed
into $\N'$ by the following rule:\\
\reductionrulenodef{
There are three distinct arcs $(n,p,\r, n')$, $(n,p,\r, n'')$,  $(n,q,\r, n')$, such that $\trans(n,q,\r) = \{n'\}$.
}{
$\trans(n,p,\r) \leftarrow \trans(n,p,\r) \setminus \{n'\}$.
}

\end{lemma}
\begin{proof}
\noindent 
Assume that $\N$ is acyclic.

($\Leftarrow$) If $\N$ can be transformed into $\N'$ by the useless arc rule, then after removing the useless arc 
$(n, p, \r, n'')$ the node $n''$ still lies on a path from the initial to the final node, and so $\N$ contains another arc leading to $n''$. So 
$\N$ can also be transformed into $\N'$ by the new rule. 

\noindent ($\Rightarrow$) Roughly speaking,
we have to show that after applying the new rule to an acyclic negotiation, the resulting pre-negotiation is actually a negotiation, 
i.e., every node is still on a path from $n_0$ to $n_f$ in the  graph of the resulting pre-negotiation. 
Obviously, all paths which do not contain an arc leading from $n$ to $n''$ are not
changed by the rule application. So it suffices to show that that every node of $\N$ lies 
on some path from $n_0$ to $n_f$ which does not contain such an arc.

Consider an arbitrary atom $\tilde{n}$. There is a path of $\N$ from $n_0$ to $\tilde{n}$. 
By assumption there is an arc leading from some atom $n_1$ to $n''$.
If the path from $n_0$ to $\tilde{n}$ contains an arc from $n$ to $n''$, then we can replace the prefix 
up to $n$ by a path from $n_0$ to $n_1$, yielding a path from $n_0$ to $\tilde{n}$ without the arc.
By acyclicity, this path does not contain $n''$.
If a path from $\tilde{n}$ to $n_f$ contains an arc from $n$ to $n''$ then we can replace the suffix from $n''$ 
by a path from $n'$ to $n_f$. Again by acyclicity, this path does not contain $n$. 
Summarizing, after the rule application the node $\tilde{n}$ is again on a path from $n_0$ to $n_f$, and we are done.

Assume now that $\N$ is acyclic and \ph{moreover} sound.

\noindent ($\Leftarrow$) Obvious.

\noindent ($\Rightarrow$) \jd{If the negotiation is sound, then the atom $n''$ definitely has more ingoing arcs 
than the useless one, and we can apply the first part of the lemma.} 
\end{proof}

\begin{theorem}
\label{thm:compacyc}
Let $\N$ be an acyclic and weakly deterministic negotiation, and let $\rho$ be a maximal
reduction sequence for $\N$ using \jd{only} the shortcut, merge  and useless arc rules. 
Then $\rho$ is finite, and it reduces $\N$ to an atomic negotiation if{}f $\N$ is sound.
\end{theorem}
\begin{proof}
We first prove that $\rho$ is finite. It suffices to show that none of the shortcut, merge, 
and useless arc rules can be applied infinitely often in a reduction sequence.
Let $\N_2$ be the result of applying any of the rules to a 
negotiation $\N_1$. For every large step $\sigma$ of $\N_2$, let $\phi(\sigma)$ 
be defined as follows: if the useless arc rule or the merge rule
has been applied, then $\phi(\sigma)=\sigma$; if the shortcut rule has been applied
to atoms $n, n'$, then let $\phi(\sigma)$ be the result of replacing every occurrence of $(n,\r_f')$ by 
the sequence $(n,\r)(n',\r')$. We have $|\sigma| \leq |\phi(\sigma)|$ for every large step $\sigma$. Moreover,
if the shortcut rule is applied, then $|\sigma'| < |\phi(\sigma')|$ for at least one large step $\sigma'$, 
indeed for all large steps containing $(n,\r_f')$. 

\jd{None of the rules destroy acyclicity, i.e., all considered negotiations are acyclic.
In an acyclic negotiation, the length of a large step is bounded by the number of atoms, 
and so the set of large steps is finite.
Therefore,} no infinite sequence of 
applications of the rules can contain infinitely many applications of the shortcut rule.
So any infinite sequence of applications must, from some point on, only apply the useless arc rule 
and the merge rule. But this is also not possible as both rules reduce the number of arcs.
\jd{This completes the proof that $\rho$ is finite.}

Assume that $\N$ is not sound. Since the three rules are correct, $\rho$ reduces $\N$ to an
unsound negotiation $\N'$. By definition of soundness,  atomic negotiations are sound. Therefore, 
$\N'$ is not atomic.

Assume now that $\N$ is sound.  We prove that $\rho$ reduces $\N$ to an atomic negotiation.  
The proof has two parts:

\noindent (1) If $\N$ has more than two atoms, then the shortcut rule or the useless arc rule
can be applied to it.

Since $\N$ is acyclic, its graph induces a partial order on atoms in the obvious way ($n < n'$ if there is a path from $n$ to $n'$).
Clearly $n_0$ and $n_f$ are the unique minimal and maximal elements, respectively.
We choose an arbitrary linearization of this partial order.
Since $\N$ has more than two atoms, this linearization begins with $n_0$ and has some second element, say $n_1 \neq n_f$.

Since $\N$ is sound \jd{and since $n_1$ does not depend on any prior atom except $n_0$}, some occurrence sequence begins with an occurrence of $n_0$ and a subsequent occurrence of $n_1$.
So $n_0$ has a result $\r_0$ such that $(n_0, \r_0)$ unconditionally enables $n_1$. 

Consider two cases:

\begin{itemize}
\item $\{n_1\} \subsetneq \trans(n_0,p,\r_0)$ for some party $p$ of $n_1$.\\
\je{Then $n_1, n_2 \in \trans(n_0,p,\r_0)$ for some atom $n_2 \neq n_1$. 

Since $\N$ is weakly deterministic, some deterministic agent $q$ must be a party of all atoms in $\trans(n_0,p,\r_0)$, and so in particular of $n_1$ and $n_2$. 

Since $(n_0, \r_0)$ unconditionally enables $n_1$ and $q$ is deterministic, $\trans(n_0,q,\r_0)=\{n_1\}$. 

So we have three distinct arcs $(n_0,p,\r_0, n_1)$, $(n_0,p,\r_0, n_2)$, and $(n_0,q,\r_0, n_1)$ such that $\trans(n_0,q,\r_0)=\{n_1\}$. Since $\N$ is sound, by Lemma \ref{lemma:uselessacyclic} the useless arc rule can be applied
with $n := n_0$, $n':= n_2$, $n'':= n_1$.}

\item $\{n_1\} = \trans(n_0,p,\r_0)$ for every party $p$ of $n_1$.\\
\je{We claim that the shortcut rule or the useless arc rule can be applied. 
Assume $\N$ does not satisfy the guard of the shortcut rule. Then, since $(n_0, \r_0)$ 
unconditionally enables $n_1$ and $n_1 \neq n_f$, we have that (a) $(n_0, \r_0)$ does not have exclusive 
access to $n_1$, and (b) no outcome distinct from $(n_0, \r_0)$ commits to $n_1$. 

By (a), there exists an arc $(n, p, \r, n_1)$ such that $(n, \r) \neq (n_0, \r_0)$. 
Since $n_1$ is the second atom in the linearization, we have $n=n_0$ and $\r \neq \r_0$.
So there exists an arc $(n_0, p, \r, n_1)$.

By (b), there is an arc $(n_0, p, \r, n_2)$ for some atom $n_2 \neq n_1$

Since $\N$ is weakly deterministic, some deterministic agent $q$ is a party of all 
atoms in $\trans(n_0,p,\r')$. Moreover, by soundness, the unique atom in $\trans(n_0, q, \r)$ must
be one of them. We can assume without loss of generality that the atom is 
either $n_1$ or $n_2$.

If the atom is $n_1$, then we have three distinct arcs $(n_0, p, \r, n_1)$, $(n_0, p, \r, n_2)$, 
and $(n_0, q, \r, n_1)$ such that $\trans(n_0,q,\r) = \{n_1\}$. By Lemma \ref{lemma:uselessacyclic}, the useless arc rule can be 
applied with $n:=n_0$, $n':=n_2$, and $n'':=n_1$.

If the atom is $n_2$, then we have three distinct arcs $(n_0, p, \r, n_1)$, $(n_0, p, \r, n_2)$, 
and $(n_0, q, \r, n_2)$ such that $\trans(n_0,q,\r) = \{n_2\}$. Further, since $(n_0, \r_0)$ 
unconditionally enables $n_1$, there is a path from $n_0$ to
$n_1$ that does not contain the arc $(n_0, p, \r, n_1)$.  By Lemma \ref{lemma:uselessacyclic} the useless arc rule can be applied 
with $n:=n_0$, $n':=n_1$, and $n'':=n_2$.}
\end{itemize}

\noindent (2) If $\N$ has exactly two atoms, it can be summarized (i.e., reduced to an equivalent atomic negotiation). \\
In this case, the atoms of $\N$ are $n_0$ and $n_f$. Since the negotiation is sound, all results of 
$n_0$ have the same transition function. So 
the merge rule can be applied until $n_0$ only has one result. Then the shortcut rule can be applied with 
$n = n_0$ and $n' = n_f$ to obtain an equivalent atomic negotiation.\\

We can now conclude the proof.
By the maximality of $\rho$ and (1), $\rho$ reduces $\N$ to a negotiation containing at most two atoms.
By (2), $\sigma$ reduces $\N$ to an atomic negotiation. 
\end{proof}

As a corollary of Theorem \ref{thm:compacyc} we obtain \jd{the} completeness result:

\begin{corollary}
\label{thm:completeness1}
The shortcut rule, merge rule and useless  arc rules are complete for acyclic, weakly 
deterministic negotiations.
\end{corollary}

\begin{figure}[p]
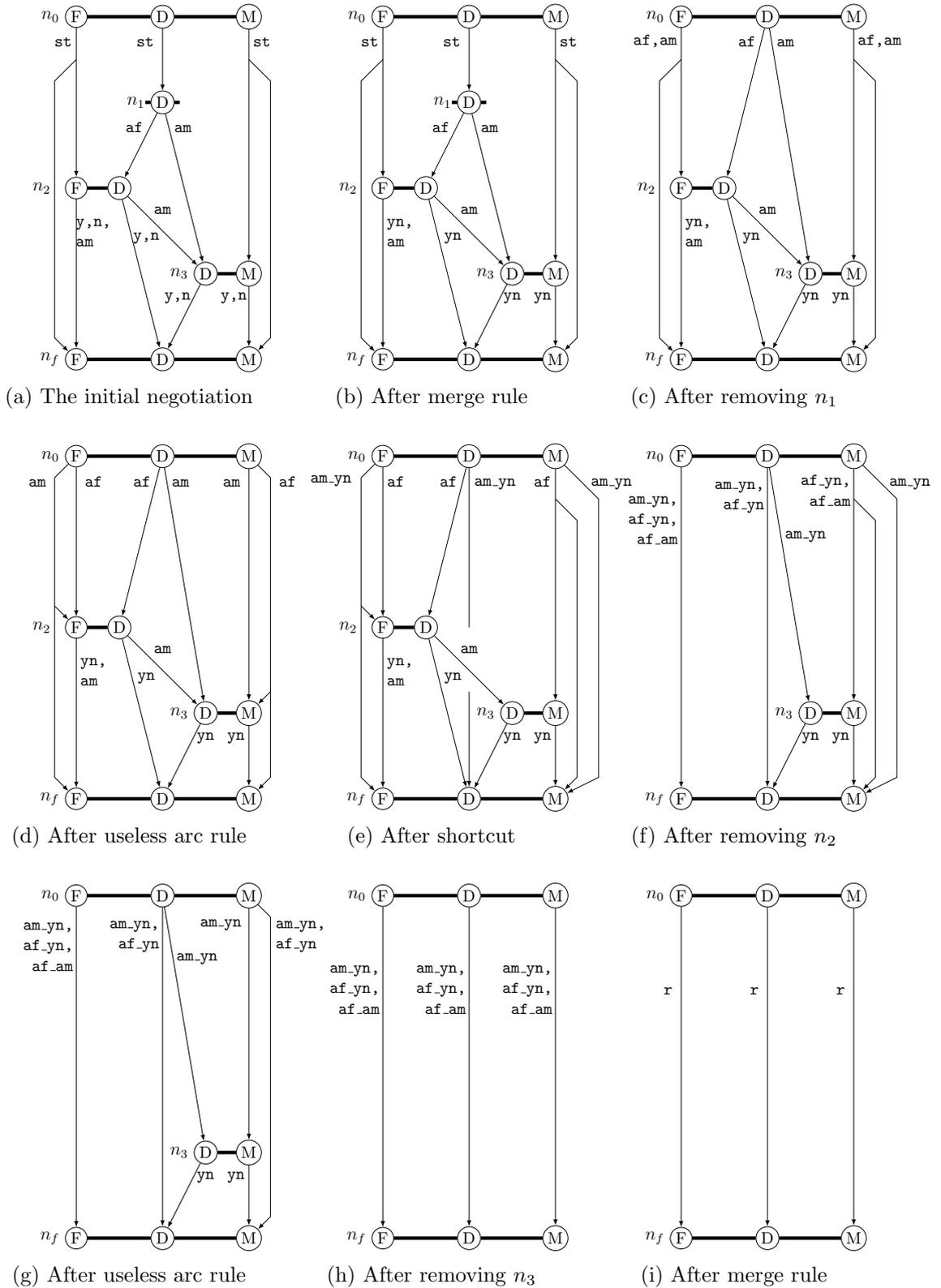

\centering
\begin{subfigure}[c]{0.32\textwidth}
	\scalebox{0.7}{
		\input{tikz/wdetExample/1}
	}
	\subcaption{The initial negotiation}
	\label{fig:summary:init}
\end{subfigure}
\begin{subfigure}[c]{0.33\textwidth}
	\scalebox{0.7}{
		\input{tikz/wdetExample/2}
	}
	\subcaption{After merge rule}
	\label{fig:summary:afterMerge}
\end{subfigure}
\begin{subfigure}[c]{0.32\textwidth}
	\scalebox{0.7}{
		\input{tikz/wdetExample/3}
	}
	\subcaption{After removing $n_1$}
	\label{fig:summary:afterShortcut}
\end{subfigure}\\
\vspace{0.5cm}
\begin{subfigure}[c]{0.32\textwidth}
	\scalebox{0.7}{
		\input{tikz/wdetExample/4}
	}
	\subcaption{After useless arc rule}
	\label{fig:summary:afterUseless}
\end{subfigure}
\begin{subfigure}[c]{0.33\textwidth}
	\scalebox{0.7}{
		\input{tikz/wdetExample/5}
	}
	\subcaption{After shortcut}
	\label{fig:summary:afterShortcut2}
\end{subfigure}
\begin{subfigure}[c]{0.32\textwidth}
	\scalebox{0.7}{
		\input{tikz/wdetExample/6}
	}
	\subcaption{After removing $n_2$}
	\label{fig:summary:afterShortcut3}
\end{subfigure}\\
\vspace{0.5cm}
\begin{subfigure}[c]{0.32\textwidth}
	\scalebox{0.7}{
		\input{tikz/wdetExample/7}
	}
	\subcaption{After useless arc rule}
	\label{fig:summary:afterUseless2}
\end{subfigure}
\begin{subfigure}[c]{0.33\textwidth}
	\scalebox{0.7}{
		\input{tikz/wdetExample/8}
	}
	\subcaption{After removing $n_3$}
	\label{fig:summary:afterShortcut4}
\end{subfigure}
\begin{subfigure}[c]{0.32\textwidth}
	\scalebox{0.7}{
		\input{tikz/wdetExample/9}
	}
	\subcaption{After merge rule}
	\label{fig:summary:afterMerge2}
\end{subfigure}
\caption{Example of the summary procedure}
\label{fig:summary}
\end{figure}

Figure \ref{fig:summary} \jd{shows a weakly deterministic negotiation, which 
is a slightly modified version of the Father-Daughter-Mother negotiation shown in 
Figure \ref{fig:summary:init}: Daughter now has the choice in $n_1$ whether to ask her 
Father (and possibly later her mother) or to directly ask her mother. Father thus has 
a nondeterministic edge to $n_f$ for the latter case. This negotiation is sound, 
acyclic and weakly deterministic, and can 
therefore be summarized by Theorem \ref{thm:completeness1}}.
The figure also  illustrates the summarization of the negotiation.  

We apply the merge rule twice to merge the results  $\texttt{y}$ and $\texttt{n}$ of 
the atoms $n_2$ and $n_3$. We call the merged results ``$\texttt{yn}$''.
The resulting negotiation is shown in Figure \ref{fig:summary:afterMerge}.

Now we shortcut the result $\texttt{st}$ of $n_0$, which unconditionally enables $n_1$ 
and has exclusive access to it. The result $\texttt{st}$ and the node $n_1$ are removed, and 
two new results $\texttt{af}$ and $\texttt{am}$ are added to $n_0$, yielding the negotiation 
in Figure \ref{fig:summary:afterShortcut}.

In this negotiation, applying the useless arc rule to $(n_0, \texttt{F}, \texttt{af})$ and 
$(n_0, \texttt{D}, \texttt{af})$ removes the arc $(n_0, \texttt{F}, \texttt{af}, n_f)$. 
Similarly, applying it to $(n_0, \texttt{M}, \texttt{am})$ and 
$(n_0, \texttt{D}, \texttt{am})$ removes the arc $(n_0, \texttt{M}, \texttt{am}, n_f)$. We 
obtain the negotiation of Figure \ref{fig:summary:afterUseless}. 

Now the result $\texttt{am}$ of $n_0$ unconditionally enables $n_3$, and 
the result $\texttt{am}$ of $n_2$ commits to $n_3$. So we can apply the shortcut rule,
removing the result $\texttt{am}$ of $n_0$, and adding a new result $\texttt{am\_yn}$.
This yields the negotiation in Figure \ref{fig:summary:afterShortcut2}.

The useless arc rule can be applied to $(n_0, F, \texttt{am\_yn})$ and $(n_0, D, \texttt{am\_yn})$,
removing the arc $(n_0, \texttt{F}, \texttt{am\_yn}, n_2)$.
After that, the result $\texttt{af}$ of $n_0$ unconditionally enables $n_2$ and has exclusive access 
to it. So we can apply the shortcut rule and remove $n_2$. As $n_2$ has two results, $\texttt{yn}$ and $\texttt{am}$, two new results $\texttt{af\_yn}$ and $\texttt{af\_am}$ are added to $n_0$. We obtain the negotiation of Figure \ref{fig:summary:afterShortcut3}.

We can now apply the useless arc rule twice, once to $(n_0, \texttt{M}, \texttt{af\_yn})$ 
and $(n_0, \texttt{F}, \texttt{af\_yn})$, and once to $(n_0, \texttt{M}, \texttt{af\_am})$
and $(n_0, \texttt{F}, \texttt{af\_am})$. The result is shown in Figure \ref{fig:summary:afterUseless2}.

Outcome $\texttt{af\_am}$ of $n_0$ unconditionally enables $n_3$, and has exclusive access to it. 
After the shortcut we obtain the negotiation of Figure \ref{fig:summary:afterShortcut4}. Applying the 
merge rule to $\texttt{am\_yn}$, $\texttt{af\_yn}$ and $\texttt{af\_am}$ yields a single result 
$\r$ of $n_0$ (result in Figure \ref{fig:summary:afterMerge2}). A final application of the shortcut 
rule to result $\r$ of $n_0$, which unconditionally enables $n_f$ with exclusive access, yields
the summary of the initial negotiation.

\subsection{Acyclic deterministic negotiations: Runtime analysis}
\label{sec:acyclictime}

Theorem \ref{thm:compacyc} shows that every maximal reduction sequence 
is finite, but does not provide any bound on its length. The family of 
negotiations shown at the 
top of Figure \ref{fig:expacyclic} shows that the length of maximal reduction sequences can
grow exponentially in the size of the negotiation, even in the sound and deterministic case. 
Indeed, consider a sequence of applications of the shortcut rule that always give priority to applications involving 
the initial atom. The second negotiation of the figure shows the result of shortcutting 
$(n_0, \texttt{a})$ and $(n_1, \texttt{a}_1)$ into a new outcome $(n_0, \texttt{a}\, \texttt{a}_1)$,
shortcutting  $(n_0, \texttt{a})$ and $(n_1, \texttt{b}_1)$ into a new outcome 
$(n_0, \texttt{a}\, \texttt{b}_1)$, and removing $n_1$. The third negotiation
shows the result of shortcutting $(n_0, \texttt{a}\, \texttt{a}_1)$ and $(n_0, \texttt{a}\, \texttt{b}_1)$
with the two outcomes of $n_2$, after which $n_0$ has four outcomes. 
Proceeding in this way we generate $2 ^{k-1}$ different results for $n_0$. 

\begin{figure}[htp]
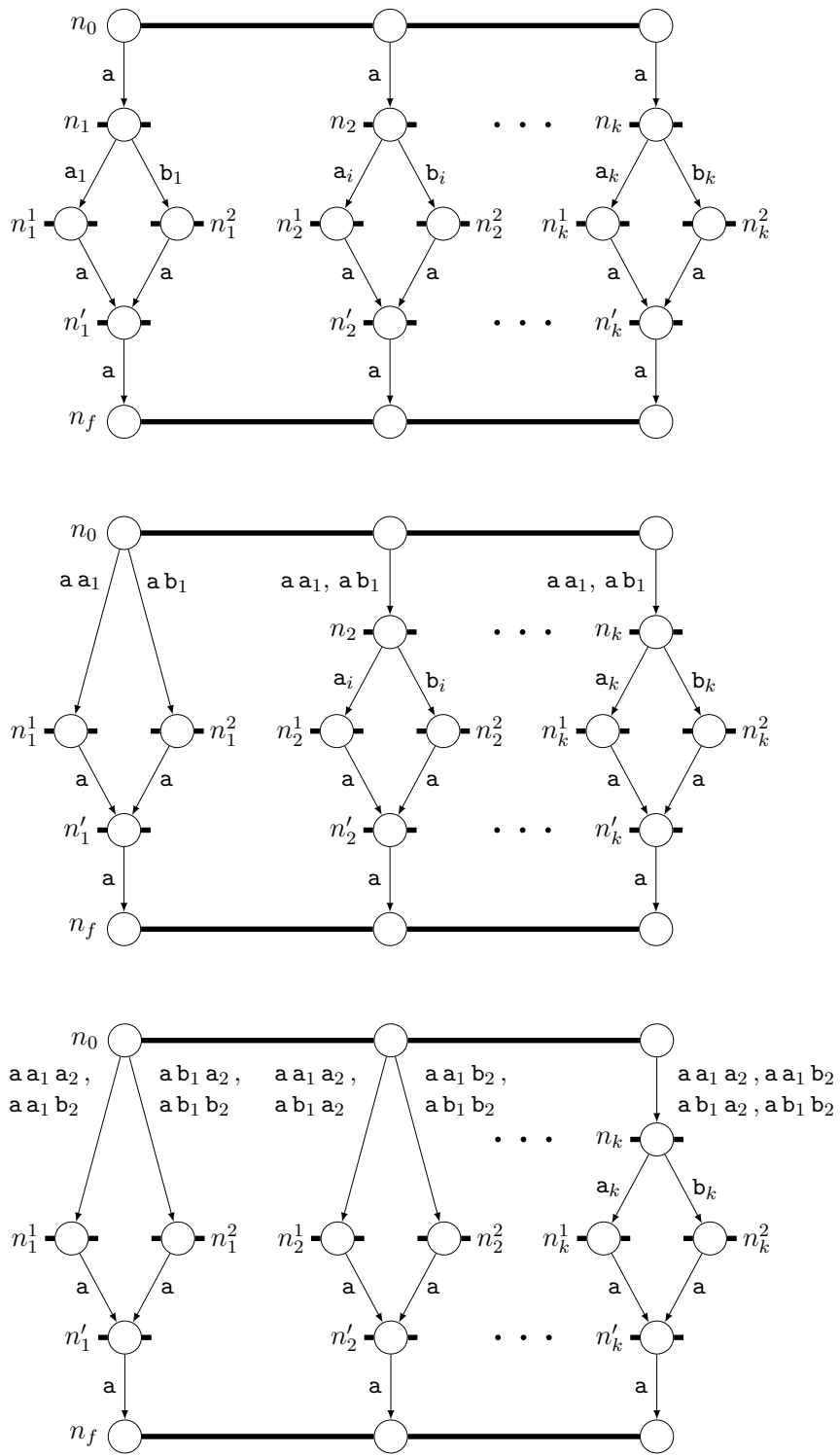

\centering	
\scalebox{0.9}{\input{tikz/expacyclic1}}\\[1.0cm]
\scalebox{0.9}{\input{tikz/expacyclic2}}\\[1.0cm]
\hspace{0.85cm}\scalebox{0.9}{\input{tikz/expacyclic3}}
\caption{An exponentially long reduction sequence}
\label{fig:expacyclic}
\end{figure}

However, there is another maximal reduction sequence with only a {\em polynomial} number of rule applications.
At the second negotiation of the figure we first apply the shortcut rule 
to $(n_0, \texttt{a}\,\texttt{a}_1)$ and $(n_1^1, \texttt{a})$, and then to $(n_0, \texttt{a}\,\texttt{b}_1)$ 
and $(n_1^2, \texttt{a})$. This yields two outcomes that can be merged with the help of the merge rule,
yielding a negotiation in which $n_0$ has again only one outcome. Proceeding in this way we
completely reduce the negotiation after only \ph{$5k+1$} rule applications.

The reason why the first reduction sequence ``blows up'' the negotiation is that an application of
the shortcut rule may actually {\em increase} the size of the negotiation. Indeed, when the 
rule is applied to an outcome $(n, \r)$ that unconditionally enables an atom $n'$ with $kl$ results for $l > 1$, 
then the rule removes the result $\r$ but adds $l$ new results, and so the total number of results increases.
This problem is avoided by applying the shortcut rule to an outcome only if the atom unconditionally enabled by it 
has one single result.

\begin{definition}
The \emph{deterministic shortcut rule}, or \emph{d-shortcut} rule, 
is defined by adding to the guard of the 
shortcut rule in Definition \ref{def:shortcut} a new condition: $n'$ has at most one result. 
(The action of the d-shortcut rule is identical
to that of the shortcut rule.)
\end{definition}

This reduction strategy corresponds to Algorithm \ref{alg:acyclic}, where $R(\N)$ denotes the set 
\ph{of} reducible outcomes of $\N$ (see Definition \ref{def:reducibleoutcome}). The algorithm also gives 
priority to the merge rule over the d-shortcut rule.

\begin{algorithm}[t]
\caption{Summarization algorithm for an acyclic negotiation $\N$}
\begin{algorithmic}[1]
\While{$R(\N) \neq \emptyset$}
\If{possible} apply the merge rule 
\ElsIf{possible} apply the d-shortcut rule 
\Else \ stop and answer ``unsound''
\EndIf
\State $\N := $ negotiation obtained after the application of the rule
\EndWhile
\end{algorithmic}
\label{alg:acyclic}
\end{algorithm}

In the rest of the section we prove that for an acyclic {\em deterministic} negotiation $\N$
Algorithm \ref{alg:acyclic} terminates after at most length $K \cdot L$ rule applications,
where $K$ and $L$ are the number of atoms and the number of outcomes of $\N$, respectively.
Moreover, if $\N$ is sound, then the algorithm returns an atomic negotiation. Whether such
an \ph{efficient} reduction algorithm also exists for the weakly deterministic case is still open.  


We first prove that Algorithm \ref{alg:acyclic} terminates after at most $K \cdot L$ rule applications
for arbitrary acyclic negotiations. For this we introduce a measure that
strictly decreases when a rule is applied. 

\begin{definition}
\label{def:nrsequence}
Let $\N$ be a negotiation, and let $(n, \r)$ be an outcome of $\N$. 
A {\em $(n, \r)$-sequence} is a finite occurrence sequence that starts with $(n, \r)$ and satisfies $P_{n'} \subseteq P_n$ for every atom occurring in $\sigma$.  A $(n,\r)$-sequence is {\em maximal} if it cannot be extended to a longer $(n,\r)$-sequence.
\end{definition}

\begin{fact}
Let $\vx_n$ be the marking given by $\vx_n(p) = \{n\}$ if $p \in P_n$ and $\vx_n(p) = \emptyset$ otherwise.
A sequence $\sigma$ is an $(n, \r)$-sequence if{}f it starts with $(n, \r)$ and is enabled at $\vx_n$.
Moreover, $\sigma$ is maximal if{}f $\vx_n \by{\sigma} \vx$ for a marking $\vx$ that enables no atom.
\end{fact}

\begin{definition}
\label{def:index}
Let $\N$ be a negotiation and let $(n, \r)$ be an outcome of $\N$. The {\em index}
of $(n,\r)$ in $\N$ is the length of a longest maximal 
$(n, \r)$-sequence of $\N$ minus $1$, if such sequence exists, or $\infty$, if $\N$ has 
arbitrarily long $(n, \r)$-sequences. The {\em index} of $\N$, 
denoted by $I(\N)$, is the sum of the indices of all outcomes
$(n, \r)$ such that $n \neq n_f$.
\end{definition}

\begin{lemma}
\label{lem:lemacyc}
Let $\N$ be an acyclic negotiation, and let $\N'$ be the result of applying the merge or the d-shortcut rule to $\N$.
Then $I(\N') < I(\N)$.
\end{lemma}
\begin{proof}
If the merge rule is applied to merge two outcomes
$(n, \r_1)$ and $(n, \r_2)$, then $(n, \r_1)$ and $(n, \r_2)$ have the same index $\ell$,
and $I(\N') =  I(\N)  - \ell$. Assume now that the d-shortcut rule is applied to some outcome $(n, \r)$ of $\N$,
yielding a new outcome $(n, \hat{\r})$. Since $(n, \hat{\r})$ is a shortcut, for every outcome 
$(n',\r') \neq (n, \hat{\r})$ appearing in both $\N$ and $\N'$, 
the index of $(n',\r')$ in $\N'$ is smaller than or equal to the index of 
$(n',\r')$ in $\N$. Moreover, we have $\hat{\ell} = \ell-1$, where $\ell$ and $\hat{\ell}$
are the indices of $(n, \r)$ in $\N$ and $(n, \hat{\r})$ in $\N$, respectively. So
$I(\N') \leq  I(\N) - \ell + \hat{\ell} < I(\N)$,
and we are done.
\end{proof}

\begin{proposition}
\label{prop:polcomp}
Let $\N$ be an acyclic negotiation with $K$ atoms and $L$ outcomes. 
Every maximal reduction sequence of $\N$ containing only applications 
of the merge and d-shortcut rules has length at most $K\cdot L$. In particular, 
Algorithm \ref{alg:acyclic} terminates after at most $K\cdot L$ iterations of the \textbf{while}-loop.
\end{proposition}
\begin{proof}
By Lemma \ref{lem:lemacyc}, and since $I(\N) \geq 0$ by definition, it suffices to show 
$I(\N) \leq  K \cdot L$. Since $\N$ is acyclic, every occurrence sequence fires an atom at most once. So the index
of any outcome of $\N$ is at most $K-1$, and hence $I(\N) \leq K \cdot L$. 
\end{proof}

In the rest of the section we prove that if $\N$ is acyclic, deterministic, and sound, then 
Algorithm \ref{alg:acyclic} completely reduces $\N$, i.e., returns an atomic negotiation.
The proof, which is rather involved, proceeds in three steps. First, 
we show a technical lemma showing that sound and deterministic acyclic negotiations can be reduced so that 
all agents participate in every atom with more 
than one result (Lemma \ref{lem:acyclicpoly1}). Then we use this result to prove that, loosely speaking, 
every sound and deterministic acyclic negotiation can be reduced to a ``replication'' of a negotiation with only one agent
(Lemma \ref{lem:acyclicpoly2}). Intuitively, a replication is a one-agent negotiation in disguise: Even if the negotiation 
has more than one agent, after each outcome they all move synchronously to the same new atom. Figure \ref{fig:rep} shows a one-agent 
negotiation and its replication to two agents. In the third step, we show that the algorithm reduces replications to atomic 
negotiations. 

\begin{figure}[ht]
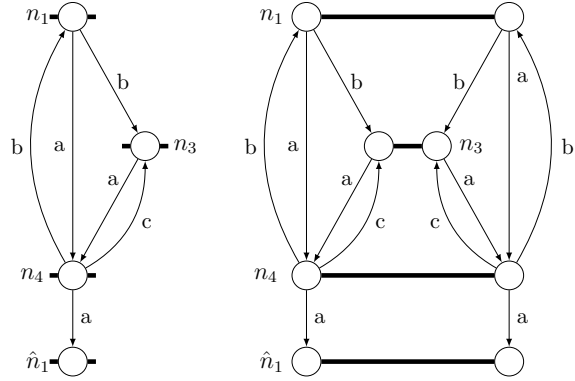

\centering
       \resizebox{!}{5cm}{
	\input{tikz/cyclicExample/fragmentafter2.tex}
        }
        \resizebox{!}{5cm}{
	\input{tikz/cyclicExample/fragmentafter1.tex}
        }
        \quad
	\caption{A one-agent negotiation and its replication to two agents.}
	\label{fig:rep}
\end{figure}

The following lemma is proved in Appendix \ref{app:irred}:

\begin{lemma}
\label{lem:acyclicpoly1} 
Let $\N$ be an irreducible sound and deterministic acyclic negotiation and 
let $n \neq n_f$ be an atom of $\N$ with more than one result.
Then every agent participates in $n$. 
\end{lemma}

Now we formally define replications. 

\begin{definition}
\label{def:uniform}
An outcome $(n, \r)$ is {\em uniform} if
$\trans(n,p,\r) = \trans(n,p',\r)$ for every two parties $p, p' \in P_n$.
A negotiation is a {\em replication} if all atoms have the same parties 
and every outcome is uniform.
\end{definition}

Observe that final outcomes $(n_f,\r)$ are uniform because
$\trans{(n_f, p, \r)} = \emptyset$ for every agent $p$ and result $\r$.
Moreover, if $(n, \r)$ is uniform and $n \neq n_f$ then there is a unique
atom $n'$ such that $\trans(n,p,\r) = \{n'\}$ for every party $p$. 

\begin{lemma} 
\label{lem:acyclicpoly2} 
Let $\N$ be an irreducible sound and deterministic acyclic negotiation. 
Then $\N$ is a replication.
\end{lemma} 
\begin{proof}
We first show that every agent participates in every atom. By Lemma \ref{lem:acyclicpoly1},
it suffices to prove that every atom $n \neq n_f$ has more than one result.
Assume the contrary, i.e., some atom different from $n_f$ has only one result.
Since, by soundness, every atom can occur, there is an initial occurrence sequence
$(n_0, r_0) (n_1, r_1) \cdots (n_k, r_k)$ such that $n_k$ has only one result
and all of $n_0, \ldots, n_{k-1}$ have more than one result. By Lemma \ref{lem:acyclicpoly1}, 
all agents participate in all of $n_0, n_1, n_{k-1}$. It follows that $(n_i, r_i)$ 
unconditionally enables $(n_{i+1}, r_{i+1})$ for every $0 \leq i \leq k-1$. 
In particular, $(n_{k-1}, r_{k-1})$ unconditionally enables $(n_k, r_k)$. But 
then, since $n_k$ has only one result, the d-shortcut rule can be applied to 
$n_{k-1}, n$, contradicting the hypothesis that $\N$ is irreducible.

For the second part, assume that some outcome $(n, \r)$ 
is not uniform. Then $n \neq n_f$, and there are two
distinct agents $p_1, p_2$ such that $\trans(n,p_1,\r)= \{n_1\} \neq \{n_2\} = \trans(n,p_2,\r)$. 
By the first part, every agent participates in $n$, $n_1$ and $n_2$. Since $\N$ is sound, some reachable
marking $\vx$ enables $n$. Moreover, since all agents participate in $n$, and $\N$ is
deterministic, the marking $\vx$ only enables $n$. Let $\vx'$ be the marking given by
$\vx \by{(n,\r)} \vx'$. Since $p_1$ participates in all atoms, no atom different from $n_1$ can be enabled
at $\vx'$. Symmetrically, no atom different from $n_2$ can be enabled
at $\vx'$. So $\vx'$ does not enable any atom, contradicting that $\N$ is sound.
\end{proof}
 
\begin{proposition}
\label{prop:irredtheo}
Let $\N$ be an irreducible sound and deterministic acyclic negotiation. Then $\N$ is atomic. 
\end{proposition}
\begin{proof}
Assume $\N$ contains more than one atom. For every atom $n \neq n_f$,
let $l(n)$ be the length of the longest path from $n$ to $n_f$ in the graph of $\N$. Let
$n_{\min}$ be any atom such that $l(n_{\min})$ is minimal, and let $\r$ be an arbitrary result 
of $n_{\min}$. By Lemma \ref{lem:acyclicpoly2} there is an atom
$n$ such that $\trans(n_{\min},a,\r) = \{n\}$ for every party $p$ of $n_{\min}$. If $n \neq n_f$, then by
acyclicity we have $l(n) < l(n_{\min})$, contradicting the minimality of $n_{\min}$.
So $n = n_f$, and therefore $\trans(n_{\min},p,\r) = \{n_f\}$ for every result $\r$ of $n_{\min}$ and every party $p$.
If $n_{\min}$ has more than one result, then the merge rule is applicable. If $n_{\min}$ has one 
result, then, since it unconditionally enables $n_f$, the d-shortcut rule is applicable. 
In both cases we get a contradiction to irreducibility.
\end{proof} 

Proposition \ref{prop:irredtheo} proves that every maximal reduction sequence that
uses the merge and d-shortcut rules leads to the same result: a summary of $\N$. 

Putting Proposition \ref{prop:polcomp} and Proposition \ref{prop:irredtheo} \ph{together}, we obtain our  result:

\begin{theorem}
\label{thm:acyclicdeterministic}
Let $\N$ be an acyclic, deterministic negotiation. 
Algorithm \ref{alg:acyclic} terminates after at most $K\cdot L$ iterations of the \textbf{while}-loop,
and returns an atomic negotiation if{}f $\N$ is sound.
\end{theorem}

\section{Summarizing Deterministic Negotiations: The One-agent Case}
\label{sec:cyclic}

We have shown that any maximal reduction sequence for a sound and deterministic acyclic negotiation  
that only applies the d-shortcut rule leads to a summary after a finite, polynomial number of steps. 
The following examples show that this does not hold for cyclic deterministic negotiations,
not even for the degenerate case of negotiations with one single agent. 
\begin{figure}[ht]
\centering
\input{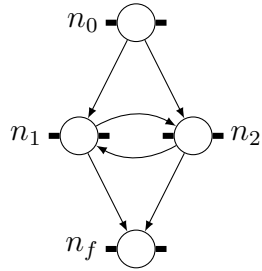}
\caption{A cyclic negotiation where every atom has two outcomes}
\label{fig:cyclictwooutcomes}
\end{figure}

A first problem is that in the cyclic case we can no longer restrict ourselves to the d-shortcut rule. 
Consider the sound negotiation of Figure \ref{fig:cyclictwooutcomes} \footnote{Since in the
examples of this section the names of the results are not important, in all the figures we omit them
for clarity.}
Every atom has two results,
and so the d-shortcut rule cannot be applied. Since the merge and iteration rules are not applicable either, the
negotiation cannot be summarized unless we allow to apply the shortcut rule in more generality.

\begin{figure}[ht]
\centering
\input{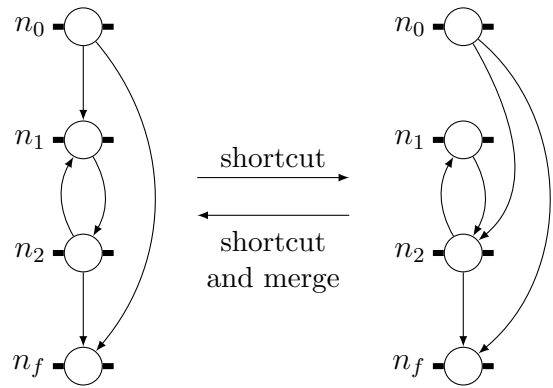}
\caption{Shortcutting and cycles}
\label{fig:cyclicendless}
\end{figure}

A second problem is that in the cyclic case infinite reduction sequences are possible. 
Consider the negotiation on the left of Figure \ref{fig:cyclicendless}. If we apply the 
shortcut rule to $n_0$ and the result leading to $n_1$, we obtain the negotiation on the right. 
Applying the shortcut rule to $n_0$ and the result leading to $n_2$, 
and then the merge rule, we obtain again the negotiation on the left. This process can be repeated 
arbitrarily often.

\begin{figure}[ht]
\centering
\input{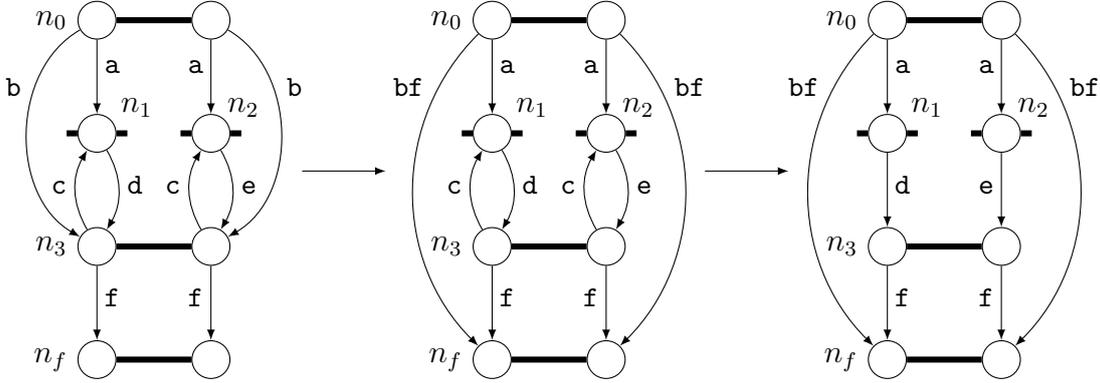}
\caption{ A removed result is produced again later on}
\label{fig:cyclicnoedgeagain}
\end{figure}

To solve this problem, one could try the following policy: 
do not apply the shortcut rule if it produces an arc that has already been removed by an earlier rule (intuitively,
we do not apply rules that ``undo'' reductions by earlier rule applications.) However,
this policy prevents the summarization of the sound negotiation on
the left of Figure \ref{fig:cyclicnoedgeagain}.
We start by applying the d-shortcut rule to $n_0$ and the result $\texttt{b}$, which yields the negotiation
in the middle  of the figure. It is easy to see that this negotiation cannot be summarized 
without generating again a result of $n_0$ leading to $n_3$. Indeed, after
shortcutting $\texttt{c}$ and applying the iteration rule we get the negotiation on the
right, where after removing $n_1$ and $n_2$ with the shortcut rule we get
a negotiation where $n_0$ contains again a result leading to $n_3$. Other reduction
sequences starting at the negotiation in the middle also generate such a result.
We therefore need a more sophisticated approach for the
summarization of cyclic deterministic negotiations. In this section 
we consider the one-agent case, which already illustrates the main ideas. The general case is studied in 
Section \ref{subsec:many-agents-cyclic}.

Let $\N = (N, n_0, n_f, \trans)$ be an arbitrary deterministic negotiation with only one agent. 
We first describe our summarization procedure, and then discuss its steps in more detail. 
We use the negotiation of Figure \ref{fig:DFSexample} as example. 

\begin{figure}[ht]
\centering
\input{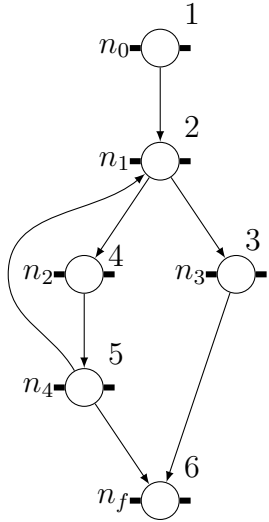}
\caption{A one-agent negotiation.}
\label{fig:DFSexample}
\end{figure}

The summarization procedure is parameterized by a fixed but arbitrary total order $\prec$ on the atoms of the negotiation. In Figure \ref{fig:DFSexample} the order is $1 \prec 2 \prec \cdots  \prec 6$.

Using this order, we classify outcomes into {\em backward} and {\em forward} outcomes, and define a 
total order on outcomes.

\begin{definition}
\label{def:backward}
Let $\N$ be a negotiation with one single agent $p$. Let $n \by{\r} n'$ denote $\trans(n,p,\r) = \{n'\}$. 
Abusing language, we call $n \by{\r} n'$ an outcome of $\N$ \ph{(instead of $(n,r)$)}. 

An outcome $n \by{\r} n'$ is {\em backward} if $n' \neq n_f$ and $n' \prec n$. Otherwise it is {\em forward}. 
We extend the order $\prec$ to a total order on outcomes, also denoted by $\prec$, 
as follows: 
\begin{center}
$(n_1 \by{\r_1} n_1') \prec (n_2 \by{\r_2} n_2')$ if{}f 
$n_1' \prec n_2'$ or $n_1' = n_2'$ and $n_1 \prec n_2$.
\end{center}
A backward outcome is {\em minimal} if no other backward outcome is 
strictly smaller with respect to $\prec$. 
\end{definition}

In Figure \ref{fig:DFSexample} the only backward outcome is $n_4 \by{} n_1$, which is also minimal.
The classification of outcomes into forward and 
backward is only defined for the one-agent case; we will generalize it later.

Observe that every backward outcome $n \by{\r} n'$ is reducible, 
because $n' \neq n_f$ and thus the shortcut rule can be applied to it. 
The summarization procedure is shown in Algorithm \ref{alg:oneAgent}.
It gives priority to the merge and iteration rules, that is, 
they are applied whenever possible. If neither of them can be applied, then the algorithm
 selects a minimal backward
outcome and applies the shortcut rule to it. If this is also not possible, then the algorithm 
applies the d-shortcut rule \ph{to another outcome}. We will see that in this case the d-shortcut rule is necessarily 
applicable, and so that the algorithm never gets stuck. 

Notice that the algorithm never answers ``not sound''. The reason is that all one-agent negotiations 
are sound. Indeed, by definition every atom of a negotiation lies on a path between the initial and 
final atoms, which in the one-agent case implies soundness.

\begin{algorithm}[t]
\caption{Summarization algorithm for an one-agent negotiation $\N$}
\begin{algorithmic}[1]
\State fix an arbitrary total order $\prec$ on the atoms of $\N$
\While{$R(\N) \neq \emptyset$}
\If{possible} apply the merge rule 
\ElsIf{possible} apply the iteration rule 
\ElsIf{possible} apply the shortcut rule to a minimal backward outcome 
\Else \ apply the d-shortcut rule 
\EndIf
\State $\N := $ negotiation obtained after the application of the rule
\EndWhile
\end{algorithmic}
\label{alg:oneAgent}
\end{algorithm}

\begin{figure}[ht]
\centering
\input{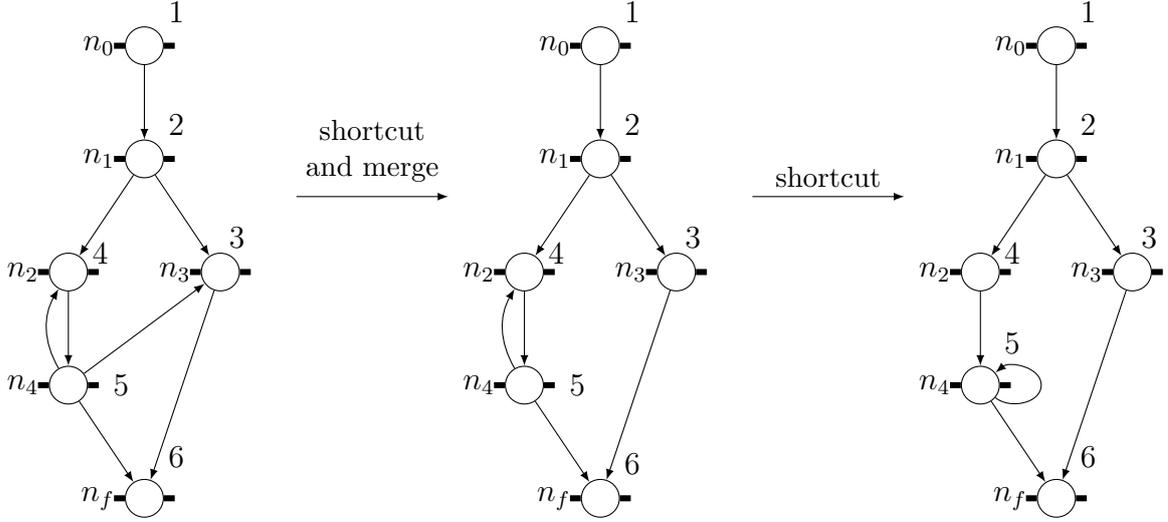}
\caption{Reduction of the negotiation of Figure \ref{fig:DFSexample}.}
\label{fig:oneAgentExample}
\end{figure}

Let us apply Algorithm \ref{alg:oneAgent} to the negotiation of  Figure \ref{fig:DFSexample}.
In the first iteration of the loop the 
algorithm shortcuts the only backward outcome $n_4 \by{} n_1$, yielding the negotiation
on the left of Figure \ref{fig:oneAgentExample}. This negotiation has two backward outcomes,
$n_4 \by{} n_2$ and $n_4 \by{} n_3$. Since $n_3 \prec n_2$, the minimal backward outcome
is $n_4 \by{} n_3$. After applying the shortcut rule to $n_4 \by{} n_3$, a merge yields the negotiation
in the middle of the figure. In the third iteration, the algorithm shortcuts the only backward result
$n_4 \by{} n_2$, yielding the negotiation on the right. An application 
of the iteration rule yields an acyclic negotiation, which is reduced by applying 
the merge and d-shortcut rules.

\subsection{Termination and Runtime Analysis}

We show that Algorithm \ref{alg:oneAgent} always terminates after at most $2K^3 + K^2 +L$ rule applications, 
where $K$ and $L$ are the number of atoms and outcomes of the negotiation, respectively. 
Moreover, the algorithm always returns a summary of the negotiation.

Let $\N$ be a one-agent negotiation, and let $\N_1, \N_2, \ldots$ be the sequence 
of negotiations produced by Algorithm \ref{alg:oneAgent} starting from $\N_1 := \N$, where
each $\N_{i+1}$ is obtained from $\N_i$ by applying one of the rules. We call this sequence
the {\em computation} of the algorithm on $\N$.

We start the proof with \ph{a} lemma:

\begin{lemma}
\label{lem:alg2shortcut}
Let $\N_1, \N_2, \ldots$ be the computation of Algorithm \ref{alg:oneAgent} on a one-agent negotiation $\N$. 
Let $\N_{t_1}$ and $\N_{t_2}$ be two negotiations of the computation to which Algorithm \ref{alg:oneAgent} applies the shortcut rule at line 5, and such such that $t_1 < t_2$. 
Let $n_1 \by{\r_1} n_1'$ and $n_2 \by{\r_2} n_2'$ be the outcomes of $\N_{t_1}$ and $\N_{t_2}$ to which the shortcut rule is applied. Then $(n_1 \by{\r_1} n_1') \prec (n_2 \by{\r_2} n_2')$.
\end{lemma}
\begin{proof}
By the transitivity of $\prec$, it suffices to prove the result for the case in which the 
algorithm does not apply the shortcut rule  \ph{to a negotiation $\N_i$ with $t_i < i < t_2$}. 


If $\N_{t_1}$ contains some outcome $n_2 \by{\r} n_2'$, then,
since $n_2 \by{\r_2} n_2'$ is a backward outcome, $n_2 \by{\r} n_2'$ is a backward outcome too. 
Since $n_1 \by{\r_1} n_1'$ is the minimal backward outcome
of $\N_{t_1}$, we have  $(n_1 \by{\r_1} n_1') \prec (n_2 \by{\r} n_2')$, and thus 
$(n_1 \by{\r_1} n_1') \prec (n_2 \by{\r_2} n_2')$.

Assume now that $\N_{t_1}$ does not contain any outcome from $n_2$ to $n_2'$. 
Since between $t_1+1$ and $t_2-1$ the algorithm only applies the merge and iteration rules, 
and these rules do not create any \ph{additional} outcomes, $n_2 \by{\r_2} n_2'$ is created by the application 
of the shortcut rule to the outcome $n_1 \by{\r_1} n_1'$ of $\N_{t_1}$. By the definition
of the shortcut rule, it follows that $n_2=n_1$ and that $\N_{t_1}$
contains an outcome $n_1' \by{\r} n_2'$. If $n_1' = n_2'$ then, since the algorithm gives 
priority to the iteration rule over the shortcut rule, it would not have applied the shortcut rule 
to $n_1 \by{\r_1} n_1'$. So we have $n_1' \neq n_2'$. 

If $n_2' \prec n_1'$ then $n_1' \by{\r} n_2'$ is a backward outcome and by definition
$(n_1' \by{\r} n_2') \prec (n_1 \by{\r_1} n_1')$. This contradicts the assumption that $(n_1 \by{\r_1} n_1')$
is the minimal backward outcome of $\N_{t_1}$. So $n_1' \prec n_2'$,
and thus $(n_1 \by{\r_1} n_1') \prec (n_2 \by{\r_2} n_2')$.
\end{proof}

\begin{theorem}
\label{thm:oneAgentPoly}
If $\N$ has $K$ atoms and $L$ outcomes, then Algorithm \ref{alg:oneAgent} terminates after at most
$2 K^3+K^2+L$ rule applications, and it summarizes $\N$.
\end{theorem}
\begin{proof}
Let $\N_1, \N_2, \ldots$ be the computation of Algorithm \ref{alg:oneAgent} on $\N$.
We divide it into two phases. Phase I
terminates immediately before the first execution of line 6. Phase II starts with the first
execution of line 6 and continues until the end of the computation. We show that
phase I and phase II terminate after at most $K^3+K^2+L$ and $K^3$ rule applications, respectively.

By Lemma \ref{lem:alg2shortcut}, if during phase I of the computation the algorithm 
applies the shortcut rule to two outcomes $(n_1 \by{\r_1} n_1')$ and $(n_2 \by{\r_2} n_2')$,
then $(n_1 \by{\r_1} n_1') \prec (n_2 \by{\r_2} n_2')$. In particular, this implies 
$n_1 \neq n_1'$ or $n_2 \neq n_2'$. So the number of applications of the shortcut rule during phase
I is bounded by the number of pairs of atoms, i.e., 
by $K^2$.

Since the algorithm gives priority to the merge and iteration rules over the shortcut rule, 
it only  applies the shortcut rule to negotiations having at most one result between any
pair of nodes. Therefore, if the shortcut rule is applied to $n \by{\r} n'$, then the atom
$n'$ has at most $K$ outcomes, and so the rule generates at most $K$ new outcomes. It follows that
in phase I the merge and iteration rules are applied at most $K$ times between two consecutive
applications of the shortcut rule, and also after the last application of the shortcut rule. 
Further, since every application removes one outcome, before the first application of the shortcut
rule merge and iteration are applied together at most $L$ times. So the total number of applications of the 
merge and iteration rules during phase I is bounded by $L + K^2 \cdot K = L + K^3$, and the 
overall total number of applications is bounded by $K^3+ K^2 + L$.

Let $\N_t$ be the first negotiation of phase II of the computation, i.e., 
the first negotiation to which the d-shortcut rule is applied at line 6.
By the definition of the algorithm, $\N_t$ has only forward outcomes. We prove that
during phase II the algorithm applies at most $K^3$ rules. We need two preliminary claims.

\medskip

\textbf{Claim 1.} $\N_t$ is acyclic. \\
If $n \by{\r} n$, then the iteration rule can be applied to $\N_t$, a contradiction. 
If $\N_t$ contains a cycle with at least two atoms, then the cycle contains at least one result
$n \by{\r} n'$ such that $n \succ n'$. Since $n_f$ does not belong to any cycle, we have
$n' \neq n_f$, and so $n \by{\r} n'$ is backward, again a contradiction. 

\medskip

\textbf{Claim 2.} During phase II the algorithm only applies the merge, iteration, 
and d-shortcut rules. \\
Follows easily from the fact that the application of the shortcut rule to a negotiation that has only 
forward outcomes leads to another negotiation satisfying the same property. So $\N_{t'}$ contains
only forward results for $t' > t$. By the definition of the algorithm, line 5 is never executed, and we are done.

\medskip

These two claims allow us to apply Proposition \ref{prop:polcomp}. We conclude that 
the length of the sequence of rules applied by the algorithm from $\N_t$ is bounded by 
$K_t \cdot L_t$, where $K_t$ and $L_t$ are
the number of atoms and non-final outcomes of $\N_t$. Clearly we have $K_t \leq K$.
Further, since neither the merge nor the iteration rule can be applied 
to $\N_t$, for every two atoms $n, n'$ of $\N$ there is at most one outcome $n \by{r} n'$. 
and so $L_t \leq K^2$. So $K_t \cdot L_t \leq K^3$, and we are done. 

It remains to show that Algorithm \ref{alg:oneAgent} summarizes 
$\N$, \ph{i.e., that it eventually reaches an atomic negotiation}. 
\ph{It suffices to look at the suffix of the sequence starting with $\N_t$.}
By Claim 1, $\N_t$ is acyclic. Further, since $\N_t$ has one agent and every atom 
of $\N_t$ lies on a path between the initial and the final atoms, $\N_t$ is sound. 
By Proposition \ref{prop:irredtheo}, the algorithm summarizes $\N_t$. 
\end{proof}

\subsubsection{Extension to Replications}

Recall that a replication is a negotiation whose atoms have all the same
parties, and whose outcomes are all uniform, that is, for every result 
all parties move together to the same atom (see Definition \ref{def:uniform}). 
Algorithm \ref{alg:oneAgent} also works for replications. Indeed, in order to 
summarize a replication, we just apply the
same sequence of rules that we would apply if all atoms had only one party. 
By Theorem \ref{thm:oneAgentPoly}, we obtain: 

\begin{theorem}
\label{thm:repred}
Algorithm \ref{alg:oneAgent} summarizes a replication $\N$ with 
$K$ atoms and $L$ outcomes
after at most $2K^3+K^2+L$ rule applications.
\end{theorem}

While this extension is straightforward, it is important because 
such negotiations will appear at intermediate stages of our general summarization algorithm.

\section{Summarizing Deterministic Negotiations with Multiple Agents}
\label{subsec:many-agents-cyclic}

We present a reduction algorithm for deterministic negotiations and show that it summarizes every sound 
deterministic negotiation with $K$ atoms and $L$ outcomes after at most 
$2K^3 + K^2+KL+ L$ rule applications.

\begin{figure}[ht]
\centering
        \resizebox{!}{8cm}{
	\input{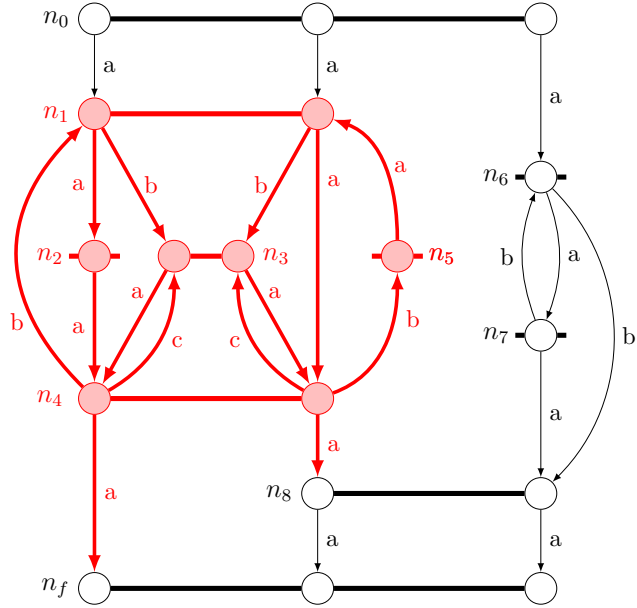}
        }
	\caption{A sound deterministic negotiation.}
	\label{fig:cyclicRunning}
\end{figure}

\subsection{Informal description}
We give a first overview of the summarization procedure using the sound deterministic 
negotiation of Figure \ref{fig:cyclicRunning} as example. 

The procedure is based on the notion of the {\em fragment} of a negotiation generated by an atom. 
A fragment is almost a negotiation, with the only difference that the final atom may be missing,
and so there may be ``dangling outcomes''.  For example, in the negotiation
of Figure \ref{fig:cyclicRunning} the fragment generated by $n_1$ is highlighted in red, and
$(n_4, \texttt{a})$ is a ``dangling outcome''. Intuitively, the fragment generated by an atom $n$ is 
the part of the negotiation that can occur from the marking $\vx$ satisfying
$\vx(p) = \{n\}$ for every party of $n$, and $\vx(p) = \emptyset$ otherwise. 

The procedure first summarizes
all {\em one-agent fragments}, i.e., all fragments generated by atoms with one party,
then all {\em two-agent fragments}, and so on. Figure \ref{fig:stage1} shows on the left
the one-agent fragments of the negotiation, and on the right the result of the reduction:
\ph{The} fragments generated by $n_2$ and $n_5$ are already summarized, and thus do not change, while the
fragment generated by $n_6$ (which coincides with the fragment of $n_7$) is summarized into
atom $n_{67}$. In the second stage the procedure summarizes the fragments of the
atoms $n_1$ and $n_8$ (see Figure \ref{fig:stage2}, where the fragments are drawn using different 
colors for clarity). Notice that the fragments generated by $n_1$, $n_3$, and $n_4$ are identical.
The fragment of $n_8$ is atomic, while the fragment for $n_1$
is summarized into $n_{15}$. The only three-agent fragment of the resulting negotiation is
the negotiation itself, and after the third step the negotiation is completely summarized.

\begin{figure}[ht]
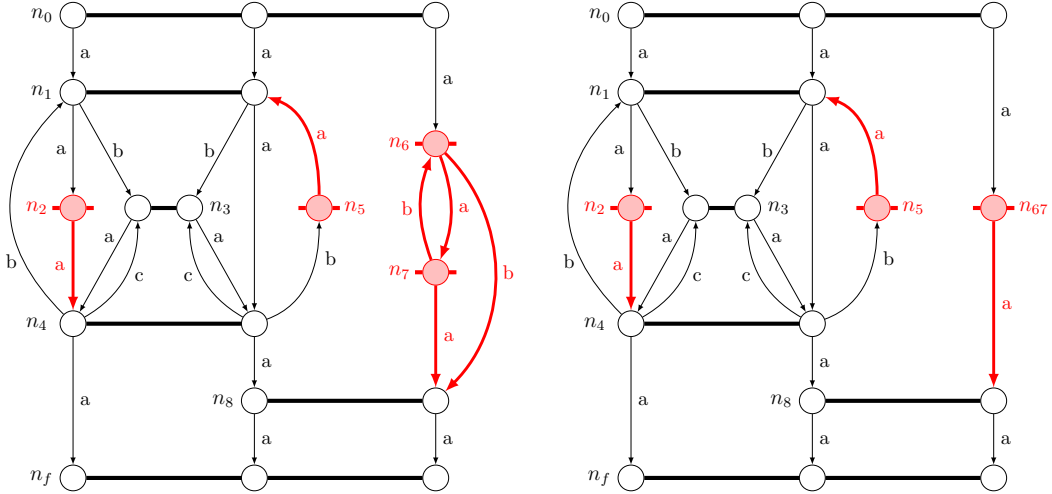

\centering
        \resizebox{!}{6.5cm}{
	\input{tikz/cyclicExample/before1}
        }
        \resizebox{!}{6.5cm}{
        \input{tikz/cyclicExample/after1red}
        }
	\caption{Before and after reducing one-agent fragments.}
	\label{fig:stage1}
\end{figure}

\begin{figure}[ht]
\centering
        \resizebox{!}{6.5cm}{
	\input{tikz/cyclicExample/before2}
        }
        \resizebox{!}{6.5cm}{
        \input{tikz/cyclicExample/after2red}
        }
	\caption{Before and after reducing two-agent fragments.}
	\label{fig:stage2}
\end{figure}

\begin{figure}[ht]
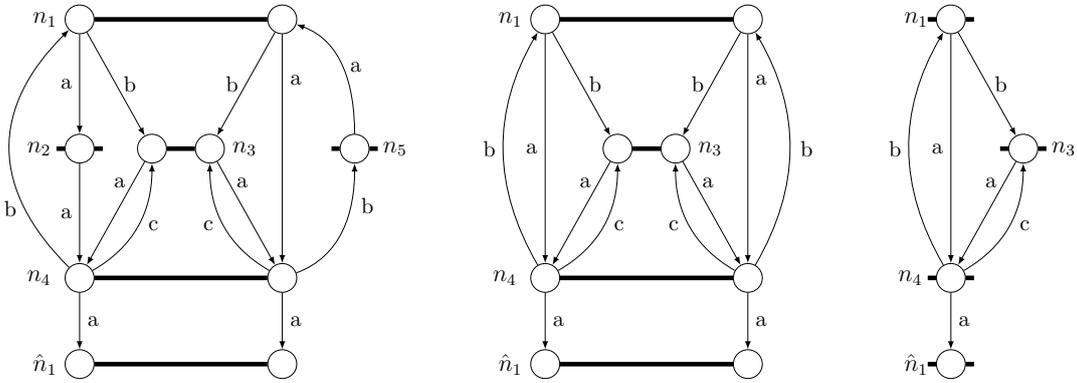

\centering
        \resizebox{!}{5cm}{
	\input{tikz/cyclicExample/fragmentbefore.tex}
	}
        \quad
        \resizebox{!}{5cm}{
	\input{tikz/cyclicExample/fragmentafter1.tex}
        }
        \quad
        \resizebox{!}{5cm}{
	\input{tikz/cyclicExample/fragmentafter2.tex}
        }
	\caption{Reducing a fragment.}
	\label{fig:fragmentExample}
\end{figure}

The procedure to reduce a fragment is illustrated in Figure \ref{fig:fragmentExample}.
On the left we see the fragment of the negotiation of Figure \ref{fig:cyclicRunning} generated by the atom $n_1$,
with a ``mock'' final atom $\hat{n}_1$ for the dangling outcome. The procedure first removes all atoms
with fewer agents than $n_1$, in this case $n_2$ and $n_5$. In this example this is achieved by two
applications of the shortcut rule, and it yields the negotiation in the middle of the figure. This negotiation is nothing 
but the replication of the one-agent negotiation 
shown on the right of the figure: The two agents behave identically. We reduce this replication
by means of the same sequence of rules we would use to reduce its one-agent counterpart.

Observe that, while Figure  \ref{fig:fragmentExample} shows the fragment separated from the rest 
of the negotiation, the sequence of rules applied to summarize the fragment can also
be applied to the negotiation of Figure \ref{fig:cyclicRunning}. The sequence produces the same effect,
and eventually replaces the fragment by one single atom. It has no side-effects. 

What happens if the summarization procedure is applied to an {\em unsound} deterministic negotiation? Since the rules
preserve unsoundness, the procedure does not summarize the negotiation. It may
not terminate, or terminate after a very large number of rule applications. Notice, however,
that we can still use the algorithm  to decide soundness in polynomial time. Indeed, if after
$2K^3 + K^2 + KL + L$ rule applications the procedure has not summarized the negotiation, then we know that
it will never \ph{do}, and that the initial negotiation is unsound. We discuss this point in more detail later.

In \ph{Appendix} \ref{subsec:loopsynch} we obtain some preliminary results. We study loops, \ph{which are} occurrence sequences leading
from a marking to itself. We prove a  lemma \ph{which turns out to be fundamental for the sequel}: \ph{The} loops of sound deterministic negotiations 
\ph{contain a ``synchronizing''} atom involving all parties of
all atoms in the loop. In Section \ref{subsec:fragseg} we formally define 
the fragment of a negotiation generated by an atom. We also define the segment generated by an outcome, and prove 
a second fundamental result: The fragments and segments of sound and deterministic negotiations
are also sound negotiations. Section \ref{subsec:summalg} describes the summarization algorithm
for the sound case. 
Sections \ref{subsec:manycomp} and \ref{subsec:manyruntime} analyze its behavior.
Section \ref{subsec:algforunsound} \ph{provides} the algorithm for the general case in which the negotiation can be sound
or unsound.

\subsection{Fragments and Segments}
\label{subsec:fragseg}

We show that every atom $n$ and every outcome $(n, \r)$
of a sound deterministic negotiation induces a ``sound subnegotiation'' of $\N$, called 
the {\em $n$-fragment} and the {\em $(n, \r)$-segment} of $\N$, respectively.

In Definition \ref{def:nrsequence} we introduced $(n, \r)$-sequences.
Loosely speaking, a $(n, \r)$-sequence is a \ph{finite} occurrence sequence that starts with $(n, \r)$ and can be executed by the parties of $n$  on their own, without involving any other agent. We now come back to this notion. 

\begin{definition}
\label{def:closed}
A sequence of outcomes is a {\em $n$-sequence} if it is a $(n, \r)$-sequence for some result $\r$ of $n$. A $n$-sequence is {\em maximal} if it cannot be extended to a longer $n$-sequence.

A $(n,\r)$-sequence is {\em strict} if every atom $n'$ that appears in it after the first outcome $(n, \r)$ satisfies 
$P_{n'} \subset P_n$. A strict $(n, \r)$-sequence is {\em maximal} if it cannot be extended to a longer strict 
$(n, \r)$-sequence.

Let $\vx_n$ be the marking given by $\vx_n(p) = \{n\}$ if $p \in P_n$ and $\vx_n(p) = \emptyset$ otherwise.
The {\em target} of a $(n,\r)$-sequence $\sigma$ is the marking reached by firing $\sigma$ from $\vx_n$. 
\end{definition}

The sequences $(n_6, \texttt{a}) \, (n_7, \texttt{a})$  and 
$(n_6, \texttt{a}) \, (n_7, \texttt{b}) \, (n_6, \texttt{b})$ of the negotiation shown in Figure \ref{fig:cyclicRunning} are 
maximal $n_6$-sequences, but not strict $(n_6, \texttt{a})$-sequences. Calling the agents
of the negotiation $A$, $B$, and $C$, the target of all these
sequences is the marking $\vx$ given by $\vx(A) = \vx(B) = \emptyset$ and $\vx(C) = \{ n_8 \}$.
The sequence
$(n_1, \texttt{a}) \, (n_2, \texttt{a})$ is a maximal strict 
$(n_1, \texttt{a})$-sequence, but not a maximal $n_1$-sequence. The sequence
$(n_3, \texttt{a}) \, (n_4, \texttt{a})$ is a maximal $n_3$-sequence, and its target is the marking
$\vx'$ given by $\vx'(A) = \{ n_f\}$, $\vx'(B) = \{ n_8\}$, and $\vx'(C) = \emptyset$.

Appendix \ref{app:fragseq} \ph{provides a proof of} the proposition below, stating that, for every atom $n$ of a sound deterministic negotiation,
{\em all maximal $n$-sequences have the same target}. In our example, 
the markings $\vx$ and $\vx'$ above are the targets of all maximal $n_6$-sequences and all maximal
$n_2$-sequences, respectively. The same property holds for 
the strict $(n, \r)$-sequences too. 

\begin{proposition}
\label{prop:unique}
Let $\N$ be a sound and deterministic negotiation, and let $n$ be an atom of $\N$.
\begin{itemize}
\item[(a)] All maximal $n$-sequences have the same target.\\
That is: there is a unique marking $\vx$ such that $\vx_n \by{\sigma} \vx$ for every 
maximal $n$-sequence $\sigma$. We call $\vx$ the {\em target} of $n$.
\item[(b)] For every outcome $(n, \r)$, all maximal strict $(n, \r)$-sequences have the same target. \\
That is: there is a unique marking $\vx$ such that $\vx_n \by{\sigma} \vx$ for every 
maximal strict $(n, \r)$-sequence $\sigma$. We call $\vx$ the {\em target} of $(n,\r)$.
\end{itemize}
\end{proposition}

While the proof of this \ph{proposition} is surprisingly non-trivial, there is a clear intuition for it.  
After a maximal $n$-sequence,
the agents of $P_n$ have to interact with agents
of $\agents \setminus P_n$ in order to reach the 
final marking, and these agents have to be ready for the interaction, otherwise a deadlock would arise. If 
maximal $n$-sequences can lead to different targets, then, since the  agents of $\agents \setminus P_n$ do not know which
target will be chosen by the agents of $P_n$, they must be ready to engage in different atoms. But
this is not possible, because in a deterministic negotiation  an agent is always committed
to at most one atom. 

Proposition \ref{prop:unique} allows us to define the fragments of sound deterministic negotiations.

\begin{definition}
\label{def:segfrag}
Let $\N = (N, n_0, n_f, \trans)$ be a sound and deterministic negotiation.
Let $n$ be an atom of $\N$, and let $\tgt{n}$ be its target.
The {\em $n$-fragment} of $\N$ is the \ph{tuple} $\F_n = (F_n, n, \hat{n}, \trans_n)$ defined as follows:
\begin{itemize}
\item $F_n$ contains the atoms occurring in the $n$-sequences of $\N$, 
plus a fresh atom $\hat{n}$. This atom \ph{$\hat{n}$} has the same parties as $n$, and \ph{it has} one single result $\hat{\r}$. 
\item For every atom $n' \in F_n$, party $p$ and result $\r$:
$$\trans_n(n',p, \r) = 
\begin{cases}
\emptyset & \mbox{ if $n' = \hat{n}$ and $\r = \hat{\r}$} \\
\hat{n}   & \mbox{ if $\trans(n', p, \r) = \tgt{n}$ } \\
\trans(n', \ph{p}, \r) & \mbox{ otherwise}
\end{cases}$$
\end{itemize} 
\noindent The $n$-fragment of $\N$  is {\em atomic} if $F_n = \{n, \hat{n}\}$ and 
$n$ has exactly one result.

\medskip

Let $(n, \r)$ be an outcome of $\N$, and let $\tgt{n,\r}$ be its target.
The {\em $(n, \r)$-segment} of $\N$ is the tuple $\S_{(n,\r)} = (S_{(n,\r)}, n, \hat{n}, \trans_{(n,\r)})$ defined exactly as 
the $n$-fragment above, but replacing the $n$-sequences by the strict $(n,\r)$-sequences of $\N$. Atomic segments are  defined \ph{like} atomic fragments.
\end{definition}

The $n_1$-fragment of the negotiation of Figure \ref{fig:cyclicRunning} is shown on the left of Figure \ref{fig:fragmentExample}. Observe that this is also the $n_3$-fragment and the $n_4$-fragment. 
The segments of the outcomes contained in the $n_1$-fragment are all atomic, \ph{except} the segments 
\ph{of} $(n_1, \texttt{a})$ and \ph{of} $(n_4, \texttt{b})$.

We show that the fragments and segments of sound deterministic negotiations are 
also sound. This result allows us to summarize sound negotiations ``inside-out'': first summarize
the fragments for agents with one party, \ph{then for agents with} two parties, and so on. 

\begin{proposition}
\label{prop:segfragsound}
Let $\N = (N, n_0, n_f, \trans)$ be a sound and deterministic negotiation.
\begin{itemize} 
\item[(1)] For every atom $n$, the $n$-fragment \ph{$\F_n$} of $\N$ is a sound and deterministic negotiation.
\item[(2)] For every outcome $(n, \r)$, the $(n, \r)$-segment of $\N$ is a sound and deterministic negotiation.
\end{itemize}
\end{proposition}
\begin{proof}
We prove the result for the $n$-fragment \ph{$\F_n$}, the proof for the $(n, \r)$-segment is analogous.
Observe first that, by the definition of $n$- and $(n, \r)$-sequences, every agent of an atom of \ph{$\F_n$} 
is also an agent of $n$, and so \ph{$\F_n$} is indeed a negotiation.

For soundness, observe first that
every atom of \ph{$\F_n$} is executable: If $n'$ belongs to \ph{$\F_n$}, then by definition
some $n$-sequence contains $n$, and this sequence is also a sequence of \ph{$\F_n$}.

Now let $\sigma$ be an occurrence sequence of \ph{$\F_n$}. Then $\sigma$ is also an $n$-sequence
of $\N$, and \ph{it} can be extended to a maximal $n$-sequence $\sigma \tau$. By Proposition \ref{prop:unique}
we have $\vx_n \by{\sigma \, \tau} \vx_{\hat{n}}$ in \ph{$\F_n$}, and so 
$\vx_n \by{\sigma \, \tau \, (\hat{n}, \hat{\r})} \vx_f$.
\end{proof}

\begin{corollary}
A deterministic negotiation is sound if{}f all its fragments and segments are sound.
\end{corollary}
\begin{proof}
One direction follows immediately from Proposition \ref{prop:segfragsound}. For the other, it suffices
to observe that, since every agent is a party of the initial atom, the fragment of the initial atom is
 the complete negotiation.
\end{proof}

\subsection{The Summarization Algorithm}
\label{subsec:summalg}

\begin{algorithm}[t]
\caption{Summarization algorithm for a sound deterministic negotiation $\N$}
\begin{algorithmic}[1]
\State fix an arbitrary order $\prec$ on the atoms of $\N$ 
\For{$k = 1$ to $A$}
\While{$R(\N,k) \neq \emptyset$}
\If{possible} 
\State apply the merge rule to some outcome of $R(\N,k)$
\ElsIf{possible} 
\State apply the iteration rule to some outcome of $R(\N,k)$
\ElsIf{possible} 
\State apply the d-shortcut rule to some non-uniform outcome of $R(\N,k)$
\ElsIf{possible} 
\State apply the shortcut rule to a backward uniform outcome of $R(\N,k)$
\State  \ph{which is} minimal w.r.t. $\prec$
\Else 
\State apply the d-shortcut rule to some outcome of $R(\N,k)$
\EndIf
\State $\N := $ negotiation obtained after the application of the rule to $\N$
\EndWhile

\EndFor
\end{algorithmic}
\label{alg:redalg}
\end{algorithm}

We first present Algorithm \ref{alg:redalg}, an algorithm that summarizes a deterministic negotiation under 
the premise that it is sound. The general case is considered in Section  
\ref{subsec:algforunsound}. Algorithm \ref{alg:redalg} is a generalization of Algorithm \ref{alg:oneAgent} for one-agent negotiations. 
We start by explaining the meaning of \ph{$R(\N, k)$} in line 3 and backward outcome in line 11.

In the one-agent case all atoms have \ph{only one and thus} the same number of parties. Now this is no longer the case, and to cope 
with this the new algorithm proceeds in stages. During the $k$-th stage the algorithm 
only applies rules to outcomes $(n, \r)$ such that $n$ has $k$ parties. 

\begin{definition}
An outcome $(n, \r)$ is {\em $k$-reducible} if it is reducible and $n$ has $k$ parties. The set of $k$-reducible outcomes of $\N$ is denoted $R(\N, k)$. 
\end{definition}

Backward outcomes were introduced in Definition \ref{def:backward} for one-agent negotiations: An 
outcome $n \by{\r} n'$ is backward if $n' \neq n_f$ and $n' \prec n$. In the many-agent case we may have $\trans(n, p_1, \r) = \{n'\}$ and $\trans(n, p_2, \r) = \{n''\}$ for $n'' \neq n'$, and so not even the notation $n \by{\r} n'$ makes sense. 
But it does make sense when $\trans(n, p, \r) = \{n'\}$ for every party $p$ of $n$, i.e., for uniform outcomes 
(recall Definition \ref{def:uniform}). So we write $n \unif{\r} n'$ to denote that $(n, \r)$ is uniform and \ph{that}
$\trans(n, p, \r) = \{n'\}$ holds for \ph{each} $p \in P_n$, and generalize Definition \ref{def:backward} as follows:

\begin{definition}
A uniform outcome  $n \unif{\r} n'$ is {\em backward} if $n' \neq n_f$ and $n' \prec n$. We extend the order $\prec$ to
an order on uniform outcomes as in Definition \ref{def:backward}. 
\end{definition}

There is another difference between Algorithm \ref{alg:redalg} and Algorithm \ref{alg:oneAgent}. 
The \textbf{if}-command of lines 4-13 contains a new clause, which \ph{concerning the shortcut rule} gives preference to non-uniform outcomes 
over uniform ones. To understand the reason for this, consider the 
fragment on the left of Figure \ref{fig:fragmentExample}. The non-uniform outcomes are $(n_1, \texttt{a})$ and \ph{$(n_4, \texttt{b})$}. By giving them priority, we eliminate from \ph{this $n_1$-fragment} all atoms 
with strictly fewer parties than \ph{$n_1$} and ``make the fragment uniform''. 

The algorithm indeed summarizes the negotiation of Figure \ref{fig:cyclicRunning}
in the way sketched at the beginning of the section. We choose the order $\prec$ given by the 
numbering of the atoms. During the execution of the \textbf{while}-loop 
for $k=1$ the algorithm executes lines 5,7, and 11 to yield the negotiation on the right of Figure \ref{fig:stage1}. 
The execution for $k=2$ begins with applications of the d-shortcut rule
to the non-uniform outcomes $(n_1, a)$, and $(n_4, b)$ in line 9, which removes the atoms $n_2$ and $n_5$. 
After that, lines 5, 7, and 11 are executed to yield the negotiation on the right of Figure \ref{fig:stage2}.
Finally, the execution for $k=3$ applies the d-shortcut rule in line 9 several times, removing the atoms
$n_{15}$, $n_{67}$, and $n_8$, and leaving a negotiation with only the atoms $n_0$ and $n_f$. A final execution 
of line 13 finishes the summarization.

We conclude the section with a small lemma. Observe that the algorithm forbids to first apply a reduction to a 
$k$-reducible outcome and then to a $k'$-reducible outcome \ph{where} $k' < k$. This may seem dangerous:
In principle,  rule applications to $k$-reducible outcomes might {\em produce} new $k'$-reducible
outcomes, and the algorithm would not be able to ``use'' these outcomes for further reductions.
The \ph{following} lemma shows that this is not the case.

\begin{lemma}
\label{lem:kred}
Let $\N$ be a negotiation without $k$-reducible outcomes for any $k \leq K$, and let $\N'$ be obtained 
from $\N$ by an arbitrary sequence of reductions. Then $\N'$ \ph{also} has no $k$-reducible outcomes for any $k \leq K$.
\end{lemma}
\begin{proof}
Clearly, it suffices to prove the result for sequences of length 1. 
Assume the sequence applies a rule to a reducible outcome $(n, \r)$. Then $n$ has more than $k$ parties. 
If the rule is the iteration rule, 
then every reducible outcome of $\N'$ was already a reducible outcome of $\N$, and we are done.
In the case of the merge and shortcut rules, the only outcomes that might be reducible in $\N'$ but not in $\N$
are those added as a consequence of the rule. But these are outcomes of $n$, 
which has more than $k$ parties.
\end{proof}

This lemma already \ph{can be used to show} that Algorithm \ref{alg:redalg} summarizes all sound and deterministic
acyclic negotiations. Assume this is not the case. By Proposition \ref{prop:polcomp} such a negotiation is reducible by any maximal 
sequence of applications of the merge and d-shortcut rules, in particular by any in which we always choose a reducible outcome
with a minimal number of parties for the next rule application, and within those we give preference to the
merge over the shortcut rule. By Lemma \ref{lem:kred}, if $k < k'$ then such sequences never choose a $k'$-reducible outcome 
and then later a $k$-reducible one, and so the algorithm executes one of them.  

\subsection{Completeness}
\label{subsec:manycomp}

We prove a proposition at the core of our completeness result which states that Algorithm \ref{alg:redalg}
summarizes all sound deterministic negotiations. Intuitively, the proposition shows that the algorithm summarizes
the negotiation ``inside-out'', meaning that after $k$ iterations of the \textbf{for}-loop all fragments for atoms with 
$k$ parties have been summarized. Intuitively, the proof shows that before the $k$th iteration 
all the segments of these fragments are acyclic, and that the $k$-iteration performs the following two steps:
\begin{itemize}
\item Summarize all these segments using the merge and d-shortcut rules. This transforms the
fragments into replications.
\item Summarize the replications using Algorithm \ref{alg:oneAgent}.
\end{itemize}

\begin{proposition}
\label{prop:mainred}
Let $\N$ be a sound deterministic negotiation, and assume Algorithm \ref{alg:redalg}
terminates for input $\N$. For every $k \geq 0$, let
$\N^k$ be the negotiation produced by Algorithm \ref{alg:redalg} after $k$ iterations
of the for-loop. 
\begin{itemize}
\item[(1)] For every atom $n$ of $\N$ with at most $k$ parties, the $n$-fragment of $\N^k$ is atomic.
\item[(2)] For every outcome $(n, \r)$ of $\N$ with at most $k+1$ parties, the $(n, \r)$-segment 
of $\N^k$ is acyclic.
\end{itemize}
\end{proposition}
\begin{proof}
We prove (1) and (2) simultaneously by induction on $k$. 
If $k=0$ then (2) is vacuously true. For (1) just observe that if $n$ has one party
then the only $(n, \r)$-sequence is $(n, \r)$, and so the $(n, \r)$ segment is acyclic.

Assume now $k > 0$, and let $n$ be an atom of $\N^k$ 
with $\ell \leq k$ parties. We prove (1) in four steps.

\medskip

\noindent {\bf Claim 1}. All segments of the $n$-fragment are acyclic. \\
By the definition of a fragment, every atom of the $n$-fragment of $\N^{k+1}$ has at most $\ell$ parties. 
By induction hypothesis on (2), all segments of the $n$-fragment of $\N^{k}$ are acyclic, 
and therefore the same holds for $\N^{k+1}$. 

\medskip

\noindent {\bf Claim 2}. All segments of the $n$-fragment are atomic.\\
By Theorem \ref{prop:segfragsound}, the segments are sound and deterministic. 
Moreover, since every atom of a segment has at most $\ell$ parties and no outcome with
at most $k-1$ parties is reducible, all segments are irreducible.
By Claim 1 and Proposition \ref{prop:irredtheo}, all segments are atomic. 

\medskip

\noindent {\bf Claim 3}. The $n$-fragment is a replication.\\
Since every segment is atomic, all atoms of the $n$-fragment have the same 
set of parties as $n$. Now let $(n', \r')$ be an outcome of the $n$-fragment.
By Claim 2 the $(n', \r')$-segment is atomic, and so 
in the negotiation $\N^k$ we have $\vx_{n'} \by{(n', \r')} \vx$ for a marking $\vx$ such that
either $\vx = \vx_{n''}$ for some atom $n''$ with the same parties as $n'$, or
$\vx$ does not enable any atom. In the first case we have $n' \unif{\r'}  n''$,
and we are done. In the second case, by the definition of a segment we have 
$n' \unif{\r'}  \hat{n}$. 

\medskip

\noindent {\bf Claim 4}. The $n$-fragment is atomic. \\
Assume the $n$-fragment is not atomic. Then by Claim 3 and Theorem \ref{thm:repred} 
the \ph{n-}fragment has at least one reducible outcome, 
and the outcome has $\ell \leq k$ parties. This contradicts the 
hypothesis that the minimal reducible outcome has $k+1$ parties.

\medskip

Now we proceed to prove (2). Let $(n, \r)$ be an outcome of $\N$ with at most $k+1$ parties. 
Assume that the $(n, \r)$-segment is cyclic.
By Theorem \ref{prop:segfragsound}, the segment is sound and deterministic. 
By Proposition \ref{prop:loopsexist}, the segment has a loop, and by Proposition
\ref{prop:minSynchronized}, proved in the Appendix, the loop has a synchronizer $s$. By the definition of a segment,
$s$ is neither the initial atom $n$ nor the final atom of the segment. Since, by definition, every atom 
of the $(n, \r)$-segment different from the initial and final atoms has 
strictly fewer parties than $n$, we get $|P_s| \leq k$. By induction hypothesis on (1), the $s$-fragment is atomic. 
But this contradicts that $s$ is a synchronizer.
\end{proof}

\begin{theorem}
\label{thm:manyAgentComp}
Algorithm \ref{alg:redalg} summarizes every sound deterministic negotiation $\N$.
\end{theorem}
\begin{proof}
Let $K$ be the number of agents of $\N$. By Theorem \ref{prop:segfragsound}, after termination the
$n_0$-fragment is atomic. It follows that $\N^K$ is atomic, and we are done. 
\end{proof}

\subsection{Runtime Analysis}
\label{subsec:manyruntime}

It is easy to give a $O(K^4 + KL)$ bound for the number of rule applications of Algorithm \ref{alg:redalg}. 
Indeed, since each $n$-fragment can be summarized using
at most $2K^3 + K^2 + L$ applications and the negotiation has $K$ atoms, the bound follows.

We obtain a better, $O(K^3+KL)$ bound. The key observation is that\ph{,} for a given 
number $k$, the fragments of the atoms with $k$ parties may share nodes. Consider for instance 
the negotiation of Figure \ref{fig:multifragment},
and take $k=1$. The three atoms with one party are $n_1$, $n_2$, and $n_5$. 
The fragments ${\cal F}_{n_1}$ and ${\cal F}_{n_2}$ share the nodes $n_3$ and $n_4$.
Therefore, adding the number of rule applications needed to summarize the three fragments separately 
gives us a very crude bound. 

We define the $k$-fragment of $\N$ as, loosely speaking, the union of all the $n$-frgamnets for all
atoms $n$ with exactly $k$ parties. In Figure \ref{fig:multifragment} the $1$-fragment of $\N$ 
is shown at the bottom, on the right. We now formally define the $k$-fragment of a 
sound deterministic negotiation, and give a bound on the number of rule applications needed to summarize
it.

\begin{figure}[ht]
\centering
\input{tikz/multifragment}
\caption{A sound deterministic negotiation, its three one-party fragments, and their
combination into the $1$-fragment.}
\label{fig:multifragment}
\end{figure}

\begin{definition}
Let $\N=(N, n_0, n_f, \trans)$ be a \ph{deterministic} sound \ph{negotiation.}
Let $N^i$ be the set of atoms of $\N$ with exactly $i$ parties.
For every $n \in N_i$, let ${\cal F}_n = (F_n, n, \hat{n}, \trans_n)$ be the $n$-fragment,
with the fresh final atoms $\{ \hat{n} \mid n \in N_i \}$ chosen in such a way that
$\hat{n}_1 = \hat{n}_2$ holds for any two atoms $n_1, n_2 \in N^i$ satisfying $P_{n_1} = P_{n_2}$
(i.e., the fragments of two atoms have the same final atom if{}f they have the same parties).

The {\em $i$-fragment of $\N$} is the pre-negotiation $\F^i=(F,\trans^i)$ defined as follows
\begin{itemize}
\item $F = \bigcup_{n \in N^i} F_n$
\item for every $n \in F$, $p \in P_n$, $\r \in R_n$: 
$\ph{\trans^i(n,p,\r)} = \trans_{n'}(n,p,\r)$, where $n'$ is any atom with $i$ parties such that
$n \in F_{n'}$.
\end{itemize}
The $k$-fragment is summarized  if $\trans^i(n,p,\r) \neq \emptyset$ implies that $n \in N_i$ and 
\ph{$\trans^i(n,p,\r) = \{\hat{n}\}$}.
\end{definition}

\begin{lemma}
\label{lem:ifragment}
Let $\N$ be a sound  deterministic negotiation. Let $\F^i$ be the $i$-fragment of $\N$,
and let $K$ and $L$ be the number of nodes and outcomes of $\F^i$. If $\F^i$ is a replication,
then Algorithm \ref{alg:oneAgent} summarizes $\N^i$ after at most $2K^3 + K^2 + L$ rule applications.
\end{lemma}
\begin{proof}
\ph{By} the correspondence between replications and one-agent (pre-)negotiations,
it suffices to prove the result for $i=1$. The proof is exactly as the proof of Theorem 
\ref{thm:oneAgentPoly}. Indeed, an inspection shows that the proof does not use the fact 
that a negotiation has one single initial atom. The only point that has to be slightly 
adapted appears at the very end of the proof: A negotiation with
only the initial and final atom can be reduced to an atomic negotiation. On the contrary, a 
pre-negotiation with multiple initial atoms and one final atom cannot, because the shortcut 
rule cannot be applied. However, this is only a technicality, because the final atoms are 
atoms added artificially to give a destination to ``dangling outcomes''.
\end{proof}

\begin{theorem}
\label{thm:cyclicupperbound}
Let $\N$ be a sound deterministic negotiation with $K$ atoms, and $L$ outcomes. 
Algorithm \ref{alg:redalg} terminates after at most $2K^3 + K^2 + K L + L$ rule applications.
\end{theorem}
\begin{proof}
For every $i \geq 1$ let $k_i$ and $\ell_i$ be the number of atoms and outcomes of the 
atoms $n$ and outcomes $(n, \r)$ such that $n$ has exactly $i$ parties, respectively. Further,
let $K_i = \sum_{j=1}^{i} k_j$ and $L_i = \sum_{j=1}^{i} \ell_j$.
We have $K_{i+1} = K_i + k_{i+1}$ and $L_{i+1} = L_i + \ell_{i+1}$ for every $i \geq 1$.

We prove that for every $i \geq 1$ the $i$-th iteration of the \textbf{for}-loop \ph{terminates} after at most 
$2K_i^3 + K_i^2 + K_i L_i + L_i$ rule applications. 
We proceed by induction on $i$.

\medskip

\noindent {\bf Case $i = 1$}. During the first iteration of the loop  the 
algorithm summarizes the $1$-fragment of $\N$, and so the result follows by Lemma \ref{lem:ifragment}.

\medskip

\noindent {\bf Case $i > 1$}. By induction hypothesis, 
the $(i$--$1)$-th iteration of the for-loop has terminated after at most $2 K_{i-1}^3 + K_{i-1}^2 + K_{i-1} L_{i-1} + L_{i-1}$ rule applications. Let $\N^{i-1}$ be the negotiation obtained after the $(i$--$1)$-th iteration. By Proposition \ref{prop:mainred}, the fragments of $\N^{i-1}$ for atoms with at 
most $i-1$ parties are atomic, and the segments of outcomes with at most $i$ parties are acyclic. 

\noindent {\bf Claim.} We claim that after further $K_i \ell_i$ rule applications all segments for atoms with at most $i$ parties are not only acyclic but also atomic. 

Consider a non-atomic $(n, \r)$-segment such that $n$ has at most $i$ parties. 
If $n$ has less than $i$ parties, then the $n$-fragment is atomic, and the claim follows.
So assume that $n$ has exactly $i$ parties. By definition
all atoms of the  $(n, \r)$-segment different from $n$ have less than $i$ parties. We have:
\begin{itemize}
\item[(1)] $(n, \r)$ unconditionally enables some atom with fewer parties than $n$, that is,
$(n, \r)$ is non-uniform. \\
Otherwise the $(n, \r)$-segment is not sound, contradicting the soundness of $\N$.
\item[(2)] The index of the $(n, \r)$-segment is equal to or smaller than its number of atoms.\\
Recall that the index of an outcome $(n', \r')$ in the $(n, \r)$-segment is the length of a longest maximal $(n', \r')$-sequence, and the index of the $(n, \r)$-segment is the sum of the indices of its atoms
(Definition \ref{def:index}). Since the $(n, \r)$-segment is acyclic, the index of $(n, \r)$ 
is at most the number of atoms of the segment minus $1$. Consider now any other outcome $(n', \r')$ of the $(n, \r)$-segment. Then $n' \neq n$, and so $n'$ has fewer parties
than $n$. But then the $(n',\r')$-segment is atomic, which implies that the index of
$(n', \r')$ is $0$. So the total index of the segment is at most equal to its number of atoms.
\end{itemize}  
By (1), the algorithm keeps choosing non-uniform outcomes of the segments with $i$ parties,
until there are no more non-uniform atoms. At this point, all $(n, \r)$-segments are irreducible,
and by Proposition \ref{prop:irredtheo} atomic. Since there at most $\ell_i$ segments, each of them
with at most $K_i$ atoms, by (2) and Lemma \ref{lem:lemacyc} the segments become atomic after at most 
$K_i \cdot \ell_i$ rule applications. This finishes the proof of the claim.

\medskip

Let $\hat{N}^{i-1}$ be the negotiation obtained after the sequence of $K_i \cdot \ell_i$ rule applications. By the claim, the 
fragments of atoms of $\hat{N}^{i-1}$ with $(i-1)$ parties and the segments of atoms with $i$ parties are atomic. 
It follows that the outcomes of $\hat{N}^{i-1}$ that have exactly $i$ parties are uniform. As shown in Claim 3 in the 
proof of Proposition \ref{prop:mainred}, this implies that all fragments for atoms with $i$ parties are 
replications. So the complete $i$-fragment of $\hat{N}^{i-1}$ has at most $k_i$ atoms 
and $\ell_i$ outcomes. By Lemma \ref{lem:ifragment}, the $i$-fragment is summarized after at
 most $2k_i^3+k_i^2+\ell_i$ rule applications.

So we obtain that the $i$-th iteration of the for-loop  \ph{terminates} after at most 
$$(2 K_{i-1}^3 + K_{i-1}^2 + K_{i-1}L_{i-1}+L_{i-1}) +  K_i \ell_i + (2 k_i^3+k_i^2+\ell_i)$$ iterations.
Now, using $K_i = K_{i-1} + k_i$ and $L_i = L_{i-1} + \ell_i$, we get: 

$$\begin{array}{rcl}
&    & (2 K_{i-1}^3 + K_{i-1}^2 + K_{i-1} L_{i-1}+L_{i-1}) + K_i \cdot \ell_i + (2 k_i^3+ k_i^2+\ell_i) \\[0.2cm]
& =  & 2 (K_{i-1}^3 + k_i^3) + (K_{i-1}^2 + k_i^2) + (K_{i-1}  L_{i-1} + K_i \cdot \ell_i)+ (L_{i-1} +\ell_i)  \\[0.2cm]
& \leq & 2(K_{i-1} + k_i)^3 + (K_{i-1} + k_i)^2 + K_{i}(L_{i-1}+\ell_i) + (L_{i-1} +\ell_i) \\[0.2cm]
& = & 2 K_i^3 + K_i^2 + K_iL_i +L_i
\end{array} $$
\noindent and we are done.
\end{proof}

\subsection{The Unsound Case}
\label{subsec:algforunsound}

Algorithm \ref{alg:redalg2} assumes that its input is a sound \ph{deterministic} negotiation. It is easy to transform it
into an algorithm for arbitrary deterministic negotiations, with the same complexity. 
The algorithm is shown as Algorithm \ref{alg:redalg2}. 
If the input negotiation is sound, Algorithm \ref{alg:redalg2}
algorithm summarizes it, and otherwise \ph{it} answers ``unsound''. The algorithm
counts the number of rule applications and answers ``unsound'' if \ph{the counter} exceeds $2K^3 + K^2 + KL + L$. 
This is correct by Theorem 
\ref {thm:cyclicupperbound}. The algorithm also reports ``unsound'' if it reaches an irreducible 
but non-atomic net, or if it can only apply the shortcut rule to non-uniform or uniform forward outcomes.

\begin{algorithm}[t]
\caption{Summarization algorithm for arbitrary deterministic negotiations}
\begin{algorithmic}[1]
\State fix an arbitrary order $\prec$ on the atoms of $\N$ 
\State $\mathit{counter} := 0$
\For{$k = 1$ to $A$}
\While{$R(\N,k) \neq \emptyset$ and $\mathit{counter} \leq 2K^3 + K^2 + KL + L$}
\If{possible} 
\State apply the merge rule to some outcome of $R(\N,k)$
\ElsIf{possible} 
\State apply the iteration rule to some outcome of $R(\N,k)$
\ElsIf{possible} 
\State apply the d-shortcut rule to some non-uniform outcome of $R(\N,k)$
\ElsIf{possible} 
\State apply the shortcut rule to a backward uniform outcome of $R(\N,k)$
\State  minimal w.r.t. $\prec$
\ElsIf{possible}
\State apply the d-shortcut rule to some outcome of $R(\N,k)$
\Else \  stop and answer ``unsound''
\EndIf
\State $\N := $ negotiation obtained after the application of the rule to $\N$
\State $\mathit{counter} := \mathit{counter}+1$
\EndWhile
\EndFor
\If{$\N$ is not atomic} answer ``unsound'' \EndIf
\end{algorithmic}
\label{alg:redalg2}
\end{algorithm}

\subsection{A Correction}

In \cite{negII} we \ph{claimed} that the merge, iteration, and shortcut rule were 
complete for sound deterministic negotiations, and presented a reduction algorithm.
While the completeness result holds, the algorithm of \cite{negII} is 
unfortunately incorrect. The algorithm is based on a lemma (lemma number 3 in \cite{negII}), 
which, rewritten in the terminology of this paper, states:

\begin{lemma}[Incorrect Lemma 3 of \cite{negII}]
A cyclic sound deterministic negotiation $\N$ contains an atom $n$ such that all loops
of the $n$-fragment of $\N$ contain $n$.
\end{lemma}

This lemma was used in \cite{negII} to prove the correctness of the 
summarization procedure that iterates the following three steps. First: identify an $n$-fragment
satisfying the property stated in the lemma. Second: use the procedure for summarizing acyclic 
negotiations to reduce this $n$-fragment to a single atom, with possibly some self-loop outcomes
(i.e., outcomes $(n, \texttt{r})$ such that $\trans(n, a, \texttt{r}) = n$ for every party $a$ of $n$).
Third, use the iteration rule to remove such outcomes. 

However, the lemma is wrong. Figure \ref{fig:wrong} shows a one-agent negotiation in which,
for every atom $n$, the $n$-fragment contains a loop which does not contain $n$.  As a consequence,
the algorithm of \cite{negII} does not summarize this negotiation. 

\begin{figure}[ht]
\centering
\input{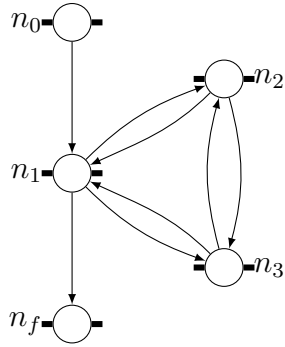}
\caption{Counterexample to Lemma 3 of \cite{negII}.}
\label{fig:wrong}
\end{figure}

In this paper we have corrected this mistake. 
The procedure of 
Section \ref{sec:cyclic} correctly summarizes the negotiation of Figure \ref{fig:wrong}.

\section{Conclusion}
\label{sec:conclusion}

We have introduced negotiations, a formal model of concurrency with
atomic negotiations (a form of synchronized choice) as primitive. 
We have defined and studied two important analysis problems: soundness, which coincides 
with the soundness notion for workflow nets, and the new {\em summarization problem}. 
We have provided a complete set of \ph{reduction} rules for sound deterministic negotiations,
and we have shown that, when applied following a simple strategy, the rules allow one 
to compute a summary of a sound deterministic negotiations with a polynomial number
of rule applications.

The reduction algorithm is based on several results about the structure of deterministic 
negotiations that have independent interest. Proposition \ref{prop:minSynchronized} and its syntactic
counterpart, Proposition \ref{prop:dominating}, are important results about the structure of loops.
Proposition \ref{prop:unique} shows that every sound deterministic negotiation is necessarily
organized compositionally. 

\paragraph{Acknowledgments.} The proof of Lemma \ref{lem:same} was obtained with the help of
Anca Muscholl and Igor Walukiewicz, whom we also thank for very helpful comments and
suggestions.

\bibliographystyle{alpha} 
\bibliography{references}

\appendix
\newpage 

\section{Appendix: Complexity of Soundness}
\label{app:complex}

\begin{reftheorem}{thm:complexity}
The soundness problem is PSPACE-complete. For acyclic negotiations, the problem is 
co-NP-hard and in DP (and so at level $\Delta^P_2$ of the polynomial hierarchy).
\end{reftheorem}
\begin{proof}

\noindent
{\em The soundness problem is in PSPACE.} 

\noindent
Membership in PSPACE could be proved by observing that the soundness problem can be formulated in CTL, and then applying the PSPACE algorithm for CTL and 1-safe Petri nets of \cite{DBLP:conf/ac/Esparza96}. This algorithm only assumes that, given a marking, one can compute a successor marking in polynomial time, which is the case both for Petri nets and for negotiations. However, since we only need a very special case of the CTL algorithm, we provide a self-contained proof.

We show that both conditions for soundness can be checked in nondeterministic polynomial space. The result then follows from Savitch's theorem (NPSPACE=PSPACE).

The first condition is: every atom is enabled at some reachable
marking. For this we consider each atom $n$ in turn, and guess step
by step an occurrence sequence ending with an occurrence of $n$.
This only requires to store the marking reached by the sequence executed so far.

The second condition is: every occurrence sequence from the initial
marking is either a large step or can be extended to a large step.
This case is a bit more involved. Let $S$ denote the problem of
checking this second condition. We prove $S \in \text{PSPACE}$.

\begin{itemize}
\item[(1)] The following problem is in PSPACE: given some
marking $\vx$, check that no occurrence sequence starting at $\vx$
ends with the final atom. \\
Let us call this problem NO-OCC. We have
$\overline{\text{NO-OCC}} \in \text{NPSPACE}$, because we can nondeterministically guess
an occurrence sequence starting at $\vx$ that
ends with the final atom (we guess one step at a time). Since
NPSPACE=PSPACE=co-PSPACE, we get $\text{NO-OCC} \in
\text{PSPACE}$.\\

\item[(2)] $\overline{S} \in \text{NPSPACE}$. \\
$\overline{S}$ consists of
checking the existence of a sequence $\sigma$, fireable from the
initial marking, that is neither a large step nor can be extended to
it. For this we guess a sequence $\sigma$ step by step that does not
end with the final atom. Then we consider the marking $\vx$ reached by
the occurrence of $\sigma$. Clearly, we have $\sigma \in \overline{S}$
if{}f $\vx \in \text{NO-OCC}$. So it suffices to apply our
deterministic polynomial-space algorithm for NO-OCC (see (1)).\\

\item[(3)] $S \in PSPACE$.\\
Follows from (2) and NPSPACE=PSPACE=co-PSPACE.
\end{itemize}

\noindent
{\em The soundness problem is PSPACE-hard.} 

\noindent
For PSPACE-hardness,
we reduce the problem of deciding if a
deterministic linearly bounded automaton (DLBA) recognizes an input to
the soundness problem. Let $A=(Q,\Sigma,\delta,q_0,F)$ be a DLBA, and
consider an input $w=a_1 \ldots a_k \in \Sigma^*$. The construction is
very similar to that of \cite{DBLP:conf/ac/Esparza96} for proving PSPACE-hardness of the reachability problem for 1-safe Petri nets, and so we do
not provide all details. The negotiation $\N_A$ has a control
agent $C$, a head agent $H$, and a cell agent $T_i$ for every tape
cell (i.e., $1 \leq i \leq k$). All agents have only one internal
state, i.e., the internal states are irrelevant. The negotiation has
an atom $n[q,h,a]$ (with only one result) for every
state $q$, every head position $1 \leq h \leq k$, and every $a \in
\Sigma$, plus an initial atom $n_0$ and a final atom $n_f$. The
parties of $n[q,h,a]$ are $C$, $H$, and $T_h$. The transition function
$\trans$ is defined so that $\N_A$ simulates $A$ in the following
sense: If $A$ is currently in state $q$ with the head at position $h$,
and the contents of the tape are $b_1 \ldots b_k$, then the current
marking $\vx$ of the negotiation satisfies the following properties:

\begin{itemize}
\item if $q\neq q_f$, then $\vx(C)$ is the set of atoms $n[h',q', a]$
such that $q'=q$, and both $h'$ and $a$ are arbitrary; if $q = q_f$,
then $\vx(C) = \{ n_f\}$;

\item $\vx(H)$ is the set of atoms $n[h',q', a]$ such that $h'=h$ and $q',a$ are arbitrary, plus the final atom;

\item $\vx(T_i)$ is the set of atoms $n[h',q', a]$ such that $h'=i$, $q'$ is arbitrary, and $a=b_i$, plus the final atom.
\end{itemize}

\noindent Intuitively, agent $C$ is only ready to
engage in atoms for the state $q$; agent $H$ is only ready to engage
in atoms for the position $h$; and $T_h$ is only ready to engage in
atoms for the letter $b_h$. These properties guarantee that the only atom enabled by $\vx$ is
$n[h,q,b_h]$ if $q \neq q_f$, or the atom $n_f$ if $q=q_f$. So the
negotiation $\N_A$ has only one initial occurrence sequence, which
corresponds to the execution of $A$ on $w$.

It remains to define $\trans$ so that it satisfies these properties. For the initial atom we take (recall that the input of the DLBA $A$ is the word $w=a_1 \ldots a_k$):

$$\begin{array}{rcl}
\trans(n_0, C, \mathtt{step}) & = & \{n[h',q_0, a'] \mid 1 \leq h' \leq k, a' \in \Sigma \} \\
\trans(n_0, H, \mathtt{step}) & = & \{n[1,q', a'] \mid q' \in Q, a' \in \Sigma \} \\
\trans(n_0, T_{i}, \mathtt{step}) & = & \{n[i,q_0, a_i] \}
\end{array}$$

For the transition function of an atom $n[q,h,a]$ we must consider the three possible cases of the transition relation
(head moves to the right, to the left, or stays put). We only deal with the case in which the machine moves to the
right, the others being analogous. Assume $\delta(q,a)=(\hat{q},\hat{a},R)$.
Then we take
$$\begin{array}{rcl}
\trans(n[h,q,a], C, \mathtt{step}) & = & \{n[h',\hat{q}, a'] \mid 1 \leq h' \leq k, a' \in \Sigma \} \\
\trans(n[h,q,a], H, \mathtt{step}) & = & \{n[h+1,q', a'] \mid q' \in Q, a' \in \Sigma \} \\
\trans(n[h,q,a], T_{h}, \mathtt{step}) & = & \{n[h,q', \hat{a}] \mid q' \in Q \}
\end{array}$$

Since $A$ is deterministic, $\N_A$ has only one maximal occurrence sequence, which is
a large step if{}f $A$ accepts. So $\N_A$ is sound if{}f $A$ accepts.\\

\noindent
{\em The soundness problem for acyclic negotiations is in DP.}

\noindent
We first observe that no occurrence of an acyclic negotiation contains an atom more than once
(loosely speaking, once the tokens of the parties of the atom have ``passed'' beyond it, they cannot return). It follows that the length of an occurrence sequence is at most equal to the number of atoms. It also follows that there are no livelocks, but there may be deadlocks. To check soundness we must
check that (1) every atom can be enabled, and that (2) every occurrence sequence can be extended to a large step.
Checking (1) can be done by guessing
in polynomial time enabling sequences for all atoms, and so (1) is in NP.
Checking the negation of (2) can be done by guessing in polynomial
time an occurrence sequence that cannot be extended to a large step because it leads to a deadlock,
and so (2) is in coNP. So the conjunction of (1) and (2) is in DP.\\

\noindent
{\em The soundness problem for acyclic negotiations is co-NP-hard.}

\noindent We reduce 3-CNF-UNSAT to soundness.
Given a boolean formula $\phi$ with variables
$x_i$, $1 \leq i \leq n$ and clauses $c_j$, $1 \leq j \leq m$, we construct a
negotiation $\N_\phi$ with an agent $X_i$ for each $x_i$, and
an agent $J$ (for judge). W.l.o.g. we assume that no clause of $\phi$ is a tautology.
For each variable $x_i$, $\N_\phi$ has an atom
$\mathit{Set\_x}_i$ with $X_i$ as only party and results $\texttt{true}$ and $\texttt{false}$.
For each clause $c_j$, the negotiation $\N_\phi$ has an atom
$\mathit{False}_j$ whose parties are the variables appearing in $c_j$ and the judge $J$.
The atom has only one result $\texttt{false}$.

After the initial atom, agent $X_i$ engages in $\mathit{Set\_x}_i$ and sets $x_i$ to a value
$b \in \{\texttt{true},\texttt{false}\}$ by choosing the appropriate result. After that, $X_i$
is ready to engage in the atoms $\mathit{False}_j$ satisfying the following condition:
the clause $c_j$ is {\em not} made true by setting $x_i$ to $b$; moreover, it
is also ready to engage in the final atom. As a consequence, $\mathit{False}_j$ becomes enabled if{}f the assignment chosen by the $X_i$'s makes $c_j$ false.
Finally, after the occurrence of a $\mathit{False}_j$, its parties are only ready to engage in the final atom.

After the initial atom, the judge $J$ is ready to engage in all atoms $\mathit{False}_j$, and then, if any of them occurs, in the final atom.

We argue that $\N_\phi$ is sound if{}f $\phi$ is unsatisfiable.
Notice first that, since by assumption no clause is
a tautology, every $\mathit{False}_j$ atom is enabled by some occurrence sequence.
So all atoms but perhaps the final atom can be enabled by some sequence. So
$\N_i$ is sound if{}f every occurrence sequence can be extended to a large step,
and therefore it suffices to show that $\phi$ is unsatisfiable if{}f
every occurrence sequence of $\N_\phi$ can be extended to a large step.

If $\phi$ is unsatisfiable then, whatever the assignment determined
by the outcome of the $\mathit{Set\_x}_i$'s, some clause is
false, and so at least one of the $\mathit{False}_j$ atoms is enabled. After
some $\mathit{False}_j$ occurs, the final atom becomes enabled, and so
the computation can be extended to a large step.

If $\phi$ is satisfiable, then consider an initial occurrence sequence in which the atoms $\mathit{Set\_x}_i$ occur, and \ph{then} choose the outcomes corresponding to a satisfying assignment. \ph{This way} none of the $\mathit{False}_j$ atoms become enabled.
Moreover, the final atom is not enabled either, because the judge $J$ is not ready to engage
in it. So the occurrence sequence cannot be extended to a large step.
\end{proof}

\newpage 

\section{Appendix: A Lemma on Irreducible Acyclic Negotiations}
\label{app:irred}

\begin{reflemma}{lem:acyclicpoly1} 
Let $\N$ be an irreducible sound and deterministic acyclic negotiation and 
let $n \neq n_f$ be an atom of $\N$ with more than one result.
Then every agent participates in $n$. 
\end{reflemma}
\begin{proof}
We proceed in two steps.\\[0.2cm]
\noindent (a) The atom $n$ has a result $\r$ such that either $(n,\r)$ unconditionally enables $n_f$, 
or $(n,\r)$ unconditionally enables some atom with more than one result.

We first prove a preliminary claim: 
if some outcome $(n,\r)$ unconditionally 
enables some atom, then (a) holds. 
Indeed: if $(n,\r)$ unconditionally enables some atom $n'$, then either $n'=n_f$ or $n'$ 
has more than one result, because otherwise the d-shortcut rule can be applied to 
$n$ and $n'$, contradicting the irreducibility of $\N$. This proves the claim. 

It remains to show that some outcome $(n,\r)$ unconditionally enables some atom. For this,
we assume the contrary, and prove that $\N$ contains a cycle, contradicting that $\N$ is acyclic.

Since the merge rule is not applicable to $\N$, the atom
$n$ has two results $r_1, r_2$ such that $\trans(n,a,r_1) \neq \trans(n,p,r_2)$ 
for some party $a$. We proceed in three steps.

\medskip

(a1) For every reachable marking $\vx$ that enables $n$ there is a
sequence $\sigma$ such that $\vx \by{(n,r_1) \, \sigma} \vx_1$ and $\vx \by{(n,r_2) \, \sigma} \vx_2$ for some markings $\vx_1, \vx_2$ such that the sets $N_1$ and $N_2$ of 
atoms enabled by $\vx_1, \vx_2$ are nonempty and disjoint.\\
Let $\sigma$ be a longest occurrence sequence such that $\vx \by{(n,r_1)\, \sigma} \vx_1$ 
and $\vx \by{(n,r_2)\, \sigma} \vx_2$ for some markings $\vx_1, \vx_2$ (notice that
$\sigma$ exists, because all occurrence sequences of $\N$ are finite by acyclicity).
We have $N_1 \cap N_2 = \emptyset$, because otherwise we can extend $\sigma$ with the 
occurrence of any atom enabled by both markings.
We prove that, furthermore, $N_1 \neq \emptyset \neq N_2$. Assume w.l.o.g. $N_1 = \emptyset$.
Then, since $\N$ is sound, we have $\vx_1 = \vx_f$, which means that the last 
step of $\sigma$ is of the form $(n_f, r_f)$. So $\vx_2$ is also a marking obtained after 
the occurrence of $(n_f,r_f)$. Since every agent participates in $n_f$ and $\trans(n_f,p,r_f)=\emptyset$ for every
agent $p$ and result $r_f$, we also have $\vx_2 = \vx_f$. So
$\vx \by{(n,r_1)} \vx_1' \by{\sigma} \vx_f$ and $\vx \by{(n,r_1)} \vx_2' \by{\sigma} \vx_f$, which implies $\vx_1' = \vx_2'$. Since $\N$ is deterministic, we then have $\trans(n,p,r_1) = \trans(n,p,r_2)$, contradicting the hypothesis.

\medskip

(a2) For every $n_1 \in N_1$ there is a path leading from some $n_2 \in N_2$ to $n_1$, and
for every $n_2 \in N_2$ there is a path leading from some $n_1 \in N_1$ to $n_2$.\\
By symmetry it suffices to prove the first part. Since $N_1$ and $N_2$ are disjoint, 
$n_1$ is enabled at $\vx_1$ but not at $\vx_2$. Moreover, since $\N$ is acyclic, every
atom can occur at most once in an occurrence sequence, and so neither $n_1$ nor $n_2$ appear
in $\sigma$. Since, furthermore, the sequences $(n,r_1)\, \sigma$ and $(n,r_2)\, \sigma$ only differ in their first element, there is an agent $p$ such that 
$\trans(n,p,r_1) = \{n_1\}$ and $\trans(n,p,r_2)=\{n_2'\}$ for some $n_2' \neq n_1$ ($n_2'$ is not necessarily the $n_2 \in N_2$ we are looking for). 
So we have $\vx_1(p) = \{n_1\}$ and $\vx_2(p)= \{n_2'\}$ (see Figure \ref{fig:polyproof}).

 \begin{figure}[h]
 \centerline{\scalebox{0.45}{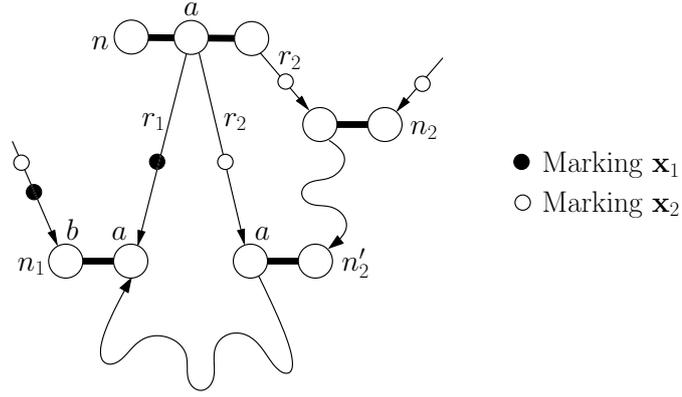}}
 \caption{Illustration of the proof of Lemma \ref{lem:acyclicpoly1}.}
 \label{fig:polyproof}
 \end{figure}

We first show that there is a path from $n_2'$ to $n_1$.
By assumption, no outcome of $n$ unconditionally enables any atom, and so 
$(n,r_1)$ does not unconditionally enable $n_1$.
So $n_1$ has a participant $q \neq p$ such that either $q$ is not a participant of $n$ or
$\trans(n,q,r_1) \neq n_1$. Since $\vx_1$ enables $n_1$ we have $\vx_1(q)= \{n_1\}$, and since
$q$ is not a participant of $n$ or $\trans(n,q,r_1) \neq n_1$, we have $\vx_2(q)=\{n_1\}$ as well.
Since $\vx_2(q)=n_1$, and $\N$ is a sound and deterministic acyclic negotiation, there is an occurrence sequence 
$\tau$ such that $\vx_2 \by{\tau} \vx_2'$ and $\vx_2'$ enables $n_1$ (intuitively,
the white token on the arc to the $q$-port can only leave the arc through the occurrence of $n_1$). Since $\vx_2(p)=\{n_2'\} \neq \{n_1\}$, there is a path from $n_2'$ to $n_1$ (intuitively,
 the white token on the arc leading to the $p$-port of $n_2'$ has to travel
to some arc leading to the $p$-port of $n_1$, and by determinism it can only do so
through a path of $p$-ports that crosses $n_2'$).

We now prove that there is a path from some $n_2 \in N_2$ to $n_2'$. If $n_2'$ is enabled at $\vx_2$, then $n_2' \in N_2$ and we are done. If $n_2'$ is not enabled at $\vx_2$ (as in the figure), then, since $\vx_2(p)= \{n_2'\}$ and $\N$ is a sound and deterministic acyclic negotiation,
there is a sequence $\tau$ such 
that $\vx_2 \by{\tau} \vx_2'$ and $\vx_2'$ enables $n_2'$ (again, by soundness 
the white token on the arc to the $p$-port of $n_2'$ can eventually leave the arc to move towards $n_f$, and by determinacy it can only leave the arc through the occurrence of $n_2'$). 
Since $N_2$ is the set of transitions enabled at $\vx_2$, we have $\tau = (n_2, \r) \, \tau'$ for some $n_2 \in N_2$. So some subword of $\tau$ is a path from some transition of $N_2$ to $n_2'$.

\medskip

(a3) $\N$ contains a cycle. \\
Follows immediately from (a2) and the finiteness of $N_1$ and $N_2$.

\medskip

\noindent (b) Every agent participates in $n$.\\
By repeated application of (a) we find a chain $(n_1, r_1) \ldots (n_k, r_k)$ such that
$n_1 = n$, $n_k = n_f$, and $(n_i, r_i)$ unconditionally enables $n_{i+1}$ for  $1 \leq i \leq k-1$. By the definition of an unconditionally enabled atom we have 
$P_1 \supseteq P_2 \supseteq \cdots \supseteq P_k = P_f$. Since $P_f = \agents$, 
we obtain $P_1 = \agents$.
\end{proof} 

\newpage

\section{Appendix: Unique Targets of Maximal Sequences}
\label{app:fragseq}

The proof of Proposition \ref{prop:unique} is very involved, and requires some preliminaries.
In Section \ref{subsec:loopsynch} we introduce the notions of loop of a negotiation, and synchronizer
of a loop. Loosely speaking, a loop is an occurrence sequence leading
from a marking to itself, and a synchronizer of a loop is an atom having as parties all the agents that 
participate in any of the atoms of the loop.
We prove a fundamental result stating that every loop of a sound deterministic negotiation
has a synchronizer (Proposition \ref{prop:loopsexist}). In Section \ref{subsec:loopsynch} we
use this result to prove Proposition \ref{prop:dominating}, stating that every cycle of a negotiation 
has a dominating atom: An atom having as parties all the agents that participate in any of the atoms 
of the cycle.

Equipped with these propositions, in Section \ref{subsec:unique}
we proceed to prove Proposition \ph{\ref{prop:unique}}.

\subsection{Loops and Synchronizers}
\label{subsec:loopsynch}

Intuitively, a loop is an occurrence sequence leading back to the marking from which it started.

\begin{definition}
\label{def:loop}
A loop of a negotiation $\N$ is a \ph{nonempty} sequence $\sigma$ of outcomes such 
that $\vx \by{\sigma} \vx$ for some marking $\vx$ of $\N$. We say that $\vx$ {\em enables} $\sigma$.
Given a loop $\sigma$, we denote $N_\sigma$ the set of atoms of $\N$ that appear in $\sigma$, and
$A_\sigma$ the set of agents that participate in at least one of these atoms.
A loop $\sigma$ is minimal if no loop $\sigma'$ satisfies $N_{\sigma'} \subset N_\sigma$. 
\end{definition}

The sequences 
$$\begin{array}{rcl}
\sigma_1 & = & (n_1, \texttt{a})\, (n_2, \texttt{a}) \, (n_4, \texttt{b}) \, (n_5, \texttt{a})\\
\sigma_2 & = & (n_6, \texttt{a}) \, (n_7, \texttt{b}) \\
\sigma_3 & = & (n_1, \texttt{b}) \, (n_3, \texttt{a}) \, (n_4, \texttt{c}) \, (n_3, \texttt{a}) \, 
(n_4, \texttt{b}) \, (n_5, \texttt{a})
\end{array}$$
\noindent are loops of the negotiation of Figure \ref{fig:cyclicRunning}.  We have 
$N_{\sigma_3} = \{ n_1, n_3, n_4, n_5 \}$. Numbering the agents of the negotiation from left to right,
the set $A_{\sigma_2}$  contains only the third agent, while $A_{\sigma_1}$ and $A_{\sigma_3}$
contain the first and second agents.

\begin{definition}
\label{def:synchronizer}
Let $\sigma$ be a loop of a negotiation. An atom $n \in N_\sigma$ is a 
{\em synchronizer of $\sigma$} if $P_n' \subseteq P_n$ for every atom  $n' \in N_\sigma$. We say that 
$\sigma$ is {\em synchronized} by $n$. An atom is a {\em synchronizer}
if it synchronizes at least one loop of $\N$.
\end{definition}

The loop $\sigma_1$ above is synchronized by $n_1$ and $n_4$, but not by $n_2$ or $n_6$. 
The atoms $n_1, n_3, n_4, n_6, n_7$ are synchronizers of the negotiation, while $n_0$, $n_2$, 
$n_5$, $n_8$, and $n_f$ are not. Indeed, the atoms $n_0$, $n_8$, and $n_f$ do not 
belong to any loop. Atoms $n_2$ and $n_5$ do belong to loops, but only to loops that 
contain atoms with strictly more parties.

We prove two important properties of a sound cyclic deterministic negotiation $\N$:
Such a negotiation always has at least one loop, and every loop has at least one synchronizer.

\begin{proposition} 
\label{prop:loopsexist}
Every sound cyclic deterministic negotiation has a loop.
\end{proposition}
\begin{proof}
Let $\pi$ be a cycle of the graph of a sound deterministic negotiation $\N$. Let $n_1$ be an 
arbitrary atom occurring in $\pi$, and let $n_2$ be its successor in $\pi$. 
Observe that $n_1 \neq n_f$ because $n_f$ has no successor, and hence no cycle contains $n_f$.
By soundness some reachable marking $\vx_1$ enables $n_1$, and so  
$n_2 \in \trans (n_1,p,\r)$ for at least one party $p$ and one 
result $\r$. By determinism, we have $n_2 \in \trans (n_1,p,\r) = \{n_2\}$. 
Let $\vx_1 \by{(n_1,\r)} \vx_1'$. Again by soundness, some occurrence sequence leads from 
from $\vx_1'$ to the final marking. This sequence must contain an occurrence of $n_2$, because
this is the only atom agent $p$ is ready to engage in. In particular, some prefix of this sequence 
leads to a marking $\vx_2$ that enables $n_2$.

Repeating this argument for the nodes $n_1$, $n_2$, $n_3$, \ldots , $n_k = n_1$ of the cycle $\pi$, 
we conclude that $\N$ has an infinite occurrence sequence
containing infinitely many occurrences of atoms of $\pi$. Since the set of reachable markings 
is finite, this sequence contains a loop. 
\end{proof}

Now we prove that every minimal loop has a synchronizer.
We start with a technical lemma. Intuitively, it states
that firing an outcome of a loop $\sigma$ never decreases the set of agents ready to engage in atoms of $N_\sigma$. 

\begin{lemma} 
\label{lem:basicloop}
Let $\sigma$ be a loop of a negotiation, let $(n, \r)$ be an 
outcome appearing in $\sigma$, and let $\vx_1, \vx_2$ be arbitrary markings such that $\vx_1 \by{(n, \r)} \vx_2$. 
For every agent $p$:  if $\vx_1(p) \cap N_\sigma \neq \emptyset$, then $\vx_2(p) \cap N_\sigma \neq \emptyset$. 
\end{lemma}
\begin{proof}
 Let $\vx_1 \by{(n, \r)} \vx_2$ and assume $\vx_1(p) \cap N_\sigma \neq \emptyset$. 
If $p$ does not participate in $n$, then $\vx_1(p) = \vx_2(p)$, and so $\vx_2(p) \cap N_\sigma \neq \emptyset$.
So assume that $p$ participates in $n$. Then, since $\vx_1 \by{(n, \r)} \vx_2$, we have 
$n \in \vx_1(p)$. By the definition of a loop, there is a sequence
$\tau$ and reachable markings $\vx,\vx'$ such that $\vx \by{(n,\r)} \vx' \by{\tau} \vx$
(the sequence $(n, \r) \, \tau$ is a circular permutation of $\sigma$).
Since we have $\vx_2(p) = \vx'(p)$, it suffices to show $\vx'(p) \cap N_\sigma \neq \emptyset$.
If $n \in \vx'(p)$, then we are done. If $n \notin \vx'(p)$ then, since $n \in \vx(p)$
and $\vx' \by{\tau} \vx$, the sequence $\tau$ contains an outcome $(n', \r')$ such that
$n' \in \vx'(p)$. Since, by definition, only atoms of $N_\sigma$ occur in $\tau$, we have
$n' \in N_\sigma$, and so $\vx'(p) \cap N_\sigma \neq \emptyset$.
\end{proof}

\begin{proposition}
\label{prop:minSynchronized}
Every minimal loop of a sound deterministic negotiation is synchronized by at least one of its atoms.
\end{proposition}
\begin{proof}
Let $\sigma$ be a minimal loop enabled at a reachable marking $\vx$. 
Since $\N$ is sound, there is an occurrence sequence $\vx \by{\sigma_f} \vx_f$.
Let $\sigma_f = (n_1, \r_1) \, \ldots, (n_k, \r_k)$. Choose an arbitrary agent $\hat{p}$ of $A_\sigma$, and let $n_{j_0}$ be the atom of $N_\sigma$ with 
$\hat{p}$ as party that appears {\em last} in $\sigma_f$, i.e.
$n_{j_0} \in N_\sigma$ and $n_{j_0+1}, \ldots, n_k \notin N_\sigma$. 
We claim that $n_{j_0}$ is a synchronizer of $\sigma$.  

\ph{We iteratively 
construct} a path $\pi$ of the graph of $\N$
containing only atoms having $\hat{p}$ as party. 
The first atom of the path is $n_{j_0}$.  Assume now that the path we have 
constructed so far is  $(n_{j_0}, \hat{p}, \r_0) \, (n_{j_1}, \hat{p}, \r_1) \ldots (n_{j_i}, \hat{p}, \r_i)$ for some $j_0 < j_1 < \cdots < j_i$. We choose $n_{j_{i+1}}$  as the {\em last} occurrence in $\sigma_f$ of an atom that has $\hat{p}$ as party  and is a successor of $n_{j_i}$, meaning that $\{n_{j_{i+1}}\} = \trans(n_{j_i}, \hat{p}, \r_i)$ for some result $\r_i$. This guarantees that the only atom of $\pi$ that 
belongs to $N_\sigma$ is $n_{j_0}$, and that all atoms of $\pi$ are distinct. 
Since all agents \ph{participate in $n_f$} 
the path 
ends with $n_f$. 

In the rest of the proof we rename $n_{j_0}$ as $n_\pi$.
Since $\vx$ enables the loop $\sigma$ and since $n_\pi \in N_\sigma$, after some prefix of $\sigma$ a marking 
$\vx_\pi$ is reached which enables $n_\pi$. The loop
$\sigma$ continues with some outcome $(n_\pi,\r_1)$, where $\r_1$ is one possible result of $n_\pi$. 
By construction of $\pi$, there is another result $\r_2$ of $n_\pi$ such that $\trans(n_\pi,\hat{p},\r_2)$ is 
the second atom of $\pi$, and this atom does not belong to $N_\sigma$. 

Let $\vx_\pi'$ be the marking such that $\vx_\pi\by{(n_\pi,\r_2)}\vx_\pi'$.
We iteratively construct an occurrence sequence enabled at $\vx_\pi'$ as follows. Let $\tau$ be the sequence 
constructed so far and let $\vx$ be the marking reached by $\tau$ (initially $\tau = \epsilon$ and $\vx = \vx_\pi'$):
\begin{itemize}
\item[(0)] If $\vx$ is the final marking, stop.
\item[(1)] Else, if $\vx$ enables an atom $n$ of $N_\sigma$, then at least one outcome $(n,\r)$ occurs in $\sigma$. 
Let $\tau:= \tau \, (n, \r)$.
\item[(2)] Else, if $\vx$ enables an atom $n$ of $\pi$, then let $\r$ be the result such that $\trans (n,\hat{p},\r)$ is the successor of $n$ in $\pi$.
Let $\tau:= \tau \, (n, \r)$.
\item[(3)] Else, let $\rho$ be a shortest occurrence sequence that either leads to the final marking 
or enables an atom that appears in $\sigma$, or  
enables an atom that appears in $\pi$ (so that after this sequence either (1) or (2) can be applied). 
Such an occurrence sequence exists because $\N$ is sound. Further, $\rho$ is not empty because otherwise we would have taken branch (0). Let $\tau:= \tau \rho$. 
\end{itemize} 

For the rest of the proof let $\tau$ be the sequence generated by this procedure.
We claim that $\tau$ is finite 
(i.e., that the procedure terminates) and leads to the final marking. For this we prove
that the procedure takes branches (1)-(3) only finitely often.

We first prove that the procedure takes branch (2) only finitely often. Since every time (2) is taken $\tau$ is extended with an outcome of an atom of $\pi$, it suffices to show that these atoms occur only finitely often in $\tau$. Recall that $\pi$ is a finite path ending with $n_f$, and that all atoms of $\pi$ involve $\hat{p}$. Let $\pi = (n_1, \hat{a}, \r_1) \, (n_2, \hat{a}, \r_2) \, \ldots \, (n_k, \hat{p}, \r_k)$, where 
$n_k = n_f$. By determinism, $\hat{p}$ is always ready to engage in at most one atom of $\pi$.
By construction, after the $i$-th occurrence in $\tau$ of an atom of $\pi$  the agent $\hat{p}$ is ready to engage in $n_{i+1}$. Therefore, the atoms of $\pi$ occur at most $k$ times in $\tau$

We now prove that the procedure takes branch (3) only finitely often. Since branch (2) is taken only finitely often, 
some suffix $\tau'$ of $\tau$  does not contain any occurrence of atoms of $\pi$. For every marking $\vx$ reached 
along the execution of $\tau'$, let $A_\sigma(\vx)$ be the set of agents $p$ such that $\vx(p) \in N_\sigma$
(that is, the agents that at $\vx$ are ready to engage in an atom of $N_\sigma$). We show that
along the execution of $\tau'$ these sets never decrease, and strictly increase each time the 
procedure takes branch (3); since the set of agents is finite, this concludes the proof.

Let $\vx_1$ be a marking of $\tau'$ at which the procedure chooses either branch (1) or branch (3).
If the procedure takes branch (1), then, by the definition of (1), the procedure selects a result $(n, \r)$ that 
occurs in $\sigma$, and extends the current sequence with the step $\vx_1 \by{(n,\r)} \vx_2$.
By Lemma \ref{lem:basicloop} we have $A_\sigma(\vx_1) \subseteq A_\sigma(\vx_2)$. If
the procedure takes branch (3), then the current sequence is extended with a shortest sequence 
$\vx_1 \by{(n_1, \r_1)} \vx_2 \by{(n_2, \r_2)}  \cdots \by{(n_{k-1}, \r_{k-1})} \vx_k$ 
such that $\{n_1, \ldots, n_k\} \cap N_\sigma = \emptyset$, and $\vx_k$ enables an atom of $N_\sigma$ or 
an atom of $\pi$. By determinism, and since $\{n_1, \ldots, n_k\} \cap N_\sigma = \emptyset$, we have
$A_\sigma(\vx_1) \subseteq A_\sigma(\vx_2) \subseteq \cdots \subseteq A_\sigma(\vx_k)$. Since 
$\vx_k$ enables an atom of $N_\sigma$ or an atom of $\pi$, but no atom of $\pi$ occurs in $\tau'$,
 the marking $\vx_k$ enables an atom of $N_\sigma$ and, since the sequence is a shortest one, 
$\vx_{k-1}$ does not enable any atom of $N_\sigma$. So there is an agent $p \in A_\sigma$ such that 
$\vx_1(p) \notin N_\sigma$ and $\vx_k(p) \in N_\sigma$, which proves $A_\sigma(\vx_1) \subset A_\sigma(\vx_k)$,
and we are done.

Finally, we prove that the procedure takes branch (1) only finitely often. Since branches (2) and 
(3) are taken only finitely often, from some point on the algorithm only takes branch (1) (if at all),
and so some suffix $\tau''$ of $\tau$ contains only outcomes of $A_\sigma$. Let $\hat{p}$ 
be the agent of $A_\sigma$ we used for the construction of $\pi$. Since all the atoms of $\pi$
have already been executed before reaching $\tau''$, no atom of $N_\sigma$ in which 
$\hat{p}$ participates, and in particular the atom $n_\pi$, can occur in $\tau''$. 
So all the atoms occurring in $\tau''$ belong to $N_\sigma \setminus \{n_\pi\}$.
Assume $\tau''$ is infinite. Then, since the number of reachable markings is finite, $N_\sigma \setminus \{n_\pi\}$
contains a loop (more precisely, there is a loop in which only atoms of $N_\sigma \setminus \{n_\pi\}$
occur). But this contradicts the minimality of $\sigma$. So $\tau''$ is finite, which concludes the proof
of the claim.

By the claim, the procedure constructs a sequence $\tau$ reaching the final marking.
Since the final atom involves all agents, no agent was able to remain in the loop.
In other words: all agents of $A_\sigma$ left the loop when $(n_\pi, \r_2)$ has occurred.
As a consequence, all these agents are parties of $n_\pi$, and so $n_\pi$ is a synchronizer of the 
loop $\sigma$.
\end{proof}




\subsection{Dominating Atoms}
\label{subsubsec:dom}

Loosely speaking, an atom $n$ of a path of a negotiation dominates the path if every agent that participates in some atom of the path also participates in $n$.

\begin{definition}
Let $\N$ be a negotiation and let $\pi = (n_1, p_1, \r_1) \, (n_2, p_2,  \r_2) \cdots (n_k, p_k, \r_k)$
be a path of $\N$. An atom $n_i$ {\em dominates} $\pi$ if 
$P_{n_j} \subseteq P_{n_i}$ for  $1 \leq j \leq k$. 
\end{definition}

We prove that every circuit of a sound deterministic negotiation has a dominating atom.
This result is a syntactic counterpart \ph{to} Proposition \ref{prop:minSynchronized}, stating that every loop has a synchronizer. 
Indeed, a loop can be seen as a circuit in the marking graph of the negotiation. The proof shows that every circuit can be ``executed'', meaning that one can find an arbitrarily long occurrence sequence that executes the outcomes of the circuit arbitrarily often, and never executes any other outcome of the atoms in the circuit. 

\begin{definition}
Let $\pi = (n_1, p_1, \r_1) \, (n_2, p_2,  \r_2) \cdots (n_k, p_k, \r_k)$ be a path of a negotiation. An {\em execution} of $\pi$ is an occurrence sequence $\vx \by{\sigma} \vx'$
such that $\vx$ is a reachable marking, $\sigma = \sigma_0 (n_1, \r_1) \, \sigma_1 \, (n_2, \r_2) \, \sigma_2  \cdots \sigma_{k-1} \, (n_k, \r_k)$, and for every atom $n$ of $\pi$, if $(n, \r)$ occurs in $\sigma$ then $(n, p, \r)$ occurs in $\pi$.

A path $\pi$ of $\N$ is {\em executable} if \ph{it has} an execution $\vx \by{\sigma} \vx'$. 
\end{definition}

\begin{lemma}
\label{lem:pathexec}
Every path of a sound deterministic negotiation is executable. 
\end{lemma}
\begin{proof}
Let $\N$ be a sound deterministic negotiation, and let 
$\pi = (n_1, p_1, \r_1)  \cdots (n_k, p_k, \r_k)$ be a path of $\N$.

We construct a sequence $\sigma$ that executes $\pi$,
by induction on $k$. If $k=1$ then, by soundness, some reachable marking enables $n_1$, and 
so we can take $\sigma = (n_1, \r_1)$.
If $k > 1$ then, by induction hypothesis, there exists a reachable marking $\vx$
and an occurrence sequence
$\sigma'$ that executes $\pi'= (n_1, p_1, \r_1)  \cdots (n_{k-1}, p_{k-1}, \r_{k-1})$. 
Let $\vx \by{\sigma'} \vx'$.
Since $\N$ is sound, there exists an occurrence sequence $\vx' \by{\tau} \vx_f$. 
Since $\N$ is deterministic and $n_{k} \in \trans(n_i, a_{k-1}, \r_{k-1})$, we have 
$\trans(n_{k-1}, p_{k-1}, \r_{k-1}) = \{n_{k}\}$, and so $\vx'(p_{k-1}) = \{n_k\}$.
Since $\vx_f(p_{k-1}) = \emptyset$, we have $\tau = \tau' \, (n, \r) \, \tau''$ for some
sequence $\tau'$ that  \ph{contains} any atom having $p_{k-1}$ as party  and 
some atom $n$ such that $p_{k-1} \in P_n$. Let $\vx' \by{\sigma'} \vx''$. 
We have $\vx''(p_{k-1}) = \vx'(p_{k-1}) = \{n_1\}$, and so $n = n_k$. So we can take
$\sigma = \sigma' \, \tau' \, (n_k, \r_k)$.
\end{proof}

\begin{proposition}
\label {prop:dominating}
Every cycle of a sound deterministic negotiation has a dominating atom.
\end{proposition}
\begin{proof}
Let $\N$ be a sound deterministic negotiation and 
let $\pi = (n_1, p_1, \r_1)  \cdots (n_k, p_k, \r_k)$ be a cycle of $\N$. Then 
$\pi^i$ (the result of concatenating $i$ copies of $\pi$) is a path of $\N$ for every
$i \geq 1$. By Lemma \ref{lem:pathexec}, $\pi^i$ is executable for every $i > 1$, and 
so for every number $\ell$ there is a firing sequence
containing at least $\ell$ occurrences of each of the outcomes $(n_1, \r_1), \ldots, (n_k, \r_k)$,
and no occurrence of any other outcome of the atoms $n_1, \ldots, n_k$. 
Since $\N$ \ph{has only}  finitely many
reachable markings, we can extract from the firing sequence for a sufficiently large
$\ell$ a loop $\vx \by{\sigma} \vx$ of $\N$ containing all of $(n_1, \r_1), \ldots, (n_k, \r_k)$,
and no other outcome of $n_1, \ldots, n_k$. By Proposition \ref{prop:minSynchronized},
$\sigma$ has a synchronizer $n$. 

We claim that $n= n_i$ for some $1 \leq i \leq k$, which proves
the proposition. Assume the contrary, and let $\vy$ be a marking of the loop that enables $n$. 
Since $n$ is a synchronizer, we have $\vy(p) = \{n\}$ for every party $p$ of 
$n_1, \ldots, n_k$. Let $\vx \by{(n_i, \r_i)} \vx' \by{\tau} \vy$ be the unique segment of the loop
such that no outcome of $\pi$ occurs in $\tau$. 
Since $(n_i, \r_i)$ is the only outcome of $n_i$ that appears in $\sigma$,
no outcome of $n_i$ appears in $\tau$. Since $\N$ is deterministic,  
we have $\vx'(p)= \vy(p)$ for every party $p$ of $n_i$. But, by the definition of the cycle $\pi$, we have $\vx'(p_i)=\{n_{i+ 1}\}$. So we get $\{n_{i+1}\} = \vx'(p_i) = \vy(p_i) = \{n\}$,
which implies $n = n_{i+1}$.
\end{proof} 

\subsection{Proof of Proposition \ref{prop:unique}}
\label{subsec:unique}

We first prove a lemma.

\begin{lemma}
\label{lem:same}
Let $\N$ be a sound deterministic negotiation with a set of agents $A$. Let $B, C$ be a 
partition of $A$, i.e., $A = B \cup C$ and $B \cap C = \emptyset$. 
Let $\vx_1, \vx_2$ be two reachable markings such that
\begin{itemize}
\item $\vx_1(B) = \vx_2(B)$, and 
\item every atom enabled at $\vx_1$ or $\vx_2$ has at least one party in $B$, 
and at least one party in $C$.
\end{itemize}
Then $\vx_1(C) = \vx_2(C)$.
\end{lemma}
\begin{proof}
The proof is by induction on the size of $C$. If $|C|= 0$ then there 
is nothing to show. So assume $|C| > 0$. For $i=1,2$ let $N_i$ be the 
set of atoms enabled at $\vx_i$, and let $B_i$ be the set of agents participating 
in some atom of $N_i$. 

\medskip

\noindent \textbf{Claim 1}: $N_1 \cap N_2 \neq  \emptyset$. \\
We show that if $N_1 \cap N_2 =  \emptyset$ then $\N$ contains a cycle without a dominating
atom, contradicting Proposition \ref{prop:dominating}. Let $n_1 \in N_1$. By soundness, 
there is a shortest sequence $\vx_2 \by{\sigma} \vx_2'$ such that $\vx_2'$ enables $n_1$. 
We prove two subclaims.

\medskip

\noindent {\bf Claim 1a:}  No atom of $\sigma$ dominates $n_1$.\\
Since $n_1$ is enabled at $\vx_1$ and at least one of its parties belongs to $B$, 
we have $\vx_1(b) = \{n_1\}$ for some agent $b \in B$,
and so, by the definition of $B$, also $\vx_2(b) = \{n_1\}$. 
Since $\sigma$ is shortest, it does not contain any occurrence of $n_1$. 
So, since $\N$ is deterministic, we also have $\vy(b) = \{n_1\}$ for every 
intermediate marking $\vy$ reached during the execution of $\sigma$. 
It follows that $b$ does not participate in any atom occurring in $\sigma$.
So no atom of $\sigma$ dominates $n_1$.

\medskip

\noindent {\bf Claim 1b:} There exists $m_0 \in N_2$ and a path 
$\pi = (m_0, p_0, \r_0) \ldots (m_k, p_k, \r_k)$ such that 
$\trans(m_k, p_k, \r_k) = \{n_1\}$ and none of $m_0, \ldots, m_k$
dominates $n_1$. \\
By Claim 1, it suffices to construct a path such that $m_0 \in N_2$, $\trans(m_k, p_k, \r_k) = \{n_1\}$,
and all of $m_0, \ldots, m_k$ occur in $\sigma$. Let $\sigma = (n, \r) \sigma'$. We proceed by induction on the length of $\sigma$. \\
If $|\sigma| = 1$ then $\sigma = (n, \r)$, and so $\vx_2 \by{(n,\r)} \vx_2'$. Since $\vx_2$ does not enable $n_1$ but $\vx_2'$ does, we have  $\trans(n,p,\r) = \{n_1\}$
for some $p \in P_n$. Choose $\pi = (n, p, \r)$. Since $n$ is enabled at $\vx_2$, 
we have $n \in N_2$, and we are done. \\
If $|\sigma| > 1$, let $\vy$ be the marking given by
$\vx_2 \by{(n, \r)} \vy \by{\sigma'} \vx_2'$. Then $\sigma'$ 
is a shortest sequence enabling $n_1$ from $\vy$ and so, by induction 
hypothesis, there is a path $\pi' = (m_1, p_1, \r_1) \cdots (m_k, p_k, \r_k)$
such that $m_1$ is enabled at $\vy$, $\trans(m_k, p_k, \r_k) = \{n_1\}$ and all of 
$m_1, \ldots, m_k$ occur in $\sigma'$. If $m_1 \in N_2$ then we can take $\pi := \pi'$. 
If $m_1 \notin N_2$, then $m_1$ is not enabled at $\vx_2$. Since $m_1$ is enabled at $\vy$,
 we have $\vx_2(p) = \{n\}$ and $\vy(p) = \{m_1\}$ for some party $p$ of both $n$ and $m_0$.
It follows $\trans{(n, p, \r)} = \{m_1\}$. So we can take  
$\pi = (n, p, \r) \pi'$, which concludes the proof.

\medskip

Observe that Claim 1b holds for every atom of $N_1$. By symmetry, for 
every $n_2 \in N_2$ there is $m_0 \in N_1$ and a path 
$\pi = (m_0, p_0, \r_0) \ldots (m_k, p_k, \r_k)$ such that 
$\trans(m_k, p_k, \r_k) = \{n_2\}$ and none of $m_0, \ldots, m_k$
dominates $n_2$. Since $\N$ \ph{has only} finitely many atoms,
there are $n_{11}, \ldots, n_{1k} \in N_1$, $n_{21}, \ldots, n_{2k} \in N_2$ 
and a cycle that visits the sequence of atoms
$n_{11}, n_{21}, n_{12}, n_{22}, \ldots, n_{1k}, n_{2k}$ in that
order, and such that no atom of the cycle dominates all others. 
So the cycle does not contain a dominating atom, contradicting 
Proposition \ref{prop:dominating}. This proves Claim 1. 

\medskip

By Claim 1, there is an atom $n \in N_1 \cap N_2$. Let $c \in C \cap P_n$, which exists 
by assumption, and let $B' = B \cup \{c\}$ and $C' = C \setminus \{c\}$.
Since $n$ is enabled at $\vx_1$ and $\vx_2$ we have $\vx_1(c) = \{n\} = \vx_2(c)$.
Since $|C'| = |C|-1$ we can apply the induction hypothesis to $B'$ and $C'$. So 
we have $\vx_1(C') = \vx_2(C')$, and therefore $\vx_1(C) = \vx_2(C)$.
\end{proof}

We are now ready to prove that all maximal $n$-sequences have the same target,
and that the same holds for maximal strict $(n, \r)$-sequences.

\begin{refproposition}{prop:unique}
Let $\N$ be a sound and deterministic negotiation, and let $n$ be an atom of $\N$.
\begin{itemize}
\item[(a)] All maximal $n$-sequences have the same target.\\
That is: there is a unique marking $\vx$ such that $\vx_n \by{\sigma} \vx$ for every 
maximal $n$-sequence $\sigma$. We call $\vx$ the {\em target} of $n$.
\item[(b)] For every outcome $(n, \r)$, all maximal strict $(n, \r)$-sequences have the same target. \\
That is: there is a unique marking $\vx$ such that $\vx_n \by{\sigma} \vx$ for every 
maximal strict $(n, \r)$-sequence $\sigma$. We call $\vx$ the {\em target} of $(n,\r)$.
\end{itemize}
\end{refproposition}
\begin{proof}
\noindent (a)  Let $\sigma_1$ and $\sigma_2$ be two maximal $P_n$-sequences and 
$\vx_0 \by{\sigma} \vy$ be an occurrence sequence that enables 
$n$ for the first time. We then have
$$\vx_0 \by{\sigma} \vy \by{\sigma_1} \vx_1 \qquad \mbox{ and  } \qquad \vx_0 \by{\sigma} \vy \by{\sigma_2} \vx_2$$
\noindent for two markings $\vx_1, \vx_2$ such that $\vx_1(a) = \vx_2(a)$ for every $a \notin P_n$.
By the maximality of $\sigma_1$ and $\sigma_2$, neither $\vx_1$ nor $\vx_2$ enable 
any atom $n'$ such that $P_{n'} \subseteq P_n$

We prove $\vx_1 = \vx_2$, which shows that $\sigma_1$ and $\sigma_2$ have the same target.
Let $B, C$ be the partition of the agents of $\N$ defined by
$p \in B$ if{}f $\vx_1(p) = \vx_2(p)$. We have $C \subseteq P_n$.
By soundness there exists a firing sequence $\vx_1 \by{\rho} \vx_f$. 
Let $\tau$ be the longest $B$-prefix of $\rho$, and let $\vx_1 \by{\tau} \vx_1'$. 
Since $\vx_1(B) = \vx_2(B)$ (meaning $\vx_1(p) = \vx_2(p)$ for every $p \in B$), 
we have $\vx_2 \by{\tau} \vx_2'$ for some marking $\vx_2'$ such that
$\vx_1'(B) = \vx_2'(B)$. 

Let $n'$ be an atom enabled at $\vx_1'$. We claim that $P_{n'}$ intersects both 
$B$ and $C$. Since $B, C$ is a partition, it suffices to show that $P_{n'}$ is not
contained in $B$ and is not contained in $C$. That $P_{n'}$ is not contained in $B$ follows from the maximality of $\tau$. Assume $P_{n'} \subseteq C$. Since $C \subseteq P_n$ we get $P_{n'} \subseteq P_n$. Since $\tau$ does not contain any atom with a party in $C$, we have $\vx_1'(C) = \vx_1(C)$, 
and so $n'$ is also enabled at $\vx_1$, contradicting the maximality of $\sigma_1$.
This proves the claim. 

So every atom enabled at $\vx_1'$ has parties in both $B$ and $C$. By symmetry, the same holds
for $\vx_2'$. Since $\vx_1'(B) = \vx_2'(C)$ we can apply Lemma \ref{lem:same} and conclude 
$\vx_1' = \vx_2'$. Since $\vx_1'$ and $\vx_2'$ are reached from $\vx_1$ and $\vx_2$ by means of the same 
sequence $\tau$, we get $\vx_1 = \vx_2$.

\noindent (b) The proof is exactly as in (a). In this case we even have $C \subset P_n$
\end{proof}

\end{document}